\documentclass[ejs]{imsart}

\usepackage{amsthm,amsmath,amsfonts,amssymb,amsbsy,dsfont,stackengine,tikz,graphicx,color,bbold}
\RequirePackage[numbers]{natbib}
\usepackage{hyperref}
\hypersetup{colorlinks,citecolor=blue,urlcolor=blue}

\startlocaldefs

\newcommand{\y}{\mathbf{y}}

\newcommand{\x}{\mathbf{x}}
\newcommand{\X}{\mathbf{X}}
\newcommand{\W}{\mathbf{W}}
\newcommand{\U}{\mathbf{U}}

\def \A {\mathbf{A}}

\def \B {\mathbf{B}}

\def \C {\mathbf{C}}
\def \D {\mathbf{D}}
\def \d {\mathbf{d}}

\def \e {\mathbf{e}}

\def \G {\mathbf{G}}
\def \g {\mathbf{g}}
\def \h {\mathbf{h}}

\def \I {\mathbf{I}}

\def \N {\mathbf{N}}

\def \Q {\mathbf{Q}}

\def \R {\mathbf{R}}
\def \s {\mathbf{s}}
\def \S {\mathbf{S}}

\def \U {\mathbf{U}}
\def \u {\mathbf{u}}

\def \W {\mathbf{W}}
\def \w {\mathbf{w}}
\def \x {\mathbf{x}}
\def \X {\mathbf{X}}

\def \y {\mathbf{y}}
\def \Z {\mathbf{Z}}
\def \z {\mathbf{z}}

\def \Dcal {\mathcal{D}}

\def \Fcal {\mathcal{F}}

\def \Hcal {\mathcal{H}}

\def \Ncal {\mathcal{N}}
\def \Ocal {\mathcal{O}}

\def \Vcal {\mathcal{V}}

\def \Cbb {\mathbb{C}}
\def \Ebb {\mathbb{E}}

\def \Nbb {\mathbb{N}}
\def \Pbb {\mathbb{P}}
\def \Rbb {\mathbb{R}}

\def \Zbb {\mathbb{Z}}

\def \epsilonbs {\boldsymbol{\epsilon}}
\def \omegabs {\boldsymbol{\omega}}

\def \xibs {\boldsymbol{\xi}}

\def \Deltabs {\boldsymbol{\Delta}}
\def \Gammabs {\boldsymbol{\Gamma}}
\def \Lambdabs {\boldsymbol{\Lambda}}
\def \Omegabs {\boldsymbol{\Omega}}

\def \Psibs {\boldsymbol{\Psi}}
\def \Sigmabs {\boldsymbol{\Sigma}}
\def \Thetabs {\boldsymbol{\Theta}}
\def \Xibs {\boldsymbol{\Xi}}
\def \Upsilonbs {\boldsymbol{\Upsilon}}
\def \Phibs {\boldsymbol{\Phi}}

\def \det {\mathrm{det\ }}

\def \Tr {\mathrm{tr}\ }
\def \Prob {\mathbb{P}}

\def \diag {\mathrm{diag}}

\newtheorem{theorem}{Theorem}[section]
\newtheorem{corollary}{Corollary}[section]
\newtheorem{lemma}{Lemma}[section]
\newtheorem{proposition}{Proposition}[section]
\newtheorem{assumption}{Assumption}[section]
\newtheorem{remark}{Remark}[section]
\newtheorem{definition}{Definition}[section]

\numberwithin{equation}{section}

\DeclareMathOperator{\Supp}{Supp}
\DeclareMathOperator{\Var}{Var}
\DeclareMathOperator*{\argmin}{argmin}

\renewcommand{\Im}{\mathrm{Im}}
\renewcommand{\Re}{\mathrm{Re}}
\newcommand\barbelow[1]{\stackunder[1.2pt]{$#1$}{\rule{.8ex}{.075ex}}}
\newcommand*\diff{\mathop{}\!\mathrm{d}}

\endlocaldefs

\begin{document}

\begin{frontmatter}

\title{Properties of linear spectral statistics of frequency-smoothed estimated spectral coherence matrix of high-dimensional Gaussian time series}

\runtitle{Coherency matrices of high-dimensional time series}


\author{\fnms{Philippe} \snm{Loubaton}\thanksref{t1}\ead[label=e2]{philippe.loubaton@univ-eiffel.fr}}
\and
\author{\fnms{Alexis} \snm{Rosuel}\thanksref{t1}\ead[label=e1]{alexis.rosuel@univ-eiffel.fr}}
\address{Laboratoire d’Informatique Gaspard Monge, UMR 8049, \\
Université Paris-Est Marne la Vallée, France\\ 
\printead{e1,e2}}

\thankstext{t1}{This work is funded by ANR Project HIDITSA, reference ANR-17-CE40-0003.} 

\runauthor{P. Loubaton and A. Rosuel}
\begin{abstract}
The asymptotic behaviour of Linear Spectral Statistics (LSS) of the smoothed periodogram estimator of the spectral coherency matrix  of a complex Gaussian high-dimensional time series $(\y_n)_{n \in \mathbb{Z}}$ with independent components is studied under the asymptotic regime where the sample size $N$ converges towards $+\infty$ while the dimension $M$ of $\y$ and the smoothing span of the estimator grow to infinity at the same rate in such a way that $\frac{M}{N} \rightarrow 0$. It is established that, at each frequency, the estimated spectral coherency matrix is close from the sample covariance matrix of an independent identically $\mathcal{N}_{\mathbb{C}}(0,\I_M)$ distributed sequence, and that its empirical eigenvalue distribution converges towards the Marcenko-Pastur distribution. This allows to conclude that each LSS has a deterministic behaviour that can be evaluated explicitly. Using concentration inequalities, it is shown that the order of magnitude of the supremum over the frequencies of the deviation of each LSS from its deterministic approximation is of the order of $\frac{1}{M} + \frac{\sqrt{M}}{N}+ (\frac{M}{N})^{3}$ where $N$ is the sample size. Numerical simulations supports our results. 
\end{abstract}




\tableofcontents

\end{frontmatter}

\sloppy

\section{Introduction}
\subsection{The addressed problem and the results}
\label{subsec:results}
We consider an $M$--variate zero-mean complex Gaussian stationary time series \footnote{any finite linear combination $x$ of the components of $({\bf y}_n)_{n \in \mathbb{Z}}$ is a complex Gaussian random variable, i.e. $\mathrm{Re}(x)$ and $\mathrm{Im}(x)$ are independent zero-mean Gaussian random variables having the same variance} $(\y_n)_{n \in \mathbb{Z}}$ and assume that the samples $\y_1, \ldots, \y_N$ are available. We introduce the traditional frequency smoothed periodogram estimate $\hat{\S}(\nu)$ of the spectral density of $\y$ at frequency $\nu$ defined by
\begin{equation}
\label{eq:def-spectral-estimate}
\hat{\S}(\nu) = \frac{1}{B+1} \sum_{b=-B/2}^{B/2} \xibs_\y\left(\nu + \frac{b}{N}\right) \, \xibs_\y\left(\nu + \frac{b}{N}\right)^{*}
\end{equation}
where $B$ is an even integer, which represents the smoothing span, and
\begin{equation}
    \label{eq:def-fourier-transform-xi}
    \xibs_\y(\nu) = \frac{1}{\sqrt{N}} \sum_{n=1}^{N} \y_n e^{-2i\pi(n-1)\nu}
\end{equation}
is the renormalized Fourier transform of $(\y_n)_{n=1\ldots,N}$. The corresponding estimated spectral coherency matrix is defined as:
\begin{equation}
\label{equation:definition_coherency}
\hat{\C}(\nu) = \diag\left(\hat{\S}(\nu)\right)^{-\frac{1}{2}} \hat{\S}(\nu) \diag\left(\hat{\S}(\nu)\right)^{-\frac{1}{2}}
\end{equation}
where $\diag(\hat{\S}(\nu))=\hat{\S}(\nu)\odot \I_M$, with $\odot$ denoting the Hadamard product (ie. entrywise product) and $\I_M$ is the $M$--dimensional identity matrix. Under the hypothesis $\Hcal_0$ that the $M$ components  $(y_{1,n})_{n \in \mathbb{Z}}, \ldots,  (y_{M,n})_{n \in \mathbb{Z}}$ of $\y$ are mutually uncorrelated, we evaluate the 
behaviour of certain Linear Spectral Statistics (LSS) of the eigenvalues of $\hat{\C}(\nu) $ in 
asymptotic regimes where $N \rightarrow +\infty$ and both $M = M(N)$ and $B = B(N)$ 
converge towards $+\infty$ in such a way that $ M(N)= \Ocal(N^{\alpha})$ 
for $\alpha \in (1/2,1)$ and $c_N = \frac{M(N)}{B(N)} \rightarrow c$ where $c \in (0,1)$. 
We denote by $\mu_{MP}^{(c)}$ the Marcenko-Pastur distribution 
with parameter $c < 1$ defined by 
\[
    d\mu_{MP}^{(c)}(\lambda)=  \frac{\sqrt{(\lambda_+-\lambda)(\lambda-\lambda_-)}}{2\pi c\lambda}\mathds{1}_{\lambda\in[\lambda_-;\lambda_+]}(\lambda)\diff\lambda, \quad \lambda_\pm=(1\pm\sqrt{c})^2
\]
and define the sequences $(u_N)_{N \geq 1}$ and $(v_N)_{N \geq 1}$ by 
\begin{equation}
    \label{eq:def-uN}
    u_N = \frac{1}{B} +  \frac{\sqrt{B}}{N} +  \left(\frac{B}{N}\right)^{3}
\end{equation}
and 
\begin{equation}
    \label{eq:def-vN}
    v_N = \frac{1}{B+1} \sum_{b=-B/2}^{B/2} \left( \frac{b}{N} \right)^{2}.
\end{equation}
We notice that 
\begin{equation}
    \label{eq:carac-uN}
    u_N = \Ocal\left( \frac{1}{B}\right) \mathbf{1}_{\frac{1}{2} \leq \alpha \leq \frac{2}{3}} + \Ocal\left( \frac{\sqrt{B}}{N}\right) \mathbf{1}_{\frac{2}{3} \leq \alpha \leq \frac{4}{5}} + \Ocal\left(\frac{B}{N}\right)^{3} \mathbf{1}_{\alpha \geq \frac{4}{5}}
\end{equation}
and $v_N = \Ocal\left( (\frac{B}{N})^{2} \right)$, as well as $\frac{u_N}{v_N} \rightarrow 0$ if $\alpha > 2/3$ 
and $\frac{u_N}{v_N} \rightarrow +\infty$ if $\alpha < 2/3$. 
Then, if $(s_m)_{m=1, \ldots, M}$ represent  the spectral densities of the scalar time series $((y_{m,n})_{n \in \mathbb{Z}})_{m=1, \ldots, M}$, for each function $f$ defined on $\mathbb{R}^{+}$ and $\mathcal{C}^{\infty}$ in a neighbourhood of the support $[\lambda_-;\lambda_+]$ of $\mu_{MP}^{(c)}$, it holds that for each
$\epsilon > 0$, there exists a $\gamma(\epsilon):=\gamma>0$ such that for each $N$ large enough:
\begin{multline}
    \label{eq:main-result}
    \Prob\left[\sup_{\nu \in [0,1]} \left|\frac{1}{M} \mathrm{Tr}\left( f(\hat{\C}(\nu))\right) -  \int_{\Rbb^{+}}f\diff\mu_{MP}^{(c_N)}  -  r_N(\nu)  \; \phi_N(f) \; v_N \; \mathbf{1}_{\alpha > 2/3} \right| \right. \\ \left. > N^\epsilon u_N\right]  \le \exp-N^{\gamma}
\end{multline}
where $r_N(\nu)$ is defined by 
\begin{equation}
    \label{def-eq-rN}
    r_N(\nu) = \left( \frac{1}{M} \sum_{m=1}^{M} \frac{s_m'(\nu)}{s_m(\nu)} \right)^{2}
\end{equation}
and where $\phi_N(f)$ is a deterministic $\Ocal(1)$ term which coincides with the action of the function $f$ on a certain compactly supported distribution $D_N$ (to be made precised later) depending on the Marcenko-Pastur distribution $\mu_{MP}^{(c_N)}$. In other words, under $\Hcal_0$, uniformly w.r.t. the frequency $\nu$, $\frac{1}{M} \mathrm{Tr}\left( f(\hat{\C}(\nu))\right)$ behaves as $\int_{\Rbb}f\diff\mu_{MP}^{(c_N)}$. If $\alpha \leq 2/3$, with high probability, the order of magnitude of the corresponding error is not larger than  $u_N = \frac{1}{B} = \Ocal(\frac{1}{N^{\alpha}})$. If $\alpha > 2/3$, 
$\frac{1}{M} \mathrm{Tr}\left( f(\hat{\C}(\nu))\right) - \int_{\Rbb}f\diff\mu_{MP}^{(c_N)}$
behaves as the deterministic $\Ocal(\frac{B}{N})^{2}$ term 
$r_N(\nu) \; \phi_N(f) \; v_N$, and the rate of convergence towards $0$ of the corrected statistics $\frac{1}{M} \mathrm{Tr}\left( f(\hat{\C}(\nu))\right) -  \int_{\Rbb^{+}}f\diff\mu_{MP}^{(c_N)}  - r_N(\nu) \; \phi_N(f) \; v_N \; \mathbf{1}_{\alpha > 2/3}$
appears to be $u_N$ which satisfies $\frac{u_N}{v_N} \rightarrow 0$. \\

Our approach is based on the observation that in the above asymptotic regime, $\hat{\S}(\nu)$ can be interpreted as the sample covariance matrix of the large vectors 
$( \xibs_\y(\nu + \frac{b}{N}))_{b=-B/2, \ldots, B/2}$. Classical time series analysis results 
suggest that the vectors 
$( \xibs_\y(\nu + \frac{b}{N}))_{b=-B/2, \ldots, B/2}$ appear as "nearly" i.i.d. zero mean complex random vectors with covariance matrix $\S(\nu)$ where $\S(\nu) = \diag\left(s_1(\nu), \ldots, s_M(\nu)\right)$. $\hat{\C}(\nu)$ can be interpreted as the sample autocorrelation matrix of the above vectors. As it is well-known that the empirical eigenvalue distribution of the sample autocorrelation matrix of i.i.d. large random vectors converges towards the Marcenko-Pastur distribution 
(see e.g. \cite{jiang2004}), it is not surprising that
$\frac{1}{M} \mathrm{Tr}\left( f(\hat{\C}(\nu))\right)$ behaves as  $\int_{\Rbb^{+}}f\diff\mu_{MP}^{(c_N)}$. Our main results are thus obtained using tools borrowed from large random matrix theory 
(see e.g. \cite{pasturshcherbina2011}, \cite{bai2010spectral}) and from frequency domain time series analysis techniques (see e.g. \cite{brillinger1981time}).

\subsection{Motivation}
This paper is motivated by the problem of testing whether the components of $\y$ 
are uncorrelated or not when the dimension $M$ of $\y$ is large and the number of observations $N$ is significantly larger than $M$. For this, a possible way
would be to estimate the spectral coherency matrix, equal to $\I_M$ at each frequency $\nu$ under $\Hcal_0$, by the standard estimate $\hat{\C}(\nu)$ defined by (\ref{equation:definition_coherency}) for a relevant choice of $B$, and to compare, for example, the supremum over $\nu$ of the spectral norm $\|\hat{\C}(\nu) - \I_{M}\|$ to a threshold. To understand the conditions under which such an approach should provide satisfying results, we mention that under some mild extra assumptions, it can be shown that
$$
\sup_{\nu} \| \hat{\S}(\nu) - \S(\nu) \| \xrightarrow[N\to+\infty]{a.s.} 0
$$
as well as 
$$
\sup_{\nu} \| \hat{\C}(\nu) - \I_M \| \xrightarrow[N\to+\infty]{a.s.} 0
$$
in asymptotic regimes where $N,B,M$ converge towards $+\infty$ in such a way that 
$\frac{B}{N} \rightarrow 0$ and $\frac{M}{B} \rightarrow 0$. Therefore, 
$\hat{\C}(\nu)$ is likely to be close to $\I_M$ for each $\nu$ if both 
$\frac{B}{N}$ and $\frac{M}{B}$ are small enough. However, if $M$ is large and the number of available samples $N$ is not arbitrarily large w.r.t. $M$, it may be impossible
to choose the smoothing span $B$ in such a way that $\frac{B}{N} \ll 1$ and $\frac{M}{B} \ll 1$.
In such a context, the predictions provided by the asymptotic regime 
$\frac{B}{N} \rightarrow 0$ and $\frac{M}{B} \rightarrow 0$ will not be accurate, and  
any test comparing $\hat{\C}(\nu)$ to $\I_M$ for each $\nu$ will provide poor results. 
To solve this issue, we propose to choose $B$ of the same order of
magnitude as $M$. In this case,  $\hat{\C}(\nu)$ has of course no reason to be close to $\I_M$ for each $\nu$. If $\frac{M}{N}$, or equivalently if $\frac{B}{N}$ is small enough, the asymptotic regime where both $M$ and $B$ converge towards $+\infty$ at the same rate appears relevant to understand the behaviour of $\hat{\C}(\nu)$. We mention in particular that the condition $\alpha > 1/2$ 
implies that the rate of convergence of $\frac{M}{N}$ towards $0$ is moderate, which is in accordance with practical situations in which the sample size is not arbitrarily large. 
Our asymptotic results thus suggest that if $\frac{M}{N}$ is small enough and if $B$ is chosen of the same order of magnitude as $M$, then it seems reasonable to test that the components of $\y$ are uncorrelated by comparing 
$$
\frac{1}{u_N} \sup_{\nu \in [0,1]} \left| \frac{1}{M} \mathrm{Tr}\left( f(\hat{\C}(\nu)\right) -  \int_{\Rbb^{+}}f\diff\mu_{MP}^{(c_N)} - \hat{r}_N(\nu) \; \phi_N(f) \; v_N \; \mathbf{1}_{\alpha > 2/3} \right|
$$
to a well 
chosen threshold, where $\hat{r}_N(\nu)$ represents an estimate of $r_N(\nu)$ accurate enough to keep equal to $u_N$ the convergence rate towards $0$ of the modified statistics.
We notice that our results just characterize the order of magnitude of the above statistics under $\mathcal{H}_0$, and that we do not provide asymptotic approximation of its distribution. While the derivation of such an approximation would be quite useful to design a well defined statistical test and to study and compare its performance with existing approaches, our results represent a first necessary step that has its own interest. We notice that we consider the supremum on the whole frequency interval $[0,1]$ because, compared to a solution where the maximum is over a low number of fixed frequencies, this allows to increase the power of the test in contexts of alternatives for which, under $\mathcal{H}_1$, 
\begin{equation}
    \label{eq:function-statistics}
\nu \rightarrow  \left| \frac{1}{M} \mathrm{Tr}\left( f(\hat{\C}(\nu))\right) -  \int_{\Rbb^{+}}f\diff\mu_{MP}^{(c_N)} - \hat{r}_N(\nu) \; \phi_N(f) \; v_N \; \mathbf{1}_{\alpha > 2/3} \right|
\end{equation}
exhibits narrow peaks that would not be visible on a low density frequency grid. We also mention that other statistics could also be considered, e.g. the integral on the frequency domain of the function (\ref{eq:function-statistics}) or of the square of this function. \\

We finally remark that the most usual asymptotic regime considered in the context of large random
matrices is $M \rightarrow +\infty, N \rightarrow +\infty$ in such a way that $\frac{M}{N}$ converges towards a non zero constant. In this regime, it is still possible to develop large random matrix-based approaches testing that the components of $\y$ are uncorrelated or not, see e.g. the contribution \cite{pangaoyang2014jasa} to be presented below which, under the extra assumption that the components of $\y$ share the same spectral density, is based on a Gaussian approximation of linear spectral statistics of the empirical covariance matrix $\hat{\R}_N$ defined by 
\begin{equation}
    \label{eq:def-hatR}
    \hat{\R}_N  = \frac{1}{N} \sum_{n=1}^{N} \y_n \y_n^{*}
\end{equation}
under $\mathcal{H}_0$. However, when the ratio $\frac{M}{N}$ is small enough, the asymptotic regime considered in the present paper seems more relevant than the standard large random matrix regime  $M \rightarrow +\infty, N \rightarrow +\infty$, and test statistics that depend on the estimated spectral coherency matrix $\hat{\C}(\nu)$ should provide better performance than functionals of the matrix $\hat{\R}_N$.

\subsection{On the literature}
The problem of testing whether various jointly stationary and jointly Gaussian time series are uncorrelated is an important problem that was extensively addressed in the past. Apart from a few works that will be discussed later, almost all the previous contributions addressed the case where the number $M$ of available time series remains finite as the sample size increases.  Two classes of methods were mainly studied. The first class uses lag domain approaches based on the observation that $M$ jointly stationary time series $(y_{1,n})_{n \in \mathbb{Z}}, \ldots,  (y_{M,n})_{n \in \mathbb{Z}}$ are mutually uncorrelated if and only if for each integer $L$, the covariance matrix of the $ML$ dimensional 
 vector $\y^{(L)}_n$ defined by 
 $$
 \y^{(L)}_n = (y_{1,n}, \ldots, y_{1,n+L-1}, \ldots, y_{M,n}, \ldots, y_{M,n+L-1})^{T}
$$
is block diagonal. The lag domain 
approach was in particular used in \cite{haugh1976checking} for $M=2$, and extended and developed in 
\cite{koch1986method}, \cite{li1994robust}, \cite{himdi1997tests}, \cite{hong1996testing}, \cite{duchesne2003robust} and \cite{elhimdiduchesneroy2003}. \\

The second approach is based on the observation 
that the $M$ jointly stationary time series $(y_{1,n})_{n \in \mathbb{Z}}, \ldots,  (y_{M,n})_{n \in \mathbb{Z}}$ are uncorrelated if and only the spectral density matrix 
$\S(\nu)$ of $\y_n = (y_{1,n}, \ldots, y_{M,n})^{T}$ is diagonal for each frequency 
$\nu$, or equivalently, if its spectral coherence matrix $\C(\nu)$ is reduced to 
$\I_M$ for each $\nu$. \cite{wahba1971some} is one of the first related contribution. This work was followed by \cite{eichler2007frequency}, \cite{taniguchi1996nonparametric},  as well as \cite{eichler2008testing}. \\

We now review the existing works devoted to the case where the number $M$ of time series converges towards $+\infty$. The particular context where the observations $\y_1, \ldots, \y_N$  are i.i.d. and where the ratio $\frac{M}{N}$ converges towards a constant $d \in (0,1)$ is the most popular. In contrast to the asymptotic regime considered in the present paper, $M$ and $N$ are of the same order of magnitude. This is because, in this context, the time series are mutually uncorrelated if and only the covariance matrix $\Ebb[\y_n \y_n^{*}]$ is diagonal. Therefore, it is reasonable to consider test statistics that are functionals of the sample covariance matrix $\hat{\R}_N$ defined by (\ref{eq:def-hatR}). In particular, when the observations are Gaussian random vectors, the generalized likelihood ratio test (GLRT) consists in comparing the test statistics $\log \mathrm{det}(\hat{\C}_N)$ to a threshold, where $\hat{\C}_N$ represents the sample autocorrelation matrix. 
\cite{jiang2004} proved that under $\Hcal_0$, the empirical eigenvalue distribution of $\hat{\C}_N$ 
converges almost surely towards the Marcenko-Pastur distribution $\mu_{MP}^{(d)}$ and therefore, that 
$\frac{1}{M} \mathrm{Tr}\left(f(\hat{\C}_N)\right)$ converges towards $\int f d\mu_{MP}^{(d)}$ for each bounded continuous function $f$. In the Gaussian case, \cite{jiangyang2013} also established a central limit theorem (CLT) for $\log \mathrm{det}(\hat{\C}_N)$ under  $\Hcal_0$  using the moment method.
In the real Gaussian case, \cite{dette2020likelihood} remarked that $\left( \mathrm{det} \hat{\C}_N \right)^{N/2}$ is the product of independent beta distributed random variables. Therefore, $\log\det(\hat{\C}_N)$ appears as the sum of independent random variables, thus deducing the CLT. More recently, in \cite{mestre2017correlation} is established a CLT on LSS of $\hat{\C}_N$ in the Gaussian case using large random matrix techniques when the covariance matrix $\Ebb[\y_n \y_n^{*}]$ is not necessarily diagonal. This allows studying the asymptotic performance of the GLRT under a certain class of alternatives. We also mention that \cite{jiangmaxcorr2004} studied the behaviour of $\max_{i,j} |(\hat{\C}_{N})_{i,j}|$ under $\Hcal_0$, and established that  $\max_{i,j} |(\hat{\C}_{N})_{i,j}|$, after recentering and appropriate normalization, converges in distribution towards a Gumbel distribution, which, of course, allows to test the hypothesis $\Hcal_0$. This first contribution was extended later in several works, in particular in \cite{chenliu2018} who considered the case where the samples $\y_1, \ldots, \y_N$ have some specific correlation pattern. Still, in the asymptotic regime $\frac{M}{N} \rightarrow d$, \cite{pangaoyang2014jasa} proposed to test hypothesis $\mathcal{H}_0$ when the components of $\y$ share the same spectral density. In this case, the rows of the $M \times N$ matrix $(\y_1, \ldots, \y_N)$ are independent and identically distributed under $\mathcal{H}_0$. \cite{pangaoyang2014jasa} established a central limit theorem for linear spectral statistics of the empirical covariance matrix $\hat{\R}_N$ defined by (\ref{eq:def-hatR}), and used this test statistics to check whether $\mathcal{H}_0$ holds or not. We notice that the results of \cite{pangaoyang2014jasa} are valid in the non-Gaussian case. \\

In our knowledge, no existing work studied the behaviour of linear spectral statistics
of the matrix $\hat{\C}(\nu)$ in the asymptotic regime defined in the present paper. However, we mention that this 
regime was considered in \cite{bohm2009shrinkage} to solve a completely different problem, i.e. the use of shrinkage in the frequency domain in order to enhance the performance of the spectral density estimate (\ref{eq:def-spectral-estimate}) when the components of $\y$ are not uncorrelated. We notice that $\frac{B^{3/2}}{N}$ is supposed to converge towards $0$ in \cite{bohm2009shrinkage}. When $B = \Ocal(N^{\alpha})$, this condition is equivalent to $\alpha < 2/3$, while we rather study situations where $\alpha > 1/2$. We finally mention that our works \cite{loubatonrosuelvallet2021jmva} and 
\cite{rosuelvalletloubatonmestre2021ieeesp} also consider the present asymptotic regime and study 
respectively the behaviour of $\sup_{i<j, \nu \in \mathcal{G}_N} |\hat{\C}_{i,j}(\nu)|$ ($\mathcal{G}_N$
is the set $\{ k \frac{B+1}{N}, k=0, \ldots, \frac{N}{B+1}\}$)
and the largest eigenvalues of $\hat{\C}(\nu)$ in the presence of an extra signal, independent from $\y$, and having a low-rank spectral density matrix.

\subsection{General approach}
To simplify the notations, we denote by $\psi_N(f,\nu)$ the statistics defined by 
\begin{equation}
    \label{eq:def-psi(f)}
    \psi_N(f,\nu) = \frac{1}{M} \mathrm{Tr}\left( f(\hat{\C}(\nu))\right) -  \int_{\Rbb^{+}}f\diff\mu_{MP}^{(c_N)}  -  r_N(\nu) \; \phi_N(f) \; v_N \; \mathbf{1}_{\alpha > 2/3}.
\end{equation}
To study the behaviour of $\sup_{\nu} |\psi_N(f,\nu)|$, we establish exponential concentration inequalities that allow to evaluate $\Pbb(|\psi_N(f,\nu)| > N^{\epsilon} u_N)$ for each $\nu$ as well as $\Pbb(\sup_{\nu \in \Vcal_N}|\psi_N(f,\nu)| > N^{\epsilon} u_N)$
for some relevant finite discrete grid $\Vcal_N$ of the interval $[0,1]$. (\ref{eq:main-result}) is then obtained by using Lipschitz properties of the function $\nu \rightarrow \psi_N(f,\nu)$. \\

To evaluate $\Pbb(|\psi_N(f,\nu)| > N^{\epsilon} u_N)$ for each $\nu$, we use the following approach:
\begin{itemize}
    \item We first study the behaviour of the modified sample spectral coherency matrix 
    $\tilde{\C}(\nu)$ defined by 
    \begin{equation}
        \label{eq:def-tildeC}
     \tilde{\C}(\nu) =  \diag\left(\S(\nu)\right)^{-\frac{1}{2}} \hat{\S}(\nu) \diag\left(\S(\nu)\right)^{-\frac{1}{2}}.
    \end{equation}
    We notice that 
    $\tilde{\C}(\nu)$ is obtained from $\hat{\C}(\nu)$ by replacing the estimated 
    diagonal matrix $\diag\left(\hat{\S}(\nu)\right)$ by its true value 
    $\diag\left(\S(\nu)\right)$. Using classical results of \cite{brillinger1981time}, we establish that for each $\nu$, $\tilde{\C}(\nu)$ 
    can be represented as 
    \begin{equation}
        \label{eq:representation-tildeC-intro}
        \tilde{\C}(\nu) = \frac{\X(\nu)\X^*(\nu)}{B+1}+ \tilde{\Deltabs}(\nu)
    \end{equation}
    where $\X(\nu)$ is an $M\times(B+1)$ random matrix with $\Ncal_\Cbb(0,1)$ i.i.d. entries, and $\tilde{\Deltabs}(\nu)$ is another matrix such that, for any $\epsilon>0$, there exists $\gamma > 0$, independent from $\nu$, such that for each large enough $N\in\Nbb$:
$$ \Prob\left[\|\tilde{\Deltabs}(\nu)\|>N^\epsilon \, \frac{B}{N} \right]\le\exp-N^{\gamma}. $$
We deduce from (\ref{eq:representation-tildeC-intro}) that $\hat{\C}(\nu)$ can be written as
\begin{equation}
        \label{eq:representation-hatC}
        \hat{\C}(\nu) = \frac{\X(\nu)\X^*(\nu)}{B+1}+ \Deltabs(\nu)
    \end{equation}
    where $\Deltabs(\nu)$ satisfies the concentration inequality
$$ \Prob\left[\|\Deltabs(\nu)\|>N^\epsilon \, \left( \frac{1}{\sqrt{B}} + \frac{B}{N} \right) \right]\le\exp-N^{\gamma} $$
for each $\epsilon > 0$, where $\gamma$ does not depend on $\nu$. Using (\ref{eq:representation-tildeC-intro}) and (\ref{eq:representation-hatC}), we establish that the eigenvalues of $\tilde{\C}(\nu)$ and $\hat{\C}(\nu)$ are localized with high probability in a neighbourhood of the support of the Marcenko-Pastur distribution $\mu_{MP}^{(c)}$. 

$\tilde{\C}(\nu)$ appears as a useful intermediate matrix because the study of 
$\frac{1}{M} \mathrm{Tr}\left( f(\hat{\C}(\nu))\right) -  \int_{\Rbb^{+}}f\diff\mu_{MP}^{(c_N)}$
is based on the evaluation of each term of the following decomposition:
\begin{align}
    \nonumber 
    \frac{1}{M} \mathrm{Tr}\left( f(\hat{\C}(\nu))\right) -  \int_{\Rbb^{+}}f\diff\mu_{MP}^{(c_N)} = \frac{1}{M} \mathrm{Tr}\left( f(\hat{\C}(\nu))\right) - \frac{1}{M} \mathrm{Tr}\left( f(\tilde{\C}(\nu))\right) + \\ 
    \nonumber \frac{1}{M} \mathrm{Tr}\left( f(\tilde{\C}(\nu))\right) - \mathbb{E} \left[ \frac{1}{M} \mathrm{Tr}\left( f(\tilde{\C}(\nu))\right) \right] + \\
    \nonumber \mathbb{E} \left[ \frac{1}{M} \mathrm{Tr}\left( f(\tilde{\C}(\nu))\right) 
    - \frac{1}{M} \mathrm{Tr}\left( f(\frac{\X(\nu)\X^*(\nu)}{B+1}) \right) \right] + \\ 
    \label{eq:fundamental-decomposition}
    \mathbb{E} \left[  \frac{1}{M} \mathrm{Tr}\left( f(\frac{\X(\nu)\X^*(\nu)}{B+1}) \right) \right] - \int_{\Rbb^{+}}f\diff\mu_{MP}^{(c_N)}.
\end{align}
Using the above-mentioned results related to the localization of the eigenvalues of  $\tilde{\C}(\nu)$ and $\hat{\C}(\nu)$, we also argue that it is sufficient to do so when $f$ is compactly supported. 
\item The term $\frac{1}{M} \mathrm{Tr}\left( f(\hat{\C}(\nu))\right) - \frac{1}{M} \mathrm{Tr}\left( f(\tilde{\C}(\nu))\right)$ is studied using the Helffer-Sjöstrand formula which allows, in a certain sense, to be back to the study of  $\frac{1}{M} \mathrm{Tr}\left( \hat{\Q}(z) - \tilde{\Q}(z) \right)$
for $z \in \mathbb{C}^{+}$, where $\hat{\Q}(z)$ and $\tilde{\Q}(z)$ represent the resolvents of 
matrices $\hat{\C}(\nu)$ and $\tilde{\C}(\nu)$ (see below for a formal definition). Using (\ref{eq:representation-tildeC-intro}) and (\ref{eq:representation-hatC}), we express $\frac{1}{M} \mathrm{Tr}\left( \hat{\Q}(z) - \tilde{\Q}(z) \right)$ in terms of the resolvent $\Q(z)$ of the matrix $\frac{\X(\nu)\X^*(\nu)}{B+1}$. As the matrix $\X(\nu)$ is Gaussian, it is possible to use standard Gaussian tools (Poincaré-Nash inequality and the integration by parts formula) to have a good understanding of the behaviour of $\Q(z)$, and to prove that for each $\epsilon > 0$, there exists $\gamma$ independent from $\nu$ such that 
\begin{multline*}
\Prob\left( \left| \frac{1}{M} \mathrm{Tr}\left( f(\hat{\C}(\nu))\right) - \frac{1}{M} \mathrm{Tr}\left( f(\tilde{\C}(\nu))\right) - \right. \right.\\ \left.\left.  \left( \frac{1}{2M} \sum_{m=1}^{M} \frac{s_m''(\nu)}{s_m(\nu)} \right) \; \tilde{\phi}_N(f) \; v_N \; \mathbf{1}_{\alpha > 2/3} \right| > 
N^{\epsilon} u_N \right)  \leq \exp-N^{\gamma}
\end{multline*}
where $\tilde{\phi}_N(f)$ is a deterministic term defined as the action of $f$ on a compactly supported 
distribution $\tilde{D}_N$ depending on $\mu_{MP}^{(c_N)}$. 
\item Using a standard Gaussian concentration inequality as well as the structure of the matrix 
$\tilde{\C}(\nu)$, we obtain that for each $\epsilon > 0$, there exists $\gamma$ independent from 
$\nu$ such that 
\begin{equation}
\Prob\left[ \left| \frac{1}{M} \mathrm{Tr}\left( f(\tilde{\C}(\nu))\right) - \mathbb{E} \left[ \frac{1}{M} \mathrm{Tr}\left( f(\tilde{\C}(\nu))\right)\right] \right| > N^{\epsilon} \frac{1}{B} \right]  \leq \exp-N^{\gamma}
\end{equation}
for each $N$ large enough. 
\item We then analyse the deterministic term $\mathbb{E} \left[ \frac{1}{M} \mathrm{Tr}\left( f(\tilde{\C}(\nu))\right) - \frac{1}{M} \mathrm{Tr}\left( f(\frac{\X(\nu)\X^*(\nu)}{B+1}) \right) \right]$ using the Helffer-Sjöstrand formula. We first show that for each $z \in \mathbb{C}^{+}$, 
$\mathbb{E} \left[ \frac{1}{M} \mathrm{Tr}(\tilde{\Q}(z) - \Q(z))\right]$ is a $\Ocal(\frac{B}{N})^{2}$ term, a non obvious result 
because the relation (\ref{eq:representation-tildeC-intro}) just leads to the conclusion that the above term is $\Ocal(\frac{B}{N})$. Moreover, using long and very tedious Gaussian calculations, we obtain that if $\alpha > \frac{2}{3}$, it holds that 
\begin{multline*}
\mathbb{E} \left[ \frac{1}{M} \mathrm{Tr}(\tilde{\Q}(z) - \Q(z))\right] = - \left( \frac{1}{2M} \sum_{m=1}^{M} \frac{s_m''(\nu)}{s_m(\nu)} \right) \; \tilde{p}_N(z) \; v_N +  \\ \left( \frac{1}{M} \sum_{m=1}^{M} \frac{s_m'(\nu)}{s_m(\nu)} \right)^{2} \; p_N(z) \; v_N + \Ocal\left(\frac{B}{N}\right)^{3}
\end{multline*}
where $p_N$ and $\tilde{p}_N$ are the Stieltjes transforms of the compactly supported distributions 
$D_N$ and $\tilde{D}_N$ introduced previously. This immediately implies that 
if $\alpha \leq \frac{2}{3}$, then 
\begin{multline*}
    \mathbb{E} \left[ \frac{1}{M} \mathrm{Tr}\left( f(\tilde{\C}(\nu))\right) - \frac{1}{M} \mathrm{Tr}\left( f(\frac{\X(\nu)\X^*(\nu)}{B+1}) \right) \right] \\ = \Ocal\left(\frac{B}{N}\right)^{2} = o\left(\frac{1}{B}\right) = o(u_N)
\end{multline*}

while if $\alpha > \frac{2}{3}$, then, 
\begin{multline*}
    \mathbb{E} \left[ \frac{1}{M} \mathrm{Tr}\left( f(\tilde{\C}(\nu))\right) - \frac{1}{M} \mathrm{Tr}\left( f(\frac{\X(\nu)\X^*(\nu)}{B+1}) \right) \right] = \\ - \left( \frac{1}{2M} \sum_{m=1}^{M} \frac{s_m''(\nu)}{s_m(\nu)} \right) \; \tilde{\phi}_N(f) \; v_N +  \left( \frac{1}{M} \sum_{m=1}^{M} \frac{s_m'(\nu)}{s_m(\nu)} \right)^{2} \; \phi_N(f) \; v_N \\ + \Ocal(u_N)
\end{multline*}
because $(\frac{B}{N})^{3} \ll u_N$ if $2/3 < \alpha \leq 4/5$ and $(\frac{B}{N})^{3}$ is equivalent to $u_N$ if $\alpha > 4/5$.
 
\item Finally, classical results imply that 
$$
\mathbb{E} \left[  \frac{1}{M} \mathrm{Tr}\left( f(\frac{\X(\nu)\X^*(\nu)}{B+1}) \right) \right] - \int_{\Rbb^{+}}f\diff\mu_{MP}^{(c_N)} = \Ocal\left(\frac{1}{B^{2}}\right) = o(u_N).$$
\end{itemize}
Gathering the above approximations and using the Lipschitz properties of the function $\nu \rightarrow \psi(f,\nu)$, we finally obtain (\ref{eq:main-result}). \\

We also indicate how the use of lag window estimators of the spectral densities $(s_m)_{m=1, \ldots, M}$ allows to design an estimator $\hat{r}_N(\nu)$ of $r_N(\nu)$ defined by (\ref{def-eq-rN}) 
for which the rate of convergence towards $0$ of the statistics $\hat{\psi}_N(f,\nu)$ obtained by replacing $r_N(\nu)$ by $\hat{r}_N(\nu)$ in Eq. (\ref{eq:def-psi(f)}) is still $u_N$. In particular, 
we establish that for each $\epsilon > 0$, $\Prob\left( \sup_{\nu} |\hat{\psi}_N(f,\nu)| > N^{\epsilon} u_N \right)$ converges towards $0$ exponentially.

\subsection{Assumptions and general notations}
\label{sec:assumptions}

\begin{assumption}
\label{assumption:gaussian_y_n}
For each $m\ge1$, $(y_{m,n})_{n\in\Zbb}$ is a zero mean stationary complex Gaussian time series, ie.
\begin{enumerate}
    \item $\Ebb[y_{m,n}]=0$ for any $m\ge1$ and any $n \in \mathbb{Z}$
    \item every finite linear combination $x$ of the random variables $(y_{m,n})_{n \in \mathbb{Z}}$ is a $\Ncal_\Cbb(0,\sigma^2)$ distributed random variable for some $\sigma^{2}$, i.e. $\mathrm{Re}(x)$ and $\mathrm{Im}(x)$ are independent and $\Ncal(0,\sigma^2/2)$ distributed. 
\end{enumerate}
\end{assumption}

\begin{assumption}
\label{assumption:H0_independence}
If $m_1\neq m_2$, then the scalar time series $(y_{m_1,n})_{n\in\Zbb}$ and $(y_{m_2,n})_{n\in\Zbb}$ are independent.
\end{assumption}

We now formulate the following assumptions on the growth rate of the quantities $N,M,B$: 
\begin{assumption}
\label{assumption:rate_NBM}
$$ B,M=\Ocal(N^\alpha) \text{ where } \frac{1}{2}<\alpha<1, \quad \frac{M}{B+1}=c_N, \quad c_N\xrightarrow[N\to+\infty]{} c\in(0,1).$$
\end{assumption}
As $M = M(N)$ converges towards $+\infty$, we assume that an infinite sequence 
$(y_{1,n})_{n \in \mathbb{Z}}, (y_{2,n})_{n \in \mathbb{Z}}, \ldots, (y_{k,n})_{n \in \mathbb{Z}}, \ldots$ of mutually independent zero mean complex Gaussian time series is given.

We denote by $(s_m)_{m \geq 1}$ the corresponding sequence of spectral densities (i.e. $s_m$ coincides with the spectral density of the times series $(y_{m,n})_{n\in\Zbb}$). For each $m \geq 1$, we denote by $r_m = (r_{m,u})_{u \in \Zbb}$ the autocovariance sequence of $(y_{m,n})_{n\in\Zbb}$, i.e. $r_{m,u} = \Ebb[y_{m,n+u} y_{m,n}^{*}]$. We formulate the following assumptions on $(s_m)_{m \geq 1}$ and $(r_m)_{m \geq 1}$:
\begin{assumption}
\label{assumption:regularity}
The time series $((y_{m,n})_{n\in\Zbb})_{m\ge1}$ are such that:
\begin{equation}
    \label{eq:s-bounded-away-from-zero}
     \inf_{m\ge1}\inf_{\nu\in[0, 1]}|s_m(\nu)|>0
\end{equation}
and
\begin{equation}
    \label{eq:condition-R}
  \sup_{m\ge 1}\sum_{u\in\Zbb}(1+|u|)^{\gamma_0}|r_{m,u}|<+\infty  
\end{equation}
\end{assumption}
where $\gamma_0 \geq 3$. 
Assumption (\ref{eq:condition-R}) of course implies that the spectral densities $(s_m)_{m \geq 1}$ are $\mathcal{C}^{3}$ and that 
\begin{equation}
\label{eq:derivative-spectral-densities}
\sup_{m \geq 1} \sup_{\nu \in [0, 1]} |s_m^{(i)}(\nu)| < +\infty
\end{equation}
for $i=0,1,2,3$ ($s_m^{(i)}$ represents the derivative of order $i$ of $s_m$). We notice that (\ref{eq:condition-R}) holds as soon as we have 
$$
\sup_{m\ge1} |r_{m,u}| \leq \frac{C}{|u|^{1+\gamma_0+\delta}}
$$
for each $u \neq 0$ as well as $\sup_{m\ge1} |r_{m,0}| < \infty$
($C >0$ and  $\delta > 0$ represent constants). If $z$ represents the backward shift operator, a simple example of time series satisfying Assumption \ref{assumption:regularity} is to consider an
ARMA time series generated as 
$$
y_{m,n} = [h_m(z)]\epsilon_{m,n}
$$
where $((\epsilon_{m,n})_{n \in \mathbb{Z}})_{m \geq 1}$  are mutually independent 
i.i.d. $\mathcal{N}_{\Cbb}(0,1)$ sequences, and where $h_m(z) = \frac{b_m(z)}{a_m(z)}$,
$a_m$ and $b_m$ being 2 polynomials having no pole or zero in the closed unit disk 
$\overline{\mathbb{D}}$. 
Moreover, $\sup_{m \geq 1} \max(\mathrm{deg}(a_m), \mathrm{deg}(b_m)) < +\infty$, and if 
$(z_{k,m})_{k=1, \ldots, \mathrm{deg}(b_m)}$ and $(p_{k,m})_{k=1, \ldots, \mathrm{deg}(a_m)}$
are the zeros of $b_m$ and $a_m$, then we should have
\begin{align*}
& \inf_{m \geq 1} \inf_k \mathrm{dist}(z_{k,m}, \overline{\mathbb{D}}) > 0, \;  \inf_{m \geq 1} \inf_k \mathrm{dist}(p_{k,m}, \overline{\mathbb{D}}) > 0  \\
& \sup_{m \geq 1} \sup_k |z_{k,m}| < +\infty,  \; \sup_{m \geq 1} \sup_k |p_{k,m}| < +\infty .
\end{align*}
It is easy to check that (\ref{eq:condition-R}) holds for each $\gamma_0 > 0$, and that 
(\ref{eq:s-bounded-away-from-zero}) is satisfied as well. \\

\textbf{Notations.} A zero mean complex valued random vector $\y$ is said to be
$\Ncal_\Cbb(0,\Sigmabs)$ distributed if $\mathbb{E}(\y \y^{*}) = \Sigmabs$ and if 
each linear combination $x$ of the entries of $\y$ is a complex Gaussian random variable, 
i.e. $\mathrm{Re}(x)$ and $\mathrm{Im}(x)$ are independent Gaussian random variables
sharing the same variance. If $x$ is a random variable, we denote by $x^{\circ}$ the random variable defined by 
\begin{equation}
    \label{eq:def-xrond}
    x^{\circ} = x - \Ebb[x].
\end{equation}

If $\A$ is a $P \times Q$ matrix, $\| \A \|$ and $\| \A \|_{F}$ denote its spectral norm 
and Frobenius norm respectively. If $P=Q$ and $\A$ is 
Hermitian, $\lambda_1(\A) \geq \ldots \geq \lambda_P(\A)$ are the eigenvalues of $\A$. The spectrum of $\A$, which is here the set of its eigenvalues $(\lambda_k(\A))_{k=1, \ldots, P}$, is denoted by
$\sigma(\A)$. For $\A$ and $\B$ square Hermitian matrices, if all the eigenvalues of $\A-\B$ are non negative, we write $\A\ge\B$. We define $\Re\, \A=(\A+\A^*)/2$ and $\Im\, \A=(\A-\A^*)/2$ where $\A^*$ is the conjugate transpose of the matrix $\A$.  \\

$\mathcal{C}^p$ represents the set of all real-valued functions defined on $\mathbb{R}$ whose first $p$ derivatives exist and are continuous, and $\mathcal{C}^{p}_c$ is the set of all compactly supported functions of $\mathcal{C}^{p}$. \\

We recall that ${\bf S}(\nu)$ represents the $M \times M$ diagonal 
matrix ${\bf S}(\nu) = \diag(s_1(\nu), \ldots, s_M(\nu))$. We notice that 
${\bf S}$ depends on $M$, thus on $N$ (through $M:=M(N)$), but we often omit to mention the corresponding
dependency in order to simplify the notations. In the following, we will denote by $\y_m$ 
the $N$--dimensional vector $\y_m = (y_{m,1}, \ldots, y_{m,N})^{T}$. \\

A nice constant is a positive a constant that does not depend on the frequency $\nu$, 
the time series index $m$, the complex variable $z$ of the various resolvents and Stieltjes transforms used throughout the paper, as well as on the dimensions $B,M$ and $N$. 
A nice polynomial is a polynomial whose degree and coefficients are nice constants. If $z \in \mathbb{C}^{+}$ and if $P_1$ and $P_2$ are two nice polynomials, terms such as $P_1(z) P_2(\frac{1}{\Im z})$ play an important role in the following. $C$ and 
$C(z)$ will represent a generic notation for respectively a nice constant and a term
$P_1(z) P_2(\frac{1}{\Im z})$, and the values of $C$ and $C(z)$ may change from one line to the other. \\

If $(a_N)_{N \geq 1}$ and $(b_N)_{N \geq 1}$ are two sequences of positive real numbers, 
we write $a_N << b_N$ if $\frac{a_N}{b_N} \rightarrow 0$ when $N \rightarrow +\infty$. \\

We also recall how a function can be applied to Hermitian matrices. For an $M\times M$ Hermitian matrix $\A$ with spectral decomposition $\U\Lambdabs\U^*$ where $\Lambdabs=\diag(\lambda_m, m=1,\ldots,M)$ and the $(\lambda_m)_{m=1, \ldots, M}$ are the real eigenvalues of $\A$, then for any function $f$ defined on $\mathbb{R}$, we define $f(\A)$ as:
$$ f(\A)=\U \begin{pmatrix} f(\lambda_1) & & \\ & \ddots & \\ & & f(\lambda_M) \end{pmatrix} \U^*$$

$\mathbb{C}^{+}$ is the upper half-plane of $\mathbb{C}$, i.e. the set of all
complex numbers $z$ for which $\mathrm{Im}\,z>0$. \\

For $\mu$ a probability measure, its Stieltjes transform $s_\mu$ is the function defined on $\Cbb \setminus \Supp\mu$ as 
\begin{equation}
\label{definition:stieltjes_transform}
    s_\mu(z)=\int\frac{\diff \mu(\lambda)}{\lambda-z}.
\end{equation}

We recall that 
\begin{equation}
    \label{eq:inegalite_stieltjes_transform}
    |s_\mu(z)|\le\frac{1}{\Im\,z}
\end{equation}
 for each $z\in\Cbb^+$. Moreover, if $\mu$ is carried by $\mathbb{R}^{+}$, then for any $a > 0$, the function $-\frac{1}{z(1+as_{\mu}(z))}$ is also the Stieljes transform of a probability distribution carried by $\mathbb{R}^{+}$, a property which implies that 
 \begin{equation}
     \label{eq:control-inverse-1pluss}
     \left| \frac{1}{1+a s_{\mu}(z)} \right| \leq \frac{|z|}{\Im z}
 \end{equation}
for each $z\in\Cbb^+$ (see \cite{hachemloubatonnajim2007detequiv}, Proposition 5-1, item 4). \\

If $\lambda_1,\ldots,\lambda_M$ denote the eigenvalues of a Hermitian matrix $\A$ and if $\mu:=\frac{1}{M}\sum_{i=1}^M\delta_{\lambda_i}$ denotes the empirical eigenvalue distribution of $\A$, then we have the following relation:
$$ s_\mu(z) = \frac{1}{M} \Tr\Q_{\A}(z)$$
where $\Q_{\A}(z)$ represents the resolvent of $\A$ defined by 
\begin{equation}
    \label{eq:def-resolvent}
    \Q_{\A}(z) = ( \A - z \I_M )^{-1}.
\end{equation}

We finally mention the following useful control for the norm $\Q_{\A}$. For each $z \in \Cbb^{+}$, we have
\begin{equation}
    \label{eq:inegalite_stieltjes_transform_Q}
    \|\Q_{\A}\| \le \frac{1}{\Im\,z}.
\end{equation}

\subsection{Overview of the paper}
We first recall in Section \ref{sec:useful-results} useful technical 
tools: in Paragraph \ref{subsec:stochastic-domination}, the concept of stochastic domination adapted from \cite{erdHos2013averaging}
which allows to considerably simplify the exposition of the following 
results,  in Paragraph \ref{section:haagerup} some useful properties of 
the extreme eigenvalues and of the resolvent of large Wishart matrices, two well-known Gaussian concentration inequalities expressed using the 
stochastic domination framework in Paragraphs \ref{subsection:lipschitz_concentration} and \ref{subsection:hanson}, and the Helffer-Sjöstrand formula in Paragraph \ref{subsec:hs-formula}. We establish in Section \ref{sec:representation-tilde-hat} the stochastic representations \eqref{eq:representation-tildeC-intro} and \eqref{eq:representation-hatC} of  $\tilde{\C}(\nu)$ and $\hat{\C}(\nu)$. In Section \ref{sec:lss}, we prove for each $\nu$ the concentration of $|\psi_N(f,\nu)|$ defined by (\ref{eq:def-psi(f)}), and indicate how it is possible to estimate the term $r_N(\nu)$ in order to keep equal to $u_N$ the rate of convergence of the statistics $\hat{\psi}_N(f,\nu)$ obtained by replacing $r_N(\nu)$ by $\hat{r}_N(\nu)$ 
in  (\ref{eq:def-psi(f)}). In Section \ref{sec:lipschitz}, we establish Lipschitz properties for the functions $\nu \rightarrow \psi_N(f,\nu)$ and $\nu \rightarrow \hat{\psi}_N(f,\nu)$ that allow to establish the concentration of $\sup_{\nu} |\psi_N(f,\nu)|$ and  $\sup_{\nu} |\hat{\psi}_N(f,\nu)|$. We finally provide in Section \ref{sec:simulations} some numerical simulations that support our results.

\section{Useful technical tools}
\label{sec:useful-results}

\subsection{Stochastic domination}
\label{subsec:stochastic-domination}
We now present the concept of stochastic domination introduced in \cite{erdHos2013averaging}. A nice introduction to this tool can also be found in the lecture notes \cite{benaych2016lectures}.

\begin{definition}{\textbf{Stochastic Domination.}}
\label{definition:stochastic_domination}
Let 
$$ X=(X^{(N)}(u):N\in\Nbb, u \in U^{(N)}), \quad Y=(Y^{(N)}(u):N\in\Nbb, u \in U^{(N)})$$
be two families of nonnegative random variables, where $U^{(N)}$ is a set that may possibly depend on $N$.  
We say that $X$ is stochastically dominated by $Y$ if for all (small) $\epsilon>0$, there exists some $\gamma>0$ (which of course depends on $\epsilon$) such that:
$$\Prob\left[X^{(N)}(u)>N^\epsilon Y^{(N)}(u)\right]\le \exp-N^\gamma $$ 
for each $u \in U^{(N)}$ and for each large enough $N>N_0(\epsilon)$, where $N_0(\epsilon)$
is independent of $u$, or equivalently
\begin{equation}
\label{equation:definition_domination_stochastique}
    \sup_{u \in U^{(N)}} \Prob\left[X^{(N)}(u)>N^\epsilon Y^{(N)}(u)\right]\le \exp-N^\gamma .
\end{equation}

for each large enough $N>N_0(\epsilon)$. If $X$ is stochastically dominated by Y we use the notation $X^{(N)}(u)\prec Y^{(N)}(u)$. To simplify the notations, we will very often denote $X^{(N)} \prec Y^{(N)}$ or $X \prec Y$ when the context will be clear enough. 
Moreover, if for some complex valued family $X$ we have $|X| \prec Y$ we also write $X=\Ocal_\prec(Y)$. \\

Finally, we say that a family of events $\Xi=\Xi^{(N)}(u)$ holds with exponentially high (small) probability if there exist $N_0$ and $\gamma>0$ such that for $N\ge N_0$, $\Prob[\Xi_N(u)]>1-\exp-N^\gamma$ ($\Prob[\Xi_N(u)]<\exp-N^\gamma$) for each $u \in U^{(N)}$.
\end{definition}

\begin{remark}
\label{remark:domination_stochastique}
Suppose $(X_N)_{N\in\Nbb}$ is a sequence of positive random variables, satisfying $X_N\prec a_N N^\epsilon$ for any $\epsilon>0$ for some positive real numbers sequence $(a_N)_{N \in \Nbb}$. It turns out that this precisely means that $X_N\prec a_N$. Indeed, consider an arbitrary $\epsilon'>0$. By the stochastic domination property of $X_N$, one can take $\epsilon$ such that $0<\epsilon<\epsilon'$ and write
\[
    \Prob\left[X_N> a_N \times N^{\epsilon'}\right] \le \Prob\left[X_N> a_N \times N^{\epsilon} \times \underbrace{N^{\epsilon'-\epsilon}}_{\gg 1}\right] \le \Prob\left[X_N>a_N \times N^{\epsilon} \right]  
\]
which goes to zero exponentially since $X_N\prec a_N N^\epsilon$ for the $\epsilon$ chosen. This argument will be used in the proof of Lemma \ref{lemma:variance-zeta}.
\end{remark}

\begin{lemma}
\label{lemma:algebra_domination}
Take four families of non negative random variables $X_1,X_2, Y_1$ and $Y_2$ defined as in Definition \ref{definition:stochastic_domination}. Then the following holds:
$$ X_1\prec Y_1 \text{ and } X_2\prec Y_2  \implies X_1+X_2\prec Y_1+Y_2 \text{ and } X_1X_2\prec Y_1Y_2.$$ 
\end{lemma}
We omit the proof of this lemma. 

\begin{remark}
Note that Definition \ref{definition:stochastic_domination} is slightly different from the original one \cite{erdHos2013averaging} which states that the left hand side of \eqref{equation:definition_domination_stochastique} should be bounded by a quantity of order $N^{-D}$ for any finite $D>0$. In the present paper,  all the random variables are Gaussian, and exponential concentration rates can be achieved. 
\end{remark}

\subsection{Properties of the eigenvalues and of the resolvent of large Wishart matrices}
\label{section:haagerup}

In this paper we will at multiple occasion use properties of the eigenvalues of matrices 
$\frac{\X_N \X_N^{*}}{B+1}$ where $\X_N$ is an $M \times (B+1)$ complex Gaussian matrix with i.i.d. $\mathcal{N}_{\Cbb}(0,1)$ entries when $M=M(N)$ and $B=B(N)$ follow Assumption \ref{assumption:rate_NBM}.
\subsubsection{Concentration of the largest and the smallest eigenvalues}
\label{subsubsection:concentration-haagerup}
We first recall concentration results of the largest and smallest eigenvalue of $\frac{\X_N \X_N^{*}}{B+1}$ due to \cite{haagerup2003random}. We have for any $\epsilon > 0$ 
\begin{eqnarray}
\label{eq:concentration-smallest-eig-wishart}
\Prob\left[\lambda_M\left(\frac{\X_N\X_N^*}{B+1}\right)  <  (1-\sqrt{c})^2-\epsilon\right] & \le & (B+1)\exp-C(B+1)\epsilon^2 \\
\label{eq:concentration-largest-eig-wishart}
\Prob\left[\lambda_1\left(\frac{\X_N\X_N^*}{B+1}\right)  >  (1+\sqrt{c})^2+\epsilon\right] & \le & (B+1)\exp-C(B+1)\epsilon^2
\end{eqnarray}
for some nice constant $C$.



Consider for $\epsilon>0$, the $\epsilon$--expansion of the support of the Marchenko-Pastur distribution $\mu_{MP}^{(c)}$: 
$$\Supp\mu_{MP}^{(c)}+\epsilon:=\left[(1-\sqrt{c})^2-\epsilon, (1+\sqrt{c})^2+\epsilon\right]$$
and the event: 
\begin{equation}
\label{equation:definition_Xi}
    \Lambda_{N,\epsilon}=\left\{\sigma\left(\frac{\X_N\X_N^*}{B+1}\right)\subset\Supp\mu_{MP}^{(c)}+\epsilon\right\}.
\end{equation}

It is clear that using \eqref{eq:concentration-smallest-eig-wishart} and \eqref{eq:concentration-largest-eig-wishart}, $\Lambda_{N,\epsilon}$ holds with exponentially high probability for any $\epsilon>0$. This will be of high importance in the following since it will enable us to work on events of exponentially high probability where the norm
of $\frac{\X_N\X_N^*}{B+1}$ and the norm of its inverse are bounded. \\

Finally, the following (weaker) statement is a simple consequence of the equations \eqref{eq:concentration-smallest-eig-wishart} and \eqref{eq:concentration-largest-eig-wishart}, which will sometimes be enough in the following:
\begin{equation}
\label{eq:stoc-domination-eig-sing-wishart}    
\lambda_1\left(\frac{\X_N\X_N^*}{B+1}\right) + \frac{1}{\lambda_M\left(\frac{\X_N\X_N^*}{B+1}\right)} \prec 1.
\end{equation}


We finally notice that if we consider a family $\X_{N}(u)\in\Cbb^{M\times(B+1)}$ with i.i.d. $\Ncal_\Cbb(0,1)$ entries, $u \in U^{(N)}$, where 
$U^{(N)}$ is a certain set possibly depending on $N$, then (\ref{eq:concentration-smallest-eig-wishart}) and (\ref{eq:concentration-largest-eig-wishart})
hold for each $u \in U^{(N)}$ because the constant $C$ in (\ref{eq:concentration-smallest-eig-wishart}) and (\ref{eq:concentration-largest-eig-wishart}) is universal. This implies that the  stochastic domination (\ref{eq:stoc-domination-eig-sing-wishart})
is still satisfied by the family $\X_{N}(u)$, $u \in U^{(N)}$. Moreover, 
the family of events $\Lambda_{N,\epsilon}(u)$ defined by \eqref{equation:definition_Xi} when $\X_N$ is replaced by $\X_N(u)$ still holds with exponentially high probability.  

\subsubsection{Asymptotic behaviour of the resolvent of \texorpdfstring{$\frac{\X_N \X_N^{*}}{B+1}$}{}}
\label{subsubsec:resolvent-MP}
 We next review known results related to the asymptotic behaviour of the resolvent $\Q_N(z)$ of 
matrix $\frac{\X_N \X_N^{*}}{B+1}$ that can be deduced from standard Gaussian tools. 
The Poincaré-Nash inequality (see e.g. \cite[Proposition 2.1.6]{pasturshcherbina2011} in the Gaussian real case and Eq. (18) in \cite{hachemkhorunzhyloubatonnajimpastur2008newapproach} in the complex Gaussian case) implies immediately that the following Lemma holds. 
\begin{lemma}
\label{le:nash-poincare-resolvent}
Consider deterministic  $M \times M$ and $(B+1) \times (B+1)$ matrices $\A$ and $\tilde{\A}$. Then, it holds that 
\begin{eqnarray}
    \label{eq:var-trace-resolvent}
    \Var \frac{1}{M} \Tr \A \Q_N^{i}(z) & \leq & \frac{C(z)}{M^{2}} \, \frac{1}{M} \Tr \A \A^{*} \\
    \label{eq:var-trace-resolvent-XX*}
    \Var \frac{1}{M} \Tr \left( \frac{\X \tilde{\A} \X^{*}}{B+1} \Q_N^{i}(z) \right) & \leq & \frac{C(z)}{M^{2}} \, \frac{1}{B+1} \Tr \tilde{\A} \tilde{\A}^{*}
\end{eqnarray}
for $i=1,2$
\end{lemma}
We recall that $C(z)$ represents a generic notation for $P_1(z) P_2(\frac{1}{\Im z})$ where $P_1$ and
$P_2$ are nice polynomials. \\

The integration by parts formula states that if $h(\X,\X^{*})$ is a $\mathcal{C}^{1}$ function of the entries of $\X$ and $\X^{*}$ with polynomially bounded first derivatives, then, it holds that 
\begin{equation}
    \label{eq:integration-by-part}
    \mathbb{E}(X_{ij} h(\X,\X^{*})) = \mathbb{E}|X_{ij}|^{2} \mathbb{E} \left[ \frac{\partial h}{\partial \overline{X}_{ij}}(\X,\X^{*})\right].
\end{equation}
\eqref{eq:integration-by-part}, in conjunction with the Poincaré-Nash inequality, allows to 
evaluate easily the asymptotic behaviour of the entries of $\mathbb{E}(\Q_N(z))$ (see e.g. \cite{pasturshcherbina2011}). We first notice that properties of the distribution of the matrix $\X_N$ immediately imply that $\mathbb{E}(\Q_N(z))$ 
is reduced to $\beta_N(z) \I_{M}$ where $\beta_N(z)$ coincides with $\mathbb{E}(\Q_{m,m}(z))$ for each $m$. 
Then, it holds that 
\begin{equation}
    \label{eq:beta-t}
    \beta_N(z) = t_N(z) + \epsilon_N(z)
\end{equation}
where the error term $\epsilon_N(z)$ satisfies $|\epsilon_N(z)| \leq \frac{C(z)}{M^{2}}$ and where $t_N(z)$ is the Stieltjes transform of the Marcenko-Pastur distribution $\mu_{MP}^{(c_N)}$. In other words, $t_N(z)$ is the unique Stieltjes transform satisfying the equation 
\begin{equation}
    \label{eq:equation-MP1}
    t_N(z) = \frac{1}{-z + \frac{1}{1+c_N t_N(z)}}.
\end{equation}
It is also convenient to define $\tilde{t}_N(z)$ by 
\begin{equation}
    \label{eq:def-ttilde}
 \tilde{t}_N(z) = - \frac{1}{z(1 + c_N t_N(z))}   
\end{equation}
so that $t_N(z)$ is also given by 
\begin{equation}
    \label{eq:equation-MP2}
    t_N(z) =  - \frac{1}{z(1 + \tilde{t}_N(z))} .
\end{equation}
It is well-known that $\tilde{t}_N(z)$ is the Stieltjes transform of the probability distribution 
$c_N \mu_{MP}^{(c_N)} + (1-c_N) \delta_0$. \\

We finally mention that $\mathbb{E}(\Q_N'(z)) = \mathbb{E}(\Q_N^{2}(z)) = \beta_N'(z) \I_M$ 
(where $'$ stands for the derivative w.r.t. $z$), and that $\epsilon_N'(z) = \beta_N^{'}(z) -t_N'(z)$
still satisfies 
\begin{equation}
    \label{eq:control-derivee-epsilon}
    |\epsilon_N'(z)| \leq  \frac{C(z)}{M^{2}}.
\end{equation}
\subsection{Concentration of functionals of Gaussian entries}
\label{subsection:lipschitz_concentration}
It is well-known (see e.g. \cite[Th. 2.1.12]{tao2011topics}) that for any 1-Lipschitz real valued function $f$ defined on $\mathbb{R}^{N}$ and any $N$--dimensional random variable $\X\sim\Ncal(0,\I_N)$, there exists a universal constant $C$ such that:
\begin{equation}
    \label{eq:gaussian-concentration-inequality}
    \Prob\left[\left|f(\X)-\Ebb f(\X)\right|>t\right] \le C\exp-Ct^2 .
\end{equation}
This inequality is still valid when $\X\sim\Ncal_\Cbb(0,\I_N)$: in this context, 
$f(\X)$ is replaced by a real-valued function 
 $f(\X, \X^{*})$ depending on the entries 
 of $\X$ and $\X^{*}$. $f(\X, \X^{*})$ can of course be written as
 $f(\X, \X^{*}) = \tilde{f}(\sqrt{2} \mathrm{Re}(\X), \sqrt{2} \mathrm{Im}(\X))$ for some 
 function $\tilde{f}$ defined on $\mathbb{R}^{2N}$. As $(\sqrt{2} \mathrm{Re}(\X), \sqrt{2}\mathrm{Im}(\X))$ is $\Ncal(0,\I_{2N})$ distributed, the concentration inequality is still valid for 
 $f(\X, \X^{*}) = \tilde{f}(\sqrt{2} \mathrm{Re}(\X), \sqrt{2} \mathrm{Im}(\X))$. We just finally mention that  $f$, considered 
 as a function of  $(\X, \X^{*})$, and $\tilde{f}$ have Lipschitz constants that are of 
 the same order of magnitude. More precisely, if  we define the differential operators
 $\frac{\partial}{\partial z}$ and $\frac{\partial}{\partial \bar{z}}$ by 
$$ \frac{\partial}{\partial z} =  \frac{\partial}{\partial x} - i  \frac{\partial}{\partial y}, \quad 
 \frac{\partial}{\partial \bar{z}} =  \frac{\partial}{\partial x} + i  \frac{\partial}{\partial y}  $$
 we can verify immediately that 
 $$
 \sum_{i=1}^{N} \left( \left|  \frac{\partial f}{\partial X_i} \right|^{2} + 
  \left|  \frac{\partial f}{\partial X_i^{*}} \right|^{2} \right) = \| \left(\nabla f \right)_{(\X, \X^{*})} \|^{2} = 4 \| \left( \nabla \tilde{f} \right)_{(\sqrt{2} \mathrm{Re}(\X), \sqrt{2} \mathrm{Im}(\X))}\|^{2} .
 $$
 
Within the stochastic domination framework, the concentration inequality (\ref{eq:gaussian-concentration-inequality}) implies that for a family $\X_N(u) \sim\Ncal(0,\I_N)$ for $u \in U^{(N)}$:
$$ \left|f(\X_N(u))-\Ebb f(\X_N(u))\right|\prec 1 $$
The proof is immediate: consider $\epsilon>0$ and obtain that 
$$ \Prob[|f(\X_N(u))-\Ebb f(\X_N(u))|>N^\epsilon]\le C\exp-CN^{2\epsilon} $$ for each $u$
as expected. This result can easily be extended in the complex case, ie. when $\X_N(u) \sim\Ncal_\Cbb(0,\I_N)$.

\subsection{Hanson-Wright inequality}
\label{subsection:hanson}
The Hanson-Wright inequality \cite{rudelson2013hanson} is useful to control deviations of a quadratic form from its expectation. While it is proved in the real case in \cite{rudelson2013hanson}, it can easily be understood that it can be extended in the complex case as follows: let $\X \sim\Ncal_\Cbb(0,\I_N)$ and $\A\in\Cbb^{N\times N}$. Then
\begin{align}
\label{eq:hanson-wright}
    \Prob[|\X^*\A\X-\Ebb\X^*\A\X|>t] \le 2\exp-C\min\left(\frac{t^2}{\|\A\|_F^2},\frac{t}{\|\A\|}\right).
\end{align}

We now write \eqref{eq:hanson-wright} in the stochastic domination framework. Consider a family of independent 
$\Ncal_\Cbb(0,1)$ random variables $(X_n(u))_{n=1, \ldots, N}$ where $u \in U^{(N)}$ and a sequence of $N\times N$ matrices $\A_N(u)$ that possibly depend on $u$. Take $\epsilon>0$ and $t=N^\epsilon\|\A_N(u)\|_F$. Since $\|\A_N(u)\|>0$, $\|\A_N(u)\|_F>0$, and $\|\A_N(u)\|\le\|\A_N(u)\|_F$:
\begin{align*}
    \min\left(\frac{t}{\|\A_N(u)\|},\frac{t^2}{\|\A_N(u)\|_F^2}\right) &= \min\left(N^\epsilon\frac{\|\A_N(u)\|_F}{\|\A_N(u)\|},N^{2\epsilon}\frac{\|\A_N(u)\|_F^2}{\|\A_N(u)\|_F^2}\right) \\
    &\ge \min(N^\epsilon,N^{2\epsilon}) = N^\epsilon.
\end{align*}
Denote $\X_N(u) = (X_1(u), \ldots, X_N(u))^{T}$. For any $u\in U^{(N)}$, it holds that:
\begin{multline}
\label{equation:hanson_wright_concentration_preuve}
    \Prob\left[|\X_N^*(u)\A_N(u)\X_N(u)-\Ebb\X_N^*(u)\A_N(u)\X_N(u) | > N^\epsilon\|\A_N(u)\|_F\right] \\ \le 2\exp-CN^\epsilon.
\end{multline}
We can therefore rewrite \eqref{equation:hanson_wright_concentration_preuve} as the following stochastic domination:
\begin{equation}
    \label{equation:hanson_wright_stochastic_domination}
    |\X_N^*(u)\A_N(u)\X_N(u)-\Ebb\X_N^*(u)\A_N(u)\X_N(u) |\prec \|\A_N(u)\|_F.
\end{equation}

\subsection{Helffer-Sjöstrand formula}
\label{subsec:hs-formula}
If $\mu$ is a probability measure, the Helffer-Sjöstrand formula can be seen as an alternative to the Stieltjes inversion formula that allows to express $\int f d\mu$ in terms of the Stieltjes transform $s_{\mu}(z)$ of $\mu$ (see \eqref{definition:stieltjes_transform}) when $f$ is a regular enough compactly supported function. In order to introduce this tool, we consider a class $\mathcal{C}^{k+1}$ compactly supported function $f$ for a certain integer $k$, and denote by $\Phi_k(f):\Cbb\to\Cbb$ the function defined on $\mathbb{C}$ by 
$$ \Phi_k(f)(x+iy)=\sum_{l=0}^k\frac{(iy)^l}{l!}f^{(l)}(x)\rho(y)$$
where $\rho:\Rbb\to\Rbb^+$ is smooth, compactly supported, with value 1 in a neighbourhood of $0$. Function $\Phi_k(f)$ coincides with $f$ on the real line and extends it to the complex plane. Let $\bar{\partial}=\partial_x+i\partial_y$. 
It is well-known that 
\begin{equation}
    \label{eq:derivee-Phi(f)}
    \bar{\partial}\Phi_k(f)(x+iy)=\frac{(iy)^k}{k!}f^{(k+1)}(x)
\end{equation}
(a proof of this result can be found in \cite{dyn1972operator} or \cite{helffer1989equation}) 
if $y$ belongs to the neighbourhood of $0$ in which $\rho$ is equal to 1. The Helffer-Sjöstrand formula can be written as
\begin{equation}
        \label{equ:helfjer-sjostrand-general}
        \int f\diff \mu   = \frac{1}{\pi}\Re\int_{\Cbb^+} \,\bar{\partial}\Phi_k(f)(z)
        s_{\mu}(z) \diff x \diff y.
\end{equation}
In order to understand why the integral at the right hand side of \eqref{equ:helfjer-sjostrand-general}
is well defined, we take, to fix the ideas, $\rho\in \mathcal{C}^\infty$ such that $\rho(y)=1$ for $|y|\le1$ and $\rho(y)=0$ for $|y|>2$, and denote by  $[a_1, a_2]$ an interval containing the support of $f$. Then, it  appears that the integral on  $\Cbb^+$ is in fact over the compact set $\Dcal =  \{x+iy:x\in[a_1,a_2],y\in[0,2]\}$. Moreover, as $|\s_{\mu}(z)| \leq \frac{1}{y}$ if $z \in \Dcal$ (see \eqref{eq:inegalite_stieltjes_transform}), \eqref{eq:derivee-Phi(f)} for $k=1$ leads to the conclusion that 
$$
|\bar{\partial}\Phi_k(f)(z) s_{\mu}(z)| \leq C
$$
for $z \in \{ x + i y \in \Dcal, y \leq 1 \}$. Therefore, the right hand side of 
\eqref{equ:helfjer-sjostrand-general} is well defined. 

We finally mention that the Helffer-Sjöstrand formula remains still valid for any compactly supported 
distribution $D$ (see e.g. \cite{loubaton2016jotp}, section $9$). The Stieltjes transform of $D$, denoted by $s_D(z)$,  is defined for each $z \in \Cbb^+$ as the action of the function $\lambda \rightarrow \frac{1}{\lambda - z}$ on $D$, i.e. $s_D(z) = < D, \frac{1}{\lambda -z}>$, and satisfies 
$$
|s_D(z)| \leq C \left(1 + \frac{1}{(\Im z)^{n_0}} \right)
$$
for each $z \in  \Cbb^+$ where $n_0$ is related to the order of the distribution. We refer the reader to 
\cite{capitainedonatiferal2009annprob} (Theorem 4.3) and the references therein for more details on Stieltjes transforms of distributions. 
Then, if $f$ is a $\mathcal{C}^{\infty}$ function supported by $[a_1, a_2]$, $<D,f>$ is given by
\begin{equation}
        \label{equ:helfjer-sjostrand-distribution}
        <D,f>   = \frac{1}{\pi}\Re\int_{\mathcal D} \,\bar{\partial}\Phi_k(f)(z)
        s_{D}(z) \diff x \diff y
\end{equation}
for $k \geq n_0$. We also recall that an alternative expression for $<D,f>$ is given by the Stieltjes inversion formula, also valid for distributions, i.e. 
\begin{equation}
    \label{eq:inversion-formula-distribution}
    <D,f>   = \frac{1}{\pi} \lim_{y \rightarrow 0} \int_{a_1}^{a_2} f(\lambda) \, \Im s_D(\lambda + iy) \, d\lambda.
\end{equation}

\section{Stochastic representations of  \texorpdfstring{$\tilde{\C}(\nu)$}{C tilde (nu)} and \texorpdfstring{$\hat{\C}(\nu)$}{C hat (nu)}}
\label{sec:representation-tilde-hat}
The first step is to show that $\tilde{\C}(\nu)$ and 
$\hat{\C}(\nu)$ can be approximated by the sample covariance matrix of a sequence of i.i.d. Gaussian random vectors, and to control the order of magnitude of the corresponding errors. This is the objective of the following result.

\begin{theorem}
\label{theorem:C_approximation_Wishart}
Under Assumptions \ref{assumption:gaussian_y_n}, \ref{assumption:H0_independence}, \ref{assumption:rate_NBM} and \ref{assumption:regularity}, for any $\nu\in[0, 1]$, there exists an $M \times(B+1)$ random matrix $\X_N(\nu)$ with $\Ncal_\Cbb(0,1)$ i.i.d. entries, and two matrices $(\tilde{\Deltabs}_N(\nu), \Deltabs_N(\nu))$ such that:
\begin{eqnarray} 
    \label{equation:tildeC_approximation_Wishart}
    \tilde{\C}_N(\nu) & = &\frac{\X_N(\nu)\X_N^*(\nu)}{B+1} + \tilde{\Deltabs}_N(\nu), \quad \|\tilde{\Deltabs}_N(\nu)\|\prec \frac{B}{N} \\
    \label{equation:C_approximation_Wishart}
    \hat{\C}_N(\nu) & =  & \frac{\X_N(\nu)\X_N^*(\nu)}{B+1} + \Deltabs_N(\nu), \quad \|\Deltabs_N(\nu)\|\prec \frac{1}{\sqrt{B}} + \frac{B}{N}.
\end{eqnarray}
\end{theorem}
\begin{remark}
Therefore, up to small additive perturbations, $\tilde{\C}_N(\nu)$ and $\hat{\C}_N(\nu)$ appear as empirical covariance matrices of i.i.d. $\Ncal_\Cbb(0,\I_M)$ random vectors. We thus expect that $\tilde{\C}_N(\nu)$ and $\hat{\C}_N(\nu)$ will satisfy a number of useful properties of empirical covariance matrices of i.i.d. $\Ncal_\Cbb(0,\I_M)$ random vectors. 
\end{remark}
In particular, Theorem \ref{theorem:C_approximation_Wishart} allows to make precise the location of the eigenvalues of 
$\tilde{\C}_N(\nu)$ and $\hat{\C}_N(\nu)$.  In order to 
formulate the corresponding result, we define some notations. We introduce
the events $\Lambda^{\tilde{\C}}_{N,\epsilon}(\nu)$ and $\Lambda^{\hat{\C}}_{N,\epsilon}(\nu)$ defined by
\begin{eqnarray}
\label{eq:def-Lambda_C_tilde}
\Lambda^{\tilde{\C}}_{N,\epsilon}(\nu)  & = & \{\sigma(\tilde{\C}_N(\nu))\subset\Supp\mu_{MP}^{(c)}+\epsilon\} \\
\Lambda^{\hat{\C}}_{N,\epsilon}(\nu)  & = & \{\sigma(\hat{\C}_N(\nu))\subset\Supp\mu_{MP}^{(c)}+\epsilon\}
\label{eqdef-Lambda_C_hat}.
\end{eqnarray}
Then, we establish in the following the Corollary: 
\begin{corollary}
\label{coro:localisation-eigenvalues-hatC-tildeC}
For each $\epsilon > 0$, the family of events $\Lambda^{\tilde{\C}}_{N,\epsilon}(\nu), N \geq 1, \nu \in [0,1]$ and 
$\Lambda^{\hat{\C}}_{N,\epsilon}(\nu), N \geq 1, \nu \in [0,1]$ hold with exponential high probability. 
\end{corollary}

\begin{remark}
In the following, we will often omit to mention that the various matrices 
under consideration depend on $N$ and $\nu$. Matrices $\hat{\C}_N(\nu), \tilde{\C}_N(\nu), \X_N(\nu), \Deltabs_N(\nu), \ldots$ will therefore 
be denoted by  $\hat{\C}(\nu), \tilde{\C}(\nu), \X(\nu), \Deltabs(\nu), \ldots$ or $\hat{\C}, \tilde{\C}, \X, \Deltabs, \ldots$. We will also denote $\Lambda^{\tilde{\C}}_{N,\epsilon}(\nu)$ and 
$\Lambda^{\hat{\C}}_{N,\epsilon}(\nu)$ by $\Lambda^{\tilde{\C}}_{\epsilon}(\nu)$ or $\Lambda^{\tilde{\C}}_{\epsilon}$ and 
$\Lambda^{\hat{\C}}_{\epsilon}(\nu)$ or $\Lambda^{\hat{\C}}_{\epsilon}$. 
\end{remark}

The proof of Theorem \ref{theorem:C_approximation_Wishart} will proceed in three steps: first we provide the result for matrix $\tilde{\C}(\nu)$, then control the deviations between $\diag(\S(\nu))^{-\frac{1}{2}}$ and $\diag(\hat{\S}(\nu))^{-\frac{1}{2}}$, and finally extend the stochastic representation of $\tilde{\C}(\nu)$ to $\hat{\C}(\nu)$.

\subsection{Step 1: Stochastic representation of \texorpdfstring{$\tilde{\C}$}{C tilde}}
In order to establish (\ref{equation:tildeC_approximation_Wishart}), we prove the 
following Proposition. 
\begin{proposition}
\label{proposition:X_Gamma}
Under Assumptions  \ref{assumption:gaussian_y_n}, \ref{assumption:H0_independence}, \ref{assumption:rate_NBM} and \ref{assumption:regularity}, for any $\nu\in[0, 1]$, there exists an $M\times(B+1)$ random matrix $\X_N(\nu)$ with $\Ncal_\Cbb(0,1)$ i.i.d. entries, and another matrix $\Gammabs_N(\nu)$ such that:
\begin{equation}
\label{equation:Smcor_X_Gamma}
    \tilde{\C}_N(\nu)=\frac{(\X_N(\nu)+\Gammabs_N(\nu))(\X_N(\nu)+\Gammabs_N(\nu))^*}{B+1}
\end{equation}
where the family of random variables $\frac{\|\Gammabs_N(\nu)\|^2}{B+1}, \nu \in [0,1] $ satisfies 
\begin{equation}
    \label{eq:domination-norm-Gamma}
 \frac{\|\Gammabs_N(\nu)\|^2}{B+1} \prec \frac{B^2}{N^2} .  
\end{equation}
\end{proposition}

\begin{proof}
Denote by $\Sigmabs$ the $M\times(B+1)$ random matrix defined by
\begin{equation}
    \label{equation:definition_Sigma}
    \Sigmabs=\left(\xibs_\y(\nu-\frac{B}{2N}),\ldots,\xibs_\y(\nu+\frac{B}{2N})\right)
\end{equation}
where we recall that the normalized Fourier transform $\xibs_\y$ is defined in \eqref{eq:def-fourier-transform-xi}, so that $\hat{\S}$ defined in \eqref{eq:def-spectral-estimate} is equal to $\Sigmabs\Sigmabs^*/(B+1)$. Denote by $\omegabs_m$ the $m$--th row of $\Sigmabs$. In other words, $\omegabs_m$ coincides with the $(B+1)$--dimensional Gaussian complex row vector defined by:
$$\omegabs_m = \left(\xi_{y_m}(\nu-\frac{B}{2N}), \ldots ,\xi_{y_m}(\nu + \frac{B}{2N})\right).$$

The covariance matrix $\Ebb[\omegabs_m^{*} \omegabs_m]$ of $\omegabs$ is given by:
$$
 \Ebb[\omegabs_m^*\omegabs_m] = \Ebb\left[\left\{\xi_{y_m}(\nu+\frac{b_1}{N})^*\xi_{y_m}(\nu+\frac{b_2}{N})\right\}_{b_1,b_2=-B/2}^{B/2}\right]
$$

By Lemma \ref{lemma:brillinger_uniformity_covariance} in Appendix, we have for $b$ and $b_1\neq b_2$:
\begin{align*}
    &\Ebb\left[\left|\xi_{y_m}\left(\nu+\frac{b}{N}\right)\right|^2\right] = s_m\left(\nu+\frac{b}{N}\right) + \Ocal\left(\frac{1}{N}\right) \\
    &\Ebb \left[\xi_{y_m}(\nu+\frac{b_1}{N})^*\xi_{y_m}(\nu+\frac{b_2}{N})\right] = \Ocal\left(\frac{1}{N}\right) 
\end{align*}
where the error is uniform over $m\ge1$ and $\nu \in [0,1]$. Therefore one can claim that there exists some Hermitian matrix $\Upsilonbs_m(\nu)$ and some nice constant $C$ such that:
\[ \Ebb[\omegabs_m^*\omegabs_m] = \diag\left(s_m\left(\nu+\frac{b}{N}\right): b=-B/2,\ldots,B/2\right)+\Upsilonbs_m\]
where $\Upsilonbs_m$ satisfies
\[ \sup_{m\ge1,b_1,b_2}\left|\left(\Upsilonbs_m\right)_{b_1,b_2}\right| \le \frac{C}{N}. \]

Moreover, the regularity of the mappings $\nu\mapsto s_m(\nu)$ specified in Assumption \ref{assumption:regularity} implies that there exists quantities $\epsilon_m$ such that:
$$
s_m(\nu + \frac{b}{N}) = s_m(\nu) + s_m'(\nu) \, \frac{b}{N} + \frac{1}{2} \, s_m''(\nu) \, (\frac{b}{N})^{2} + \epsilon_m(\nu + \frac{b}{N})
$$
where: 
$$\sup_{m\ge1}\sup_{-B/2\le b\le B/2}|\epsilon_m(\nu + \frac{b}{N})| \le C\,\left(\frac{B}{N}\right)^{3}$$ 
for some nice constant $C$. Therefore, it holds that 
\begin{multline*}
\diag\left(s_m(\nu + \frac{b}{N}): b=-B/2, \ldots, B/2 \right) \\ = s_m(\nu) \I_{B+1} + s_m'(\nu) \;  \diag\left( \frac{b}{N}: b=-B/2, \ldots, B/2 \right) + \\
\frac{1}{2} \, s_m''(\nu) \;  \diag\left( (\frac{b}{N})^{2}: b=-B/2, \ldots, B/2 \right) + \\
\diag\left( \epsilon_m(\nu + \frac{b}{N}):  b=-B/2, \ldots, B/2 \right).
\end{multline*}

If we define matrix $\Phibs_m$ as:
$$\Phibs_m = \frac{1}{s_m}\left[ \Upsilonbs_m + \diag\left( \s_m(\nu + \frac{b}{N}) -s_m(\nu): b=-B/2, \ldots, B/2 \right) \right]$$ 
then $\Ebb[\omegabs_m^{*} \omegabs_m] = s_m \left( \I_{B+1} + \Phibs_m \right)$ with 
\begin{equation}
\label{eq:expre-cov-omegam}
    \sup_{m\ge1, b_1\neq b_2}|(\Phibs_{m})_{b_1,b_2}| \le \frac{C}{N}, \quad  \sup_{m\ge1, b}|(\Phibs_{m})_{b,b}| \le \frac{CB}{N}
\end{equation}
as well as 
\begin{equation}
    \label{eq:expre-Trace-Phim}
    \frac{1}{B+1} \Tr \Phibs_{m} = \frac{1}{2} \frac{s_m''(\nu)}{s_m(\nu)} v_N + \Ocal\left(\left(\frac{B}{N}\right)^{3} + \frac{1}{N} \right) 
\end{equation}
where we recall that $v_N$ is defined by (\ref{eq:def-vN}). 
The spectral norm of $\Phibs_m$ can be roughly bounded by the following inequality:
$$\sup_{m\ge1}\|\Phibs_m\| \le \sup_{m\ge1}\sup_{-B/2\le b_1\le B/2}\sum_{b_2=-B/2}^{B/2}\left|(\Phibs_m\right)_{b_1,b_2}|\le C\frac{B}{N}.$$
Moreover, it is easily checked that the Frobenius norm of $\frac{\Phibs_m}{B+1}$ satisfies
\begin{equation}
    \label{eq:frobenius-norm-phim}
    \left\|\frac{\Phibs_m}{B+1}\right\|_{F} \leq C \, \frac{\sqrt{B}}{N} = \Ocal(u_N) .  
\end{equation}

Using the Gaussianity of the vector $\omegabs_m$ and the expression (\ref{eq:expre-cov-omegam}), we obtain that $\omegabs_m$ can be represented as 
\begin{equation}
\label{equation:representation-etam}
\omegabs_m = \sqrt{s_m} \,  \x_m \, \left( \I + \Phibs_m \right)^{1/2}, \quad \x_m \sim \Ncal_\Cbb(0,I_{B+1})
\end{equation}
where $\x_{m_1}$ and $\x_{m_2}$ are independent for $m_1\neq m_2$. This comes from the mutual independence of the time series $((y_{m,n})_{n\in\Zbb})_{m=1, \ldots, M}$. It is clear that 
$\left( \I + \Phibs_m \right)^{1/2}$ can be written as 
\begin{equation}
    \label{eq:expre-IplusPhi1/2}
 \left( \I + \Phibs_m \right)^{1/2} = \I + \Psibs_m   
\end{equation}
where the matrix $\Psibs_m$ satisfies
\begin{equation}
    \label{eq:controle-norm-Psi}
    \sup_{m} \| \Psibs_m \| \leq C \, \frac{B}{N}
\end{equation}
Therefore, it holds that:
$$ \omegabs_m = \sqrt{s_m} \x_m \left( \I +  \Psibs_m \right) = \sqrt{s_m} \left(\x_m + \x_m \Psibs_m \right)  $$
We denote by  $\X$ and $\Gammabs$
the $M \times (B+1)$ matrices with rows $(\x_m)_{m=1, \ldots, M}$, and $\left( \x_m \Psibs_m \right)_{m=1, \ldots, M}$ respectively. Then, it holds that 
\begin{equation}
\label{eq:representation-Sigma}
\Sigmabs = \diag\left( \sqrt{s_m}, m=1, \ldots, M \right) \left( \X + \Gammabs \right)
\end{equation}
where we recall that $\Sigmabs$ is defined by \eqref{equation:definition_Sigma}. We recall the definition of the matrix $\tilde{\C}$ given by 
\begin{eqnarray}
    \label{eq:useful-expression-tildeC}
 \tilde{\C} & =  & \diag(\sqrt{s_m}, m=1,\ldots,M)^{-1/2}  \; \hat{\S} \; \diag(\sqrt{s_m}, m=1,\ldots,M)^{-1/2}   \\
 & =  & \diag(\sqrt{s_m}, m=1,\ldots,M)^{-1/2} \frac{\Sigmabs \Sigmabs^*}{B+1}\diag(\sqrt{s_m}, m=1,\ldots,M)^{-1/2}.
 \nonumber
 \end{eqnarray}
The representation \eqref{eq:representation-Sigma} implies that $\tilde{\C}$ can also be written as 
$$\tilde{\C} = \frac{(\X+\Gammabs)(\X+\Gammabs)^*}{B+1}. $$
Equivalently, for each $m_1,m_2$, the entry $(\tilde{\C})_{m_1,m_2}$ is given by 
\begin{equation}
    \label{eq:expre-entries-tildeC}
    (\tilde{\C})_{m_1,m_2} = \frac{1}{B+1} \x_{m_1} (\I + \Phibs_{m_1})^{1/2} (\I + \Phibs_{m_2})^{1/2} \x_{m_2}^{*}.
\end{equation}

This completes the proof of (\ref{equation:Smcor_X_Gamma}). It remains to show \eqref{eq:domination-norm-Gamma}. We denote by $\Z$ the $M \times M$ matrix 
$\Z = \frac{1}{B+1}\Gammabs\Gammabs^*$. As $\|\Z\|$ satisfies
$$ \|\Z\| \le \|\Z-\Ebb\Z\| + \|\Ebb\Z\| $$
it is enough to prove the two following facts:
\begin{equation}
    \label{eq:esperance-Z}
    \|\Ebb\Z\| \leq C \, \frac{B^2}{N^2}
\end{equation}
\begin{equation}
    \label{eq:domination-stochastique-Z}
    \|\Z-\Ebb\Z\| \prec \frac{B^2}{N^2}.
\end{equation}

We start with \eqref{eq:esperance-Z}. The definition of $\Gammabs$ leads to
$$
    \Ebb[\Z_{i,j}] = \frac{1}{B+1}\Ebb[\Gammabs\Gammabs^*]_{i,j} 
    =  \frac{1}{B+1}\Ebb[\x_i \Psibs_i \Psibs_j^*\x_j^*] 
    = \delta_{ij}\frac{1}{B+1}\Tr\Psibs_i\Psibs_j^*
$$
so that it is clear that $\Ebb[\Z]$ is the diagonal matrix with diagonal entries $(\frac{1}{B+1} \Tr \Psibs_m \Psibs_m^*)_{m=1, \ldots, M} $. By the estimation in equation (\ref{eq:controle-norm-Psi}), we easily have \eqref{eq:esperance-Z}. 

It remains to prove \eqref{eq:domination-stochastique-Z}. We use the observation 
that $\|\Z-\Ebb[\Z]\| = \max_{\| {\bf h} \| = 1} \left| {\bf h}^{*}(\Z-\Ebb[\Z]){\bf h} \right|$, and use a classical $\epsilon$--net argument that allows to 
deduce the behaviour of  $\|\Z-\Ebb[\Z]\|$ from the behaviour of any recentered quadratic 
form $\g^*\Z\g -\Ebb\g^*\Z\g$ where  $\g\in\Cbb^M$ is a deterministic unit norm vector. We thus first concentrate $\g^*\Z\g -\Ebb\g^*\Z\g$ using the Hanson-Wright inequality \eqref{equation:hanson_wright_stochastic_domination}. For this, we need to express $\g^*\Z\g$ as a quadratic form of a certain complex Gaussian random vector with i.i.d. entries. We denote by $\z$ the $M$--dimensional random vector  $\z=\frac{\Gammabs^*(\nu)\g}{\sqrt{B+1}}$. Its covariance matrix $\G = \G(\nu)$ is equal to 
$$\G(\nu)=\Ebb[\z\z^*]=\frac{1}{B+1} \sum_{m=1}^{M} |\g_m|^{2}  (\Psibs_m(\nu))^* \Psibs_m(\nu).$$
Therefore, $\z$ can be written as $\z=\G^{1/2}\w$ for some $\w\sim\Ncal_\Cbb(0,{\bf I}_M)$ random vector. As a consequence, the quadratic form $\g^*\Z\g -\Ebb\g^*\Z\g$ can be written as
$$ \g^*\Z\g -\Ebb\g^*\Z\g = \w^*\G\w-\Ebb\w^*\G\w.$$

The Hanson-Wright inequality \eqref{equation:hanson_wright_stochastic_domination} can now be applied:
\begin{equation}
    \label{equation:hanson_wGw}
    |\w^*\G\w-\Ebb\w^*\G\w| \prec \|\G\|_F.
\end{equation}

Since $\sum_{m=1}^M|\g_m|^2=1$, it is clear that $\|\G\|\le\frac{1}{B+1} \sup_{m=1,\ldots,M} \|  \Psibs_m(\nu) \|^2$.
Therefore, \eqref{eq:controle-norm-Psi} and the rough bound $\|\G\|^2_F\le (B+1) \|\G\|^2$ leads to
\begin{equation}
\label{equation:estimees_norme_G}
    \|\G\| \leq C \, \frac{1}{B+1}\left(\frac{B}{N}\right)^2, \quad \|\G\|^2_F \leq C \, \frac{1}{B+1} \left(\frac{B}{N}\right)^4
\end{equation}

The substitution of \eqref{equation:estimees_norme_G} in equation \eqref{equation:hanson_wGw} gives the following control of 
$\g^*\Z\g -\Ebb\g^*\Z\g$:
\begin{equation}
    \label{equation:domination_polynomiale}
    |\g^*\Z\g -\Ebb\g^*\Z\g|\prec \frac{1}{\sqrt{B}}\left(\frac{B}{N}\right)^2
\end{equation}

Consider $\epsilon > 0$, and an $\epsilon$--net $N_\epsilon$ of $\mathbb{C}^{M}$, that is a set of $\Cbb^M$ unit norm vectors $\{\h_k:k=1,\ldots,\mathcal{K}\}$  such that for each unit norm vector $\u\in\Cbb^{M}$, there exists a vector $\h\in N_\epsilon$ for which $\|\u-\h\|\le\epsilon$. It is well 
known that the cardinality of $N_\epsilon$ is bounded by $C_0 \left(\frac{1}{\epsilon}\right)^{2M}$ where $C_0$ is a universal constant. Then, denote $\g_s$ a (random) unit norm vector 
such that $|\g_s^*\Z\g_s -\Ebb\g_s^*\Z\g_s| = \| \Z - \Ebb \Z \|$, and define $\h_s\in\N_\epsilon$ as the closest vector from $\g_s$. Therefore, we have
\begin{align*}
    \|\Z-\Ebb\Z\| &= |\g_s^*(\Z-\Ebb\Z)\g_s| \\
    &= |(\g_s^*-\h_s^*+\h_s^*)(\Z-\Ebb\Z)(\g_s-\h_s+\h_s)| \\
    &\le |(\g_s^*-\h_s^*)(\Z-\Ebb\Z)(\g_s-\h_s)|+ |(\g_s^*-\h_s^*)(\Z-\Ebb\Z)\h_s|\\
    &\hskip2cm+|\h_s^*(\Z-\Ebb\Z)(\g_s-\h_s)|+ |\h_s^*(\Z-\Ebb\Z)\h_s|.
\end{align*}
It is clear that:
\begin{align*}
    &|(\g_s^*-\h_s^*)(\Z-\Ebb\Z)(\g_s-\h_s)| \le \epsilon^2\|\Z-\Ebb\Z\| \\
    &|(\g_s^*-\h_s^*)(\Z-\Ebb\Z)\h_s|\le \epsilon\|\Z-\Ebb\Z\|
\end{align*}
and 
$$ \|\Z-\Ebb\Z\|\le|\h_s^*(\Z-\Ebb\Z)\h_s|+\epsilon^2\|\Z-\Ebb\Z\|+2\epsilon\|\Z-\Ebb\Z\|$$
which leads to 
$$ (1-2\epsilon-\epsilon^2)\|\Z-\Ebb\Z\| \le |\h_s^*(\Z-\Ebb\Z)\h_s|.$$

This implies that for each $t > 0$, 
$$
\{ \|\Z-\Ebb\Z\|  > t \} \subset \cup_{ h \in N_\epsilon} \{ |\h^*(\Z-\Ebb\Z)\h| > C_1 t \}
$$
where $C_1 = (1-2\epsilon-\epsilon^2)$. Using the union bound, we obtain that 
\begin{equation}
\label{eq:union-bound-epsilon-net}
\Prob \left[  \|\Z-\Ebb\Z\|  > t \right] \leq \sum_{h \in N_\epsilon} 
\Prob \left[ |\h^*(\Z-\Ebb\Z)\h| > C_1 t \right].
\end{equation}
Here, we would like to use equation \eqref{equation:domination_polynomiale}. By the definition of $\prec$, \eqref{equation:domination_polynomiale} is valid uniformly on any set of vector with cardinality polynomial in $N$. Here, the cardinality of the set $N_\epsilon$ is a $\Ocal(\epsilon^{-2M})$ term and therefore exponential in $M$. As a consequence, we have to accept to lose some speed when going from the stochastic domination of $|\g^*(\Z-\Ebb\Z)\g|$ for a fixed $\g$ to the same stochastic domination but uniformly over $N_\epsilon$.

More specifically, write again \eqref{equation:domination_polynomiale} but here without the notation $\prec$ in order to understand precisely how a change in speed affects the probability. Take $t_N$ a sequence of positive numbers such that $t_N\ge B^2/N^2$. Using the estimates \eqref{equation:estimees_norme_G} of $\|\G\|$ and $\|\G\|_F^2$, and the fact that $\min(a_1,a_2)>\min(b_1,b_2)$
when $a_1>b_1$ and $a_2 > b_2$, we obtain that there exists some nice constant $C>0$ such that:
\begin{multline*}
    \min\left(\frac{t_N}{\|\G\|}, \frac{t_N^2}{\|\G\|_F^{2}}\right) \ge C\, B\min\left(t_N \left(\frac{N}{B}\right)^2, \left(t_N\left(\frac{N}{B}\right)^2\right)^2 \right) \\ = C\, B\, t_N \left(\frac{N}{B}\right)^2.
\end{multline*} 
The Hanson-Wright inequality \eqref{eq:hanson-wright} provides:
$$ \Prob\left[|\g^*\left[\Z-\Ebb\Z)\right]\g|>C_1 t_N\right]\le 2 \, \exp\left\{-CB \frac{t_N}{(B/N)^2}\right\} $$
for some nice constant $C$ that depends on $C_1$. 
Finally, the union bound on $N_\epsilon$ gives:
\begin{multline}
    \label{eq:concentration-normZ-EZ}
    \Prob \left[ \|\Z - \mathbb{E}(\Z) \| >  t_N \right] \leq \sum_{h \in N_\epsilon} \Prob \left[ |\h^*(\Z-\Ebb\Z)\h| > C_1 t_N \right] \\ \le 2 \, C_0 \exp\left\{-CB \frac{t_N}{(B/N)^2} + 2M\log\frac{1}{\epsilon}\right\} .
\end{multline}
If we take $t_N=N^{\epsilon'} (B^2/N^2)$, then, there exists $\gamma > 0$ such that 
$$ \exp\left\{-CB \frac{t_N}{(B/N)^2} + 2CM\log\frac{1}{\epsilon}\right\} \le \exp-N^\gamma$$ 
holds for each $N$ large enough. (\ref{eq:union-bound-epsilon-net}) thus implies  \eqref{eq:domination-stochastique-Z}. This completes the proof of \eqref{equation:Smcor_X_Gamma}.
\end{proof}

Corollary \ref{corollary:concentration_delta} is a rewriting of Proposition \ref{proposition:X_Gamma} in a more concise way. Define:
\begin{equation}
    \label{definition:tilde_Deltabs}
    \tilde{\Deltabs}=\frac{\X\Gammabs^*+\Gammabs\X^*+\Gammabs\Gammabs^*}{B+1}.
\end{equation}

\begin{corollary}
\label{corollary:concentration_delta}
For any $\nu\in[0, 1]$, $\tilde{\C}(\nu)$ can be written as 
\begin{equation}
    \label{eq:representation-tildeC}
    \tilde{\C}(\nu)=\frac{\X(\nu)\X^*(\nu)}{B+1}+\tilde{\Deltabs}(\nu)
\end{equation}
where the family of random variable $\|\tilde{\Deltabs}(\nu)\|, \, \nu \in [0,1]$ satisfies
\begin{equation}
\label{eq:concentration-tilde-Delta}
\quad \|\tilde{\Deltabs}\| \prec \frac{B}{N}.
\end{equation}
\end{corollary}

\begin{proof}
Let $\nu\in[0, 1]$. 
By equation \eqref{eq:domination-norm-Gamma} from Theorem \ref{theorem:C_approximation_Wishart} and equation \eqref{eq:stoc-domination-eig-sing-wishart} from Paragraph \ref{section:haagerup}, we have the two following estimates: $$\frac{\|\Gammabs\|}{\sqrt{B+1}} \prec \frac{B}{N}, \quad  \frac{\|\X\|}{\sqrt{B+1}} \prec 1$$
The result is immediate using decomposition $\tilde{\Deltabs}$ from \eqref{definition:tilde_Deltabs}:
\end{proof}

We now take benefit of Corollary \ref{corollary:concentration_delta} to establish 
the first part of Corollary \ref{coro:localisation-eigenvalues-hatC-tildeC} and to 
analyse the location of the eigenvalues of matrices $\hat{\S}$. We denote by $\D$ and $\hat{\D}$ the matrices $\D = \D(\nu):=\diag(\S(\nu))^{\frac{1}{2}}$ and $\hat{\D} = \hat{\D}(\nu):=\diag(\hat{\S}(\nu))^{\frac{1}{2}}$. Denote by $\bar{s}$ and $\barbelow{s}$ the quantities such that:
$$ \barbelow{s}:=\inf_{m\ge1}\inf_{\nu\in[0,1]}s_m(\nu), \quad \bar{s}:=\sup_{m\ge1}\sup_{\nu\in[0,1]}s_m(\nu) $$
which are by Assumption \ref{assumption:regularity} in $(0,+\infty)$. We consider the event:
\begin{equation}
    \label{equation:definition_Lambda_S} \Lambda_\epsilon^{\hat{\S}}(\nu)=\left\{\sigma(\hat{\S}(\nu))\subset\Supp\mu_{MP}^{(c)}\times[\barbelow{s},\bar{s}]+\epsilon\right\} 
\end{equation}
where the notation $\Supp\mu_{MP}^{(c)}\times[\barbelow{s},\bar{s}]$ stands for $[(1-\sqrt{c})^2\barbelow{s},(1+\sqrt{c})^2\bar{s}]$.
Note that in our settings, $c\in(0,1)$ so $\Supp\mu_{MP}^{(c)}$ is bounded and away from zero. In conjunction with Assumption \ref{assumption:regularity}, the same holds for $\Supp\mu_{MP}^{(c)}\times[\barbelow{s},\bar{s}]$. We also note that $\Lambda_\epsilon^{\hat{\S}}(\nu)$ of course depends on $N$. 
\begin{corollary}
\label{corollary:spectre-tildeC-spectre-hatS}
For any $\epsilon>0$, the families of events $\Lambda^{\tilde{\C}}_\epsilon(\nu)$, $\nu\in[0,1]$ and $\Lambda_\epsilon^{\hat{\S}}(\nu)$, $\nu\in[0,1]$ hold with exponentially high probability.
\end{corollary}

\begin{proof}
Equation (\ref{eq:representation-tildeC}) implies that 
$$ \frac{\X\X^*}{B+1}-\|\tilde{\Deltabs}\|\I_M\le\tilde{\C}\le\frac{\X\X^*}{B+1}+\|\tilde{\Deltabs}\|\I_M.$$
Therefore, the event $\{ \lambda_1(\tilde{\C}) > (1 +\sqrt{c})^{2} + \epsilon \}$ is included in
$\{ \lambda_1(\frac{\X\X^*}{B+1}) + \| \tilde{\Deltabs} \| >  (1 +\sqrt{c})^{2} + \epsilon \}$, which is itself included in 
$$
\left\{ \lambda_1(\frac{\X\X^*}{B+1}) >  (1 +\sqrt{c})^{2} + \epsilon/2 \right\} \cup \left\{  \| \tilde{\Deltabs} \| > \epsilon/2 \right\}.
$$
Therefore, 
\begin{multline*}
    \Prob\left[  \lambda_1(\tilde{\C}) > (1 +\sqrt{c})^{2} + \epsilon \right] \leq \Prob\left[ \lambda_1(\frac{\X\X^*}{B+1}) >  (1 +\sqrt{c})^{2} + \epsilon/2 \right] \\ + \Prob\left[ \| \tilde{\Deltabs} \| > \epsilon/2 \right].
\end{multline*}
Equations (\ref{eq:concentration-largest-eig-wishart}) and (\ref{eq:concentration-tilde-Delta}) 
imply that $\Prob\left[  \lambda_1(\tilde{\C}) > (1 +\sqrt{c})^{2} + \epsilon \right]$ 
converges towards 0 exponentially. A similar evaluation of 
$\Prob\left[  \lambda_M(\tilde{\C}) < (1 -\sqrt{c})^{2} - \epsilon \right]$ leads to the same 
conclusion. This, in turn, establishes that  $\Lambda^{\tilde{\C}}_\epsilon(\nu), \nu \in [0,1]$ holds with exponential high probability. 

In order to establish that the same property holds for  $\Lambda_\epsilon^{\hat{\S}}(\nu), \nu \in [0,1]$, we just need to write \eqref{eq:def-tildeC} as  $ \hat{\S} = \D^{1/2}\tilde{\C}\D^{1/2} $. 
Therefore, for each $k=1, \ldots, M$, the eigenvalues of $\hat{\S}$ satisfy 
$$
\barbelow{s} \, \lambda_{M}(\tilde{\C}) \leq \lambda_k(\hat{\S}) \leq \bar{s} \, 
 \lambda_{1}(\tilde{\C}).
$$
This, of course, implies that $\Lambda_\epsilon^{\hat{\S}}(\nu), \nu \in [0,1]$ holds 
with  exponential high probability
(indeed, one can change $\epsilon$ to $\tilde{\epsilon}$ such that $(\Supp\mu_{MP}^{(c)}+\tilde{\epsilon})\times[\barbelow{s},\bar{s}] \subset \Supp\mu_{MP}^{(c)}\times[\barbelow{s},\bar{s}]+\epsilon$.

\end{proof}

\begin{remark}
Corollary \ref{corollary:spectre-tildeC-spectre-hatS} implies the following weaker property, which will be useful: 
\begin{equation}
    \label{eq:domination-norm-hatS}
    \|\hat{\S}(\nu)\|\prec 1 .
\end{equation}
\end{remark}

Before ending the section and proving Theorem \ref{theorem:C_approximation_Wishart}, we need some stochastic control on the diagonal elements of $\hat{\S}$ 
in order to evaluate $\Thetabs$ defined by
\begin{equation}
\label{equation:definition_Theta}
    \Thetabs := \hat{\C} - \tilde{\C} .
\end{equation} 
Using the definition of $\hat{\C}$ from \eqref{equation:definition_coherency} and $\tilde{\C}$ from \eqref{eq:def-tildeC}, $\Thetabs$ can be written as
\begin{align}
\label{equation:decomposition_Theta}
    \Thetabs = (\hat{\D}^{-1/2}-\D^{-1/2})\hat{\S}\hat{\D}^{-1/2} + \D^{-1/2}\hat{\S}(\hat{\D}^{-1/2}-\D^{-1/2}).
\end{align}
Since we proved  that $\|\hat{\S}\|\prec1$, it remains to show that $\|\hat{\D}^{-1/2}\|$ and $\|\hat{\D}^{-1/2}-\D^{-1/2}\|$ can also be stochastically dominated by some relevant quantity in order to control $\| \Thetabs \|$. Define 
\begin{equation}
\label{eq:def-hatsm}
 \hat{s}_m(\nu):=\hat{\S}_{m,m}(\nu)
 \end{equation}
the diagonal elements of $\hat{\S}(\nu)$ spectral density estimator (note that they coincide with the traditional smoothed periodogram estimator of the spectral density $s_m$). The aim of the following Paragraph \ref{paragraph:domination_stochastique_s} is to establish stochastic domination results for $\hat{s}_m$, $\|\hat{\D}^{-1/2}\|$ and $\|\hat{\D}^{-1/2}-\D^{-1/2}\|$. 

\subsection{Step 2: Estimates for \texorpdfstring{$\hat{s}_m(\nu)$}{s m tilde (nu)} }
\label{paragraph:domination_stochastique_s}
We write $s_m(\nu):=s_m$, $\D(\nu):=\D,$ in order to simplify the 
notations. Define as in \eqref{equation:definition_Lambda_S} the following quantity
\begin{equation}
    \label{equation:definition_Lambda_D}
    \Lambda_\epsilon^{\hat{\D}}(\nu)=\{\sigma(\hat{\D}(\nu))\subset [\barbelow{s},\bar{s}]+\epsilon\}.
\end{equation}

\begin{lemma}
\label{lemma:localization_s_m}
Let $\epsilon>0$. The family of events $\Lambda_\epsilon^{\hat{\D}}(\nu), \nu\in[0, 1]$ holds with exponentially high probability.
\end{lemma}

\begin{proof}
See Appendix \ref{appendix:concentration_sm}. 
\end{proof}

Roughly speaking, this ensures that with exponentially high probability, $\hat{s}_m$ stays bounded and away from zero. This result implies the following (weaker) statement, but will still be enough for some proofs and reduces the complexity of the arguments. 

\begin{lemma}
\label{lemma:concentration_sm_inf_sup}
The family of random variables $(|\hat{s}_m(\nu)|+ \frac{1}{|\hat{s}_m(\nu)|})_{m=1, \ldots, M}$, $\nu\in[0, 1]$, satisfies
$$ \left(|\hat{s}_m| + \frac{1}{|\hat{s}_m|}\right)\prec 1.$$
\end{lemma}

\begin{proof}
Immediate from Lemma \ref{lemma:localization_s_m}.
\end{proof}

\begin{lemma}
\label{lemma:concentration_ratio_s_shat}
The set of random variable $(|\hat{s}_m(\nu)^{-1/2} - s_m(\nu)^{-1/2}|)_{m=1,\ldots,M}$ and $(|\sqrt{\frac{s_m(\nu)}{\hat{s}_m(\nu)}}-1|)_{m=1,\ldots,M}$, $\nu\in[0, 1]$, satisfies
\begin{equation}
\label{equation:domination_D_racine_carre}
|\hat{s}_m^{-1/2} - s_m^{-1/2}| \prec \frac{1}{\sqrt{B}}+\frac{B^2}{N^2}, \quad  \left|\sqrt{\frac{s_m}{\hat{s}_m}}-1\right|\prec \frac{1}{\sqrt{B}}+\frac{B^2}{N^2}.
\end{equation}

\end{lemma}

\begin{proof}
See Appendix \ref{appendix:various}
\end{proof}

\subsection{Step 3: Stochastic representation of \texorpdfstring{$\hat{\C}$}{C hat}} 
We are now in a position to prove the result concerning $\hat{\C}$ of Theorem \ref{theorem:C_approximation_Wishart} and of Corollary \ref{coro:localisation-eigenvalues-hatC-tildeC}.
\begin{proof}
We have first to control the operator norm of:
\begin{equation}
    \label{eq:domination_stochastique_Thetabs}
    \Deltabs=\hat{\C}-\frac{\X\X^*}{B+1}=\hat{\C}-\tilde{\C}+\tilde{\C}-\frac{\X\X^*}{B+1}=\Thetabs+\tilde{\Deltabs}.
\end{equation}

The operator norm of $\|\tilde{\Deltabs}\|$ has already been proved in Corollary \ref{corollary:concentration_delta} to satisfy $\|\tilde{\Deltabs}\|\prec(\frac{B}{N})$. Moreover, recall that $\Thetabs$ can be written as a function of $\hat{\D}^{-1/2}-\D^{-1/2}$ in \eqref{equation:decomposition_Theta}, so that one can use Lemma \ref{lemma:concentration_sm_inf_sup} and Lemma \ref{lemma:concentration_ratio_s_shat} to dominate each term and get:
\begin{equation}
\label{equation:domination_Thetabs}
    \|\Thetabs\| \prec \frac{1}{\sqrt{B}}+\frac{B^2}{N^2}.
\end{equation}

Summing the estimate of $\Thetabs$ and the one of $\tilde{\Deltabs}$, one gets:
$$ \|\Deltabs\| \prec \frac{1}{\sqrt{B}}+ \frac{B}{N}$$
which is the desired result.
\end{proof}

As a consequence, we state here Corollary \ref{corollary:localization_C_hat} about the localization of the eigenvalues of $\hat{\C}(\nu)$. 

\begin{corollary}
\label{corollary:localization_C_hat}
For each $\epsilon > 0$, we define $\Lambda^{\hat{\C}}_\epsilon(\nu)$ as the event
\begin{equation}
    \label{definition:hat_Xi}
    \Lambda^{\hat{\C}}_\epsilon(\nu)=\left\{\sigma(\hat{\C}(\nu))\subset\Supp\mu_{MP}^{(c)}+\epsilon\right\}.
\end{equation}
Then, the family of events $\Lambda^{\hat{\C}}_\epsilon(\nu), \nu \in [0,1]$ holds with exponentially 
high probability. 
\end{corollary}

\begin{proof}
We simply write:
$$ \frac{\X\X^*}{B+1}-\|\Deltabs\|\I_M\le\hat{\C}\le\frac{\X\X^*}{B+1}+\|\Deltabs\|\I_M$$
and use the same arguments as in the proof of Corollary \ref{corollary:spectre-tildeC-spectre-hatS}. 
\end{proof}

\section{Stochastic domination of the family \texorpdfstring{$\psi_N(f,\nu), N \geq 1, \nu \in [0,1]$}{}}
\label{sec:lss}
We have first to define the distribution $D_N$ introduced in the definition (\ref{eq:def-psi(f)})
of $\psi_N(f,\nu)$. For this, we consider the function $p_N(z)$ defined by 
\begin{equation}
    \label{eq:def-pN}
    p_N(z) = - \frac{ c_N \, (z \, t_N(z) \, \tilde{t}_N(z))^{3}}{1 - c (z \, t_N(z) \, \tilde{t}_N(z))^{2}}
\end{equation}
where we recall that $t_N$ and $\tilde{t}_N$ are defined by (\ref{eq:equation-MP1})  and 
(\ref{eq:def-ttilde}). 
Then (see Lemma 9.2 in \cite{loubaton2016jotp}),  
$p_N$ is the Stieltjes transform of a distribution whose support is contained in the 
support $\Supp\mu_{MP}^{(c_N)} = [(1 - \sqrt{c}_N)^{2}, (1 + \sqrt{c}_N)^{2}]$ of the Marcenko-Pastur distribution $\mu_{MP}^{(c_N)}$. This distribution is $D_N$ introduced in (\ref{eq:def-psi(f)}). 
In the following, we consider LSS for function $f$ satisfying the following assumptions.
\begin{assumption}
\label{assumption:f}
$f$ is defined on $\Rbb_+$ and there exists some $\epsilon>0$ such that its restriction on $\Supp_{MP}^{(c)}+\epsilon$ is $\mathcal{C}^{\infty}$.
\end{assumption}
We now state the main result of this section.  
\begin{theorem}
\label{theo:domination-psi}
Let $f$ be a function satisfying the conditions of Assumption \ref{assumption:f}. Then, under Assumptions \ref{assumption:gaussian_y_n}, \ref{assumption:H0_independence},   \ref{assumption:rate_NBM} and  \ref{assumption:regularity}, the family  $|\psi_N(f,\nu)|, N \geq 1, \nu \in [0,1]$ satisfies 
\begin{equation}
\label{eq:stochastic-domination-psi}
|\psi_N(f,\nu)| \prec u_N.
\end{equation}
\end{theorem}
Before starting the proof of Theorem \ref{theo:domination-psi}, we first mention that it is sufficient
to establish (\ref{eq:stochastic-domination-psi}) when $f$ is compactly supported by a neighbourhood of $\Supp\mu_{MP}^{(c)}$. To justify this claim, we consider $\kappa>0$ and define $\chi: \Rbb\to\Rbb$ as a $\mathbb{C}^\infty$ function such that:
\begin{equation}   
\label{eq:definition_chi}
\chi(\lambda) = 
     \begin{cases}
       1 &\quad\text{if }  \lambda \in \Supp\mu_{MP}^{(c)}+\kappa \\
       0 &\quad\text{if }\lambda \notin \Supp\mu_{MP}^{(c)}+2\kappa.
     \end{cases}
\end{equation}  
We consider the function $\bar{f}$ given by $\bar{f}=f\times\chi$. Then, as $c_N \rightarrow c$, for $N$ large enough, $\Supp\mu_{MP}^{(c_N)}$ is contained in $\Supp\mu_{MP}^{(c)}+\kappa$. Therefore, $f = \bar{f}$ on 
 $\Supp\mu_{MP}^{(c_N)}$ for $N$ large enough, and it holds that $<D_N,f> = <D_N,\bar{f}>$ and $\int f d\mu_{MP}^{(c_N)} = \int \bar{f} d\mu_{MP}^{(c_N)}$. For each $\epsilon > 0$, we express $\Pbb(|\psi_N(f,\nu)| > N^{\epsilon} u_N)$ as 
\begin{align*}
 \Pbb(|\psi_N(f,&\nu)| > N^{\epsilon} u_N) \\
  &=  \Prob(|\psi_N(f,\nu)| > N^{\epsilon} u_N, \Lambda^{\hat{\C}}_\kappa(\nu)) + \Pbb(|\psi_N(f,\nu)| > N^{\epsilon} u_N, (\Lambda^{\hat{\C}}_\kappa(\nu))^{c}) \\
   & \leq  \Pbb(|\psi_N(f,\nu)| > N^{\epsilon} u_N, \Lambda^{\hat{\C}}_\kappa(\nu)) + \Pbb\left((\Lambda^{\hat{\C}}_\kappa(\nu))^{c}\right) \\
   & \leq  \Pbb(|\psi_N(\bar{f},\nu)| > N^{\epsilon} u_N) + \Pbb\left((\Lambda^{\hat{\C}}_\kappa(\nu))^{c}\right)
 \end{align*}
 where the last inequality follows from the observation that 
 $\frac{1}{M} \Tr f(\hat{\C}) = \frac{1}{M} \Tr \bar{f}(\hat{\C})$ on $\Lambda^{\hat{\C}}_\kappa(\nu)$. Moreover, the family of events $\Lambda^{\hat{\C}}_\kappa(\nu)$
 holds with exponential high probability, which implies that $\Pbb\left((\Lambda^{\hat{\C}}_\kappa(\nu))^{c}\right)$ converges towards $0$ exponentially fast. Therefore, $|\psi_N(\bar{f},\nu)| \prec u_N$ 
 implies (\ref{eq:stochastic-domination-psi}) as expected. From now on, we thus assume that the function $f$ is supported by $\Supp\mu_{MP}^{(c)}+2\kappa$ \\

In order to establish (\ref{eq:stochastic-domination-psi}), we evaluate the four terms of the righhandside of (\ref{eq:fundamental-decomposition}).

\subsection{Step 1: Evaluation of \texorpdfstring{$\mathbb{E} \left[  \frac{1}{M} \mathrm{Tr}\left( f(\frac{\X_N(\nu)\X_N^*(\nu)}{B+1}) \right) \right] - \int_{\Rbb^{+}}f\diff\mu_{MP}^{(c_N)}$}{}}
We evaluate this term using the Helffer-Sjöstrand formula. We keep the notations 
of paragraphs \ref{subsec:hs-formula} and \ref{section:haagerup}: we assume that the support of $f$ is included in 
$[a_1, a_2]$ with  $a_1 = (1-\sqrt{c}_N)^{2} - 2 \kappa$ and $a_2 = (1+\sqrt{c}_N)^{2} + 2 \kappa$.
Moreover, the resolvent of the matrix $\frac{\X_N \X_N^*}{B+1}$ is denoted $\Q_N(z)$ (we omit to mention that the matrices depend on $\nu$), and $\beta_N(z)$ represents $\mathbb{E}((\Q_N(z))_{mm})$ 
for each $m$. We also denote by $\epsilon_N(z)$ the error term defined by (\ref{eq:beta-t})
which satisfies $|\epsilon_N(z)| \leq \frac{1}{M^{2}} P_1(|z|) P_2(\frac{1}{\Im z})$ on $\mathbb{C}^{+}$ for some nice polynomials $P_1$ and $P_2$. Then, for $k \geq \mathrm{deg}(P_2)$, it holds that 
\begin{multline*}
\mathbb{E} \left[  \frac{1}{M} \mathrm{Tr}\left( f(\frac{\X(\nu)\X^*(\nu)}{B+1}) \right) \right] - \int_{\Rbb^{+}}f\diff\mu_{MP}^{(c_N)} \\ = \frac{1}{\pi}\Re\int_{\mathcal{D}} \,\bar{\partial}\Phi_k(f)(z)
      (\beta_N(z) - t_N(z)) \diff x \diff y
\end{multline*}
where $\mathcal{D}$ is defined as in paragraph \ref{subsec:hs-formula}. 
$\int_{\mathcal{D}} \,
\left| \bar{\partial}\Phi_k(f)(z) \right|  P_1(|z|) P_2(\frac{1}{\Im z}) \diff x \diff y$
is finite, and by \eqref{eq:beta-t}, the following bound holds:
\begin{multline*}
\left| \mathbb{E} \left[  \frac{1}{M} \mathrm{Tr}\left( f(\frac{\X(\nu)\X^*(\nu)}{B+1}) \right) \right] - \int_{\Rbb^{+}}f\diff\mu_{MP}^{(c_N)} \right| \leq  \\ 
\frac{1}{M^{2}} \int_{\mathcal{D}} \,
\left| \bar{\partial}\Phi_k(f)(z) \right|  P_1(|z|) P_2(\frac{1}{\Im z}) \diff x \diff y \leq \frac{C}{B^{2}}
\end{multline*}
for some nice constant $C$. We have therefore established the following result. 
\begin{lemma}
\label{le:evaluation-terme4}
there exists a nice constant $C$ such that, for each $\nu$,
\begin{equation}
    \label{eq:control-terme4}
\left| \mathbb{E} \left[  \frac{1}{M} \mathrm{Tr}\left( f(\frac{\X(\nu)\X^*(\nu)}{B+1}) \right) \right] - \int_{\Rbb^{+}}f\diff\mu_{MP}^{(c_N)} \right| \leq \frac{C}{B^{2}}.
\end{equation}
\end{lemma}

\subsection{Step 2: Evaluation of \texorpdfstring{$\frac{1}{M} \mathrm{Tr}\left( f(\tilde{\C}(\nu))\right) - \mathbb{E} \left[ \frac{1}{M} \mathrm{Tr}\left( f(\tilde{\C}(\nu))\right) \right]$}{}}

In order to evaluate the above term, we use the Gaussian concentration inequality introduced in Paragraph \ref{subsection:lipschitz_concentration}. We recall that $\tilde{\C}$ can be interpreted as a function of $(\X,\X^{*})$ (see \eqref{eq:expre-entries-tildeC})). Therefore, 
$\frac{1}{M} \mathrm{Tr}\left( f(\tilde{\C}(\nu))\right)$ can be written as $g(\X,\X^*)$
for some real valued function $g$. We establish in the following that $g$ is $\Ocal(\frac{1}{B})$--Lipschitz, which in turn, will imply that 
\begin{equation}
    \label{eq:concentration-f(tildeC)}
    \left|\frac{1}{M} \mathrm{Tr}\left( f(\tilde{\C}(\nu))\right) - \mathbb{E} \left[ \frac{1}{M} \mathrm{Tr}\left( f(\tilde{\C}(\nu))\right) \right]\right| \prec \frac{1}{B} .
\end{equation}
For this, we evaluate 
\begin{equation}
\label{norme_lipschitz}
\|\nabla g(\X,\X^{*})\|^2=\sum_{i,j}\left|\frac{\partial g}{\partial X_{i,j}}\right|^2 + \left|\frac{\partial g}{\partial\overline{X_{i,j}}}\right|^2 = 2 \, \sum_{i,j}\left|\frac{\partial g}{\partial X_{i,j}}\right|^2.
\end{equation}

Using classic identities for the derivation of Hermitian matrices, we obtain that 
$$
\frac{1}{M} \frac{\partial \, \Tr f(\tilde{\C})}{\partial X_{ij}} = \frac{1}{M} \Tr \left(f'(\tilde{\C}) \frac{\partial \tilde{\C}}{\partial X_{ij}} \right) 
$$
Straightforward calculations lead to 
$$
\sum_{i,j}\left|\frac{\partial g}{\partial X_{i,j}}\right|^2 = 
\frac{1}{M^{2}(B+1)^{2}} \sum_{i=1}^{M}  \left( f'(\tilde{\C}) (\X + \Gammabs) (\I + \Phibs_i)   (\X + \Gammabs)^{*} f'(\tilde{\C}) \right)_{ii}.
$$
Using  $\sup_i \| \I + \Phibs_i \| \leq C$ for some nice constant $C$ as well as $\tilde{\C} = \frac{1}{B+1}(\X + \Gammabs) (\X + \Gammabs)^{*}$, we obtain immediately that 
$$
\sum_{i,j}\left|\frac{\partial g}{\partial X_{i,j}}\right|^2 \leq \frac{C}{B^{2}} \, \frac{1}{M} \Tr\left( f'^2(\tilde{\C}) \tilde{\C} \right).
$$
\bigbreak
As $f\in C^{\infty}$  and is compactly supported, the function $\lambda \rightarrow \lambda \,  f'^{2}(\lambda)$ is bounded by some constant, and there exists a nice constant $C$ such that 
$$
\|\nabla g(\X,\X^{*})\|^2 \leq \frac{C}{B^{2}}.
 $$
 This proves that $g$ is $\Ocal(\frac{1}{B})$--Lipschitz.  Paragraph \ref{subsection:lipschitz_concentration} thus leads to (\ref{eq:concentration-f(tildeC)}). 
 \subsection{Step 3: Evaluation of \texorpdfstring{$\frac{1}{M} \mathrm{Tr}\left( f(\hat{\C}(\nu))\right) - \frac{1}{M} \mathrm{Tr}\left( f(\tilde{\C}(\nu))\right)$}{}}
The goal of this paragraph is to establish the following Proposition.
\begin{proposition}
\label{prop:evaluation-trace-hatC-tildeC}
Let $\tilde{D}_N$ the distribution supported by $\Supp(\mu_{MP}^{(c_N)})$ with Stieltjes transform 
\begin{equation}
 \label{eq:def-tildep}
 \tilde{p}_N(z) = (z \, t_N(z))' = \frac{(z \, t_N(z) \, \tilde{t}_N(z))^{2}}{1 - c (z \, t_N(z) \, \tilde{t}_N(z))^{2}}.
\end{equation}
Then, if we denote $<\tilde{D}_N, f>$ by $\tilde{\phi}_N(f)$, we have 
\begin{multline}
    \label{eq:evaluation-trace-hatC-tildeC}
    \left| \frac{1}{M} \mathrm{Tr}\left( f(\hat{\C}(\nu))\right) - \frac{1}{M} \mathrm{Tr}\left( f(\tilde{\C}(\nu))\right) \right. \\ \left. -  \left( \frac{1}{2M} \sum_{m=1}^{M} \frac{s_m''(\nu)}{s_m(\nu)} \right) \; \tilde{\phi}_N(f) \; v_N \; \mathbf{1}_{\alpha > 2/3} \right| \prec u_N.
\end{multline}
\end{proposition}

\begin{remark}
(\ref{eq:evaluation-trace-hatC-tildeC}) implies that $\left| \frac{1}{M} \mathrm{Tr}\left( f(\hat{\C}(\nu))\right) - \frac{1}{M} \mathrm{Tr}\left( f(\tilde{\C}(\nu))\right) \right| \prec \frac{1}{B} $ if $\alpha \leq 2/3$. If $\alpha > 2/3$, the dominant term of $\frac{1}{M} \mathrm{Tr}\left( f(\hat{\C}(\nu))\right) - \frac{1}{M} \mathrm{Tr}\left( f(\tilde{\C}(\nu))\right)$ is the deterministic $\Ocal\left(\frac{B}{N}\right)^{2}$ term $\left( \frac{1}{M} \sum_{m=1}^{M} \frac{s_m''(\nu)}{s_m(\nu)} \right) \; \tilde{\phi}_N(f) \; v_N$, and its substraction from  $ \frac{1}{M} \mathrm{Tr}\left( f(\hat{\C}(\nu))\right) - \frac{1}{M} \mathrm{Tr}\left( f(\tilde{\C}(\nu))\right)$ allows to retrieve a term stochastically dominated by $u_N$. 
\end{remark}

\begin{remark}
We notice that \eqref{equation:domination_Thetabs} leads immediately to 
\begin{equation}
\label{eq:rough-evaluation-hatC-tildeC}
\left|\frac{1}{M}\Tr f(\hat{\C})(\nu)-\frac{1}{M}\Tr f(\tilde{\C})(\nu)\right| \prec\frac{B}{N} + \frac{1}{\sqrt{B}}
\end{equation}
an approximation which is considerably more pessimistic than (\ref{eq:evaluation-trace-hatC-tildeC}). As seen below, the derivation of (\ref{eq:evaluation-trace-hatC-tildeC}) is rather demanding, and 
is based on subtle effects. In order to understand why (\ref{eq:rough-evaluation-hatC-tildeC}) can be improved, we consider the simple case $f(\lambda) = \log \lambda$. We thus have
$$
\frac{1}{M}\Tr f(\hat{\C})(\nu)-\frac{1}{M}\Tr f(\tilde{\C})(\nu) =  \frac{1}{M} \sum_{m=1}^{M}
\left( \log s_m(\nu) - \log \hat{s}_m(\nu)\right)
$$
which depends only on the estimators $(\hat{s}_m(\nu))_{m=1, \ldots, M}$. We just provide a 
brief analysis of the above term. For this, we first remark that it is possible to study 
$\frac{1}{M} \sum_{m=1}^{M} \left( \log s_m(\nu) - \log \hat{s}_m(\nu)\right)$
on the event $\Lambda_\epsilon^{\hat{\D}}(\nu)$ defined by \eqref{equation:definition_Lambda_D}. For each $m$, we expand around $s_m$ the logarithm up to the second order, and obtain that 
\begin{multline}
    \label{eq:expansion-log}
    \frac{1}{M} \sum_{m=1}^{M} \left( \log s_m(\nu) - \log \hat{s}_m(\nu)\right) \\ = -\frac{1}{M}\sum_{m=1}^M (\hat{s}_m-s_m) \frac{1}{s_m} + \frac{1}{M}\sum_{m=1}^M \frac{1}{2}\left(\frac{\hat{s}_m-s_m}{\theta_m}\right)^{2}
\end{multline}
where for each $m$, $\theta_m$ is located between $s_m$ and $\hat{s}_m$. 
Lemma \ref{lemma:concentration_somme_carre_shat} allows to conclude that 
the second term of the right hand side of (\ref{eq:expansion-log}) is dominated 
by $\frac{1}{B} + \left( \frac{B}{N} \right)^{4} = \Ocal(u_N)$ term. In order to evaluate the first term of the r.h.s. of  (\ref{eq:expansion-log}), we note that (\ref{eq:expre-biais-hatsm}) leads to
$$
\frac{1}{M}\sum_{m=1}^M (\mathbb{E}(\hat{s}_m-s_m)) \frac{1}{s_m} =  \frac{1}{2M} \sum_{m=1}^{M} \frac{s_m''}{s_m} \, v_N + \Ocal\left( \left( \frac{B}{N} \right)^{3} + \frac{1}{N} \right).
$$
As $\frac{\hat{s}_m}{s_m} = \frac{\x_m (\I + \Phibs_m) \x_m^*}{B+1}$ we finally remark that
$$
\frac{1}{M}\sum_{m=1}^M \frac{\hat{s}_m - \mathbb{E}(\hat{s}_m)}{s_m}
$$
can be interpreted as a recentered quadratic form of the $M(B+1)$--dimensional 
vector $\x = (\x_1^{T}, \ldots, \x_M^{T})^{T}$. The stochastic domination
relation
$$
\left| \frac{1}{M}\sum_{m=1}^M \frac{\hat{s}_m - \mathbb{E}(\hat{s}_m)}{s_m} \right|\prec \frac{1}{B}
$$
then follows from the Hanson-Wright inequality. 
Putting all the pieces together, and using that $ \frac{1}{B} + \Ocal\left( \left(  \frac{B}{N} \right)^{3} + \frac{1}{N} \right)  = \Ocal(u_N)$ and that $v_N = o(u_N)$ if $\alpha < 2/3$, we obtain that 
$$
\left| \frac{1}{M} \sum_{m=1}^{M} \left( \log s_m(\nu) - \log \hat{s}_m(\nu)\right) +  \frac{1}{2M} \sum_{m=1}^{M} \frac{s_m''(\nu)}{s_m(\nu)} \, v_N \, \mathbf{1}_{\alpha > 2/3} \right| \prec u_N.
$$
Comparing this result with (\ref{eq:evaluation-trace-hatC-tildeC}), we deduce that  $<\tilde{D}_N, f> = -1$. We just check this formula directly. For this, we notice that function $z \rightarrow  \log z$ is holomorphic inside a neighbourhood of the interval $[a_1, a_2]$. We consider the expression (\ref{eq:inversion-formula-distribution}) of $<\tilde{D}_N, f >$ and remark that if $\left(\partial \mathcal{R}_{\epsilon}\right)_{-}$ denotes the negatively oriented contour 
$$
\left(\partial \mathcal{R}_{\epsilon}\right)_{-} = \{ \lambda \pm i \epsilon, \lambda \in [a_1,a_2] \} \cup \{ a_1 + 
iy, y \in [-\epsilon, \epsilon] \} \cup \{ a_2 + 
iy, y \in [\epsilon, -\epsilon] \} 
$$
then, by (\ref{eq:inversion-formula-distribution}), $<\tilde{D}_N, f >$ can also be written as the contour integral
$$
<\tilde{D}_N, f > = \lim_{\epsilon \rightarrow 0} \frac{1}{2 i \pi } \int_{ (\partial \mathcal{R}_{\epsilon})_{-}} \log z \, \tilde{p}_N(z) \, dz.
$$
But, the above contour integral does not depend on $\epsilon$, so that for each $\epsilon$, we have 
$$
<\tilde{D}_N, f > =  \frac{1}{2 i \pi } \int_{ (\partial \mathcal{R}_{\epsilon})_{-}} \log z \, \tilde{p}_N(z) \, dz .
$$
Using the expression of $\tilde{p}_N(z)$ and the integration by parts trick, we get that 
$$
<\tilde{D}_N, f > = - \frac{1}{2 i \pi } \int_{ (\partial \mathcal{R}_{\epsilon})_{-}} t_N(z) \, dz.
$$
Taking the limit $\epsilon \rightarrow 0$, and using the Stieltjes inversion formula for the Marcenko-Pastur distribution $\mu_{MP}^{(c_N)}$, we finally obtain that 
$$
<\tilde{D}_N, f > = - \frac{1}{ \pi} \lim_{\epsilon \rightarrow 0} \int_{a_1}^{a_2} \Im (t_N(\lambda + i \epsilon)) \, d \lambda = - \mu_{MP}^{(c_N)}([a_1,a_2]) =  -1
$$
which is the expected result.

\end{remark}

\begin{proof}
We now establish (\ref{eq:evaluation-trace-hatC-tildeC}). In order to simplify the notations, we put 
\begin{equation}
    \label{eq:def-tilder}
    \tilde{r}_N(\nu) = \frac{1}{2M} \sum_{m=1}^{M} \frac{s_m''(\nu)}{s_m(\nu)}.
\end{equation}
The Helffer-Sjöstrand formula implies that 
\begin{multline*}
    \frac{1}{M}\Tr f(\hat{\C}) - \frac{1}{M}\Tr f(\tilde{\C}) - \tilde{r}_N(\nu) \; \tilde{\phi}_N(f) \; v_N \; \mathbf{1}_{\alpha > 2/3} = \\ 
    \frac{1}{\pi}\Re\int_{\Dcal}\diff x \diff y\,\bar{\partial}\Phi_k(f)(z) \left[ \frac{1}{M}(\Tr\hat{\Q}(z)-\Tr\tilde{\Q}(z)) - 
    \tilde{r}_N(\nu) \; \tilde{p}_N(z) \; v_N \; \mathbf{1}_{\alpha > 2/3} \right].
    \label{equation:helfjer_sjostrand}
\end{multline*}

\subsubsection{Reduction to the study of \texorpdfstring{$\zeta$}{} }
We define
\begin{equation}
    \label{eq:def_zeta}
    \zeta=\int_\Dcal\diff x \diff y\,\overline{\partial}\Phi_k(f)(z) \frac{1}{M} \sum_{m=1}^{M} (z\Q)'_{mm} \left(\frac{\|\x_m\|_2^{2}}{B+1} -1 \right)
\end{equation}
where we recall that the row vectors $(\x_m)_{m=1, \ldots, M}$ 
are the rows of the i.i.d. matrix $\X$. We establish in this paragraph that 
\begin{equation}
\label{eq:reduction}
 \left| \int_{\Dcal} \diff x \diff y\, \bar{\partial}\Phi_k(f)(z) \left( \frac{1}{M} \Tr\{\hat{\Q}-\tilde{\Q}\} - \tilde{r}_N(\nu) \; \tilde{p}_N(z) \; v_N \; \mathbf{1}_{\alpha > 2/3}\right) - \zeta\right| \prec u_N.
 \end{equation}
It turns out that by Lemma \ref{lemma:variance-zeta} and Lemma \ref{lemma:esperance-zeta}
in Paragraph \ref{paragraph:emmes-zeta} below, 
$\zeta$ satisfies the key properties: 
$$ |\zeta| \le |\zeta-\Ebb\zeta|+|\Ebb\zeta| \prec \frac{1}{B}.$$
(\ref{eq:evaluation-trace-hatC-tildeC}) will then follow directly from \eqref{eq:reduction}. \\

Plugging in the integral expression of $\zeta$, and using the expression (\ref{eq:def-tildep}), we get:
\begin{multline*}
     \left|\int_{\Dcal} \bar{\partial}\Phi_k(f)(z) \left( \frac{1}{M} \Tr\{\hat{\Q}-\tilde{\Q}\} -  \tilde{r}_N(\nu) \; \tilde{p}_N(z) \; v_N \; \mathbf{1}_{\alpha > 2/3}\right) \diff x \diff y - \zeta \right| \\
     = \left|\int_\Dcal \diff x \diff y\overline{\partial}\Phi_k(f)(z) \left( \frac{1}{M} \Tr\{\hat{\Q}-\tilde{\Q}\} -  \tilde{r}_N \; (z t_N(z))' \; v_N \; \mathbf{1}_{\alpha > 2/3} \right. \right. \\ \left. \left.  - \frac{1}{M} \sum_{m=1}^M (z \Q)'_{mm} \left(\frac{\|\x_m\|_2^{2}}{B+1} -1 \right) \right)\right| .
\end{multline*}

We recall the definition of $\Thetabs:=\hat{\C} - \tilde{\C}$ from \eqref{equation:definition_Theta}. We will proceed in three steps, which, in turn, will imply (\ref{eq:reduction}): 
\begin{enumerate}
    \item \begin{equation}
            \label{eq:controle_Q_hat_tilde_Q_carre}
            \left|\int_\Dcal\diff x \diff y\, \overline{\partial}\Phi_k(f)(z)\left(\frac{1}{M} \Tr\{\hat{\Q}-\tilde{\Q}\}+ \frac{1}{M} \Tr\{\Q^2\Thetabs\}\right)\right| \prec u_N
        \end{equation}
    \item \begin{multline}
        \label{eq:controle_Q_carre}
        \left|\int_\Dcal\diff x \diff y\, \overline{\partial}\Phi_k(f)(z)\left(\frac{1}{M} \Tr\{\Q^2\Thetabs\} \right. \right. \\ \left. \left. -2\, \frac{1}{M} \Tr\frac{\X\X^*}{B+1}\Q^2(\hat{\D}^{-1/2}\D^{1/2}-\I)\right)\right| \prec u_N
    \end{multline}
    \item \begin{multline}
        \label{eq:controle_Q_XX}
        \left|\int_\Dcal\diff x \diff y\, \overline{\partial}\Phi_k(f)(z)\times\right. \\
        \left.\left(2\, \frac{1}{M} \Tr\frac{\X\X^*}{B+1}\Q^2(\I - \hat{\D}^{-1/2}\D^{1/2})-  \tilde{r}_N  \; (z t_N(z))' \; v_N \; \mathbf{1}_{\alpha > 2/3}  - \right. \right. \\ \left. \left. \frac{1}{M} \sum_{m=1}^M(z\Q)'_{mm} \left(\frac{\|\x_m\|_2^{2}}{B+1} -1  \right)\right)\right| \prec u_N
    \end{multline}
\end{enumerate}

\textbf{Step 1.} Using the well-known identity $\A^{-1}-\B^{-1}=\B^{-1}(\B-\A)\A^{-1}$, we express $\hat{\Q}-\tilde{\Q}$ as: 
\begin{equation}
    \label{equation:perturbation_Q}
    \hat{\Q}-\tilde{\Q} = - \tilde{\Q}\Thetabs\hat{\Q}.
\end{equation}
We claim that it is possible to approximate $\Tr\tilde{\Q}\Thetabs\hat{\Q}$ by 
$ \Tr\Q\Thetabs\Q$. Indeed, we have
\begin{align*}
    &|\Tr\tilde{\Q}\Thetabs\hat{\Q} - \Tr\Q\Thetabs\Q| \\
    &\hskip1cm= |\Tr\tilde{\Q}\Thetabs\hat{\Q} - \Tr\tilde{\Q}\Thetabs\tilde{\Q} + \Tr\tilde{\Q}\Thetabs\tilde{\Q} - \Tr\tilde{\Q}\Thetabs\Q + \Tr\tilde{\Q}\Thetabs\Q- \Tr\Q\Thetabs\Q|  \\
    &\hskip1cm\le |\Tr\tilde{\Q}\Thetabs\hat{\Q}  - \Tr\tilde{\Q} \Thetabs\tilde{\Q} | + |\Tr\tilde{\Q}\Thetabs\tilde{\Q} - \Tr\tilde{\Q}\Thetabs\Q| + |\Tr\tilde{\Q}\Thetabs\Q-\Tr\Q\Thetabs\Q| \\
    &\hskip1cm:= T_1 + T_2 + T_3.
\end{align*}
The following rough bounds are enough to control $T_1$ (we used \eqref{eq:inegalite_stieltjes_transform_Q} to control the norm of the resolvents):
$$ T_1 = |\Tr\tilde{\Q}\Thetabs(\hat{\Q}-\tilde{\Q})| = |\Tr\tilde{\Q}\Thetabs\tilde{\Q}\Thetabs\hat{\Q}| \le M\|\tilde{\Q}\|^2\|\hat{\Q}\|\|\|\Thetabs\|^2 \le \frac{1}{\Im^3z}M\|\Thetabs\|^2.$$

Concerning $T_2$ and $T_3$, we write similarly that $\tilde{\Q} - \Q = - \tilde{\Q} \tilde{\Deltabs} \Q$, and obtain that 
$$ T_2 = |\Tr\tilde{\Q}\Thetabs\tilde{\Q}-\Tr\tilde{\Q}\Thetabs\Q|\le M\|\tilde{\Q}\|^2\|\Q\|\|\tilde{\Deltabs}\|\Thetabs\| \le \frac{1}{\Im^3z}M\|\tilde{\Deltabs}\|\|\Thetabs\| $$
$$ T_3 = |\Tr\tilde{\Q}\Thetabs\Q-\Tr\Q\Thetabs\Q| \le M\|\tilde{\Q}\|\|\Q\|^2\|\tilde{\Deltabs}\|\Thetabs\| \le \frac{1}{\Im^3z}M\|\tilde{\Deltabs}\|\|\Thetabs\|.$$

Plugging these estimations into the left hand side of \eqref{eq:controle_Q_hat_tilde_Q_carre}, 
we obtain that
\begin{align*}
    &\left|\int_\Dcal\diff x \diff y\, \overline{\partial}\Phi_k(f)(z)\left(\frac{1}{M} \Tr\{\hat{\Q}-\tilde{\Q}\}-\frac{1}{M} \Tr\{\Q^2\Thetabs\}\right)\right| \\
    &\hskip4cm\le  \int_\Dcal\diff x \diff y |\overline{\partial}\Phi_k(f)(z)|\frac{1}{M} (T_1+T_2+T_3) \\
    &\hskip4cm\le C (\|\Thetabs\|^2+2\|\tilde{\Deltabs}\|\Thetabs\|) .
\end{align*}
Moreover, the concentration results \eqref{equation:domination_Thetabs} for $\|\Thetabs\|$ and \eqref{eq:concentration-tilde-Delta} for $\|\tilde{\Deltabs}\|$ from Proposition \ref{proposition:X_Gamma}, imply that 
$$
\|\Thetabs\|^2 + 2 \|\Thetabs\| \|\tilde{\Deltabs}\| \prec \frac{1}{B} + \frac{1}{\sqrt{B}}
\, \frac{B}{N} + \left( \frac{B}{N} \right)^{3} = u_N.
$$
This finally establishes (\ref{eq:controle_Q_hat_tilde_Q_carre}). 

\bigbreak
\textbf{Step 2.} We claim that:
\begin{equation}
    \label{eq:approximation_Thetabs}
    \left\|\Thetabs - \left((\hat{\D}^{-1/2}\D^{1/2}-\I)\frac{\X\X^*}{B+1} + \frac{\X\X^*}{B+1}(\D^{1/2}\hat{\D}^{-1/2}-\I)\right)\right\| \prec u_N.
\end{equation}
We recall that $\hat{\S}$ can be written using the definition \eqref{eq:def-tildeC} of $\tilde{\C}$, and use the decomposition \eqref{eq:representation-tildeC} of $\tilde{\C}$ from Corollary \ref{corollary:concentration_delta}. Using these results, 
we get that
$$ \hat{\S} = \D^{1/2} \, \tilde{\C} \, \D^{1/2} = \D^{1/2}\left(\frac{\X\X^*}{B+1}+\tilde{\Deltabs}\right)\D^{1/2} . $$
Plugging this expression of $\hat{\S}$ into \eqref{equation:decomposition_Theta}, we obtain easily that 
 \begin{align*}
    \Thetabs = &
    (\hat{\D}^{-1/2}\D^{1/2}-\I)\left(\frac{\X\X^*}{B+1}+\tilde{\Deltabs}\right)\D^{1/2}\hat{\D}^{-1/2} \\
    &\hskip5cm+ \left(\frac{\X\X^*}{B+1}+\tilde{\Deltabs}\right)(\D^{1/2}\hat{\D}^{-1/2}-\I)\\
    &:= \Thetabs_1 + \Thetabs_2 .
 \end{align*}
 
As $\tilde{\Deltabs}$ is a negligible quantity, one should expect that the leading quantity in $\Thetabs_1$ and $\Thetabs_2$ is respectively $(\hat{\D}^{-1/2}\D^{1/2}-\I)\frac{\X\X^*}{B+1}\D^{1/2}\hat{\D}^{-1/2}$ and $\frac{\X\X^*}{B+1}(\D^{1/2}\hat{\D}^{-1/2}-\I)$. To prove it, write:
\begin{align}
\notag
    &\left\|\Thetabs_1 -  (\hat{\D}^{-1/2}\D^{1/2}-\I)\frac{\X\X^*}{B+1}\D^{1/2}\hat{\D}^{-1/2} \right\| \\
    \notag 
    &\hskip4cm= \left\|(\hat{\D}^{-1/2}\D^{1/2}-\I)\tilde{\Deltabs}\D^{1/2}\hat{\D}^{-1/2}\right\| \\
\label{equation:approximation_1_Thetabs_1}
    &\hskip4cm\le \|\hat{\D}^{-1/2}\D^{1/2}-\I\|\|\tilde{\Deltabs}\|\|\D^{1/2}\hat{\D}^{-1/2}\|.
\end{align}

$\tilde{\Deltabs}$ is controlled by \eqref{eq:concentration-tilde-Delta} from Corollary \ref{corollary:concentration_delta}, and $\hat{\D}^{-1/2}\D^{1/2}-\I$ is controlled by \eqref{equation:domination_D_racine_carre} from Lemma \ref{lemma:concentration_ratio_s_shat} (it is a diagonal matrix which elements are stochastically dominated by Lemma \ref{lemma:concentration_ratio_s_shat}). Moreover, from Lemma \ref{lemma:concentration_ratio_s_shat}, it holds that $\|\D^{1/2}\hat{\D}^{-1/2}\|\prec 1$. Combining these estimates into \eqref{equation:approximation_1_Thetabs_1}, one gets:
\begin{equation}
\label{equation:approximation_2_Thetabs_1}
    \left\|\Thetabs_1 -  (\hat{\D}^{-1/2}\D^{1/2}-\I)\frac{\X\X^*}{B+1}\D^{1/2}\hat{\D}^{-1/2} \right\| \prec \left(\frac{1}{\sqrt{B}}+\frac{B^2}{N^2}\right)\frac{B}{N} = \Ocal(u_N).
\end{equation}

Using that $\|\hat{\D}^{-1/2}\D^{1/2}-\I\| \prec \frac{1}{\sqrt{B}}+\frac{B^2}{N^2}$ as well as  \eqref{eq:stoc-domination-eig-sing-wishart} from Paragraph \ref{section:haagerup} to control the norm of $\X\X^*/(B+1)$, one can further approximate $(\hat{\D}^{-1/2}\D^{1/2}-\I)\frac{\X\X^*}{B+1}\D^{1/2}\hat{\D}^{-1/2}$ by 
$(\hat{\D}^{-1/2}\D^{1/2}-\I)\frac{\X\X^*}{B+1}$. In particular, it is easy to check that
\begin{equation}
\label{equation:approximation_finale_Thetabs_1}
   \left \|\Thetabs_1 -  (\hat{\D}^{-1/2}\D^{1/2}-\I)\frac{\X\X^*}{B+1}\right\| \prec u_N.
\end{equation}


Similarly for $\Thetabs_2$, one would obtain:
\begin{equation}
\label{equation:approximation_finale_Thetabs_2}
   \left \|\Thetabs_2 -  \frac{\X\X^*}{B+1}(\D^{1/2}\hat{\D}^{-1/2}-\I)\right\| \prec u_N.
\end{equation}

Combining \eqref{equation:approximation_finale_Thetabs_1} and  \eqref{equation:approximation_finale_Thetabs_2}, we obtain \eqref{eq:approximation_Thetabs}. To finish the proof of Step 2, it remains to consider $\Tr\Q^2\Thetabs$ and prove \eqref{eq:controle_Q_carre}. Remark that $\X\X^*/(B+1)$ and its resolvent $\Q$ commutes.
\begin{align}
\notag
    &\Tr\Q^2\left((\hat{\D}^{-1/2}\D^{1/2}-\I)\frac{\X\X^*}{B+1} + \frac{\X\X^*}{B+1}(\D^{1/2}\hat{\D}^{-1/2}-\I)\right) \\
\label{eq:approximation_Thetabs_permute}
    &\hskip2cm= 2 \, \Tr\frac{\X\X^*}{B+1}\Q^2(\hat{\D}^{-1/2}\D^{1/2}-\I)
\end{align}

Therefore, using \eqref{eq:approximation_Thetabs_permute}:
\begin{multline}
\label{equation:approximation_trace_Q_square_Thetabs}
    \left| \frac{1}{M} \Tr\Q^2\Thetabs - 2\, \frac{1}{M}  \Tr\frac{\X\X^*}{B+1}\Q^2(\hat{\D}^{-1/2}\D^{1/2}-\I)\right | \\ \le \|\Q\|^2\left\|\Thetabs - \left((\hat{\D}^{-1/2}\D^{1/2}-\I)\frac{\X\X^*}{B+1} + \frac{\X\X^*}{B+1}(\D^{1/2}\hat{\D}^{-1/2}-\I)\right)\right\|
\end{multline}
so that the left hand side of \eqref{eq:approximation_Thetabs} is recognised in the right hand side of \eqref{equation:approximation_trace_Q_square_Thetabs}. We can finally prove \eqref{eq:controle_Q_carre} by following the same idea as in Step 1:
\begin{align*}
    &\left|\int_\Dcal\diff x \diff y\, \overline{\partial}\Phi_k(f)(z)\, \frac{1}{M} \left(\Tr\{\Q^2\Thetabs\}-2\Tr\frac{\X\X^*}{B+1}\Q^2(\hat{\D}^{-1/2}\D^{1/2}-\I)\right)\right| \\
    &\hskip1cm\le  \left\|\Thetabs - \left((\hat{\D}^{-1/2}\D^{1/2}-\I)\frac{\X\X^*}{B+1} + \frac{\X\X^*}{B+1}(\D^{1/2}\hat{\D}^{-1/2}-\I)\right)\right\| \\
    &\hskip7cm\times\underbrace{\int_\Dcal |\overline{\partial}\Phi_k(f)(z)|\frac{1}{\Im^2z}\diff x \diff y}_{<+\infty}.
\end{align*}
This proves  \eqref{eq:controle_Q_carre} and ends Step 2.

\bigbreak
\textbf{Step 3.}
By definition of the resolvent, the following identity holds $\left(\frac{\X\X^*}{B+1}-z\I_M\right)\Q(z) = \I_M $, which leads to the so-called resolvent identity:
\begin{equation}
\label{equation:resolvent_identity}
    \frac{\X\X^*}{B+1}\Q = \I_M+z\Q.
\end{equation}

Using \eqref{equation:resolvent_identity} as well the identity $\Q'(z) = \Q^{2}(z)$ one can write:
\begin{align}
    \label{eq:trace_Q_carre}
    \notag
    \frac{1}{M} \Tr\frac{\X\X^*}{B+1}\Q^2(\I - \D^{1/2}\hat{\D}^{-1/2}) &= \frac{1}{M} \Tr(\I+z\Q)\Q(\I - \D^{1/2}\hat{\D}^{-1/2}) \\
    &= \frac{1}{M} \, \sum_{m=1}^{M} \left( z \Q \right)'_{mm} \, \left( 1 - \sqrt{\frac{s_m}{\hat{s}_m}}\right).
\end{align}

To handle $1 - \sqrt{\frac{s_m}{\hat{s}_m}}$ we use the following Taylor expansion: define the mapping $h$ by $h(u)=1 - \frac{1}{\sqrt{u}}$, with $h'(u)=\frac{1}{2}\frac{1}{u^{3/2}}$ and $h''(u)=-\frac{3}{4}\frac{1}{u^{5/2}}$.
A Taylor expansion to the second order of $h$ around $1$ provides:
\begin{multline*}
    h\left(\frac{\hat{s}_m}{s_m}\right) = h(1) + \left(\frac{\hat{s}_m}{s_m} - 1\right)h'(1) + \frac{1}{2}\left(\frac{\hat{s}_m}{s_m} - 1\right)^2h''(\theta_m) \\ = \frac{1}{2s_m}(\hat{s}_m-s_m) + \frac{1}{2}\frac{h''(\theta_m)}{s_m^2}(\hat{s}_m-s_m)^2
\end{multline*} 
where $\theta_m$ is some random quantity between $\hat{s}_m$ and $s_m$. Therefore \eqref{eq:trace_Q_carre} becomes 
\begin{multline*}
    \frac{1}{M}\Tr\left((z\Q)'(\I - \D^{1/2}\hat{\D}^{-1/2})\right) \\ =  \frac{1}{M}\Tr\left((z\Q)'\diag\left(\frac{\hat{s}_m-s_m}{2s_m}+\frac{1}{2}\frac{h''(\theta_m)(\hat{s}_m-s_m)^2}{s_m^2}: m\in\{1,\ldots,M\}\right)\right).
\end{multline*}
Lemma \ref{lemma:localization_s_m} implies that the set $\Lambda_\epsilon^{\hat{\D}}(\nu)$ defined by \eqref{equation:definition_Lambda_D} holds with exponentially high probability. Therefore, 
it is sufficient to study the term  $\frac{1}{M} \Tr (z \Q)'(\I - \D^{1/2}\hat{\D}^{-1/2})$ on the event 
$\Lambda_\epsilon^{\hat{\D}}(\nu)$. If $\Lambda_\epsilon^{\hat{\D}}(\nu)$ holds, $\theta_m$
belongs to $ [\barbelow{s},\bar{s}]+\epsilon$ for each  $m\in\{1, \ldots, M\}$, and  
$\sup_{m\ge1}|h''(\theta_m)|$ is bounded by a nice constant. Moreover, as $\inf_\nu\inf_{m\ge1} s_m(\nu)$ is bounded away from zero, 
there exists a nice constant $C$ for which the inequality
\begin{multline*}
   \left|\frac{1}{M} \Tr\left((z\Q)' \diag\left(\frac{1}{2}\frac{h''(\theta_m)(\hat{s}_m-s_m)^2}{s_m^2}: m\in\{1,\ldots,M\}\right)\right) \right|  \\ \le C (\|\Q\|+z\|\Q\|^2) \frac{1}{M} \sum_{m=1}^M(\hat{s}_m-s_m)^2 \leq C(z) \,  \frac{1}{M} \sum_{m=1}^M(\hat{s}_m-s_m)^2 
\end{multline*} 
holds on $\Lambda_\epsilon^{\hat{\D}}(\nu)$, where we recall that $C(z)$ can be written as 
$P_1(|z|) P_2(\frac{1}{\Im z})$ for some nice polynomials $P_1$ and $P_2$.
Following again the same argument as in Step 1, we obtain that
\begin{align*}
    &\left|\int_\Dcal\diff x \diff y\overline{\partial}\Phi_k(f)(z)\left\{\frac{1}{M} \Tr(z\Q)'(\I - \hat{\D}^{-1/2}\D^{1/2})\right.\right.\\
    &\hskip3cm\left.\left.-\frac{1}{M}\Tr(z\Q)'\diag\left(\frac{\hat{s}_m-s_m}{2s_m}: m\in\{1,\ldots,M\}\right)\right\}\right| \\
    &\quad\le C \frac{1}{M} \sum_{m=1}^M(\hat{s}_m-s_m)^2
\end{align*}
on $\Lambda_\epsilon^{\hat{\D}}(\nu)$ provided $k \geq \mathrm{Deg}(P_2)$ . Lemma \ref{lemma:concentration_somme_carre_shat} 
in Appendix implies that
$$ \frac{1}{M} \sum_{m=1}^M(\hat{s}_m-s_m)^2 \prec \frac{1}{B} +\frac{B^4}{N^4} = \Ocal(u_N) .$$
We have thus shown that 
\begin{align*}
    &\left|\int_\Dcal\diff x \diff y\overline{\partial}\Phi_k(f)(z)\left\{ \frac{1}{M} \Tr(z\Q)'(\I - \hat{\D}^{-1/2}\D^{1/2})\right.\right.\\
    &\hskip3cm\left.\left.-\frac{1}{M} \Tr(z\Q)'\diag\left(\frac{\hat{s}_m-s_m}{2s_m}: m\in\{1,\ldots,M\}\right)\right\}\right| \\
    &\prec u_N.
\end{align*}

We denote by $\eta_N(z)$ the term defined by 
\begin{multline}
    \label{def-etaN}
    \eta_N(z) =  \frac{1}{M} \sum_{m=1}^{M} (z\Q)'_{mm} \left( \frac{\hat{s}_m-s_m}{s_m} - \left(\frac{\|\x_m\|_2^{2}}{B+1} -1  \right)\right) \\ - \tilde{r}_N(z) (z t_N(z))' v_N \mathbf{1}_{\alpha > 2/3}
\end{multline}
and define $\delta_N$ as 
$$
\delta_N = \int_\Dcal\diff x \diff y\, \overline{\partial}\Phi_k(f)(z) \, \eta_N(z).
$$
In order to establish (\ref{eq:controle_Q_XX}), it is sufficient to prove that $|\delta_N| \prec u_N$. For this, we first remark 
that $\hat{s}_m = s_m \frac{\x_m (\I + \Phibs_m) \x_m^*}{B+1}$, so that $\eta_N(z)$ can also be written as
\begin{equation}
    \label{eq:expre-etaN}
   \eta_N(z) = \frac{1}{M} \sum_{m=1}^{M} (z\Q)'_{mm} \, \frac{\x_m \Phibs_m \x_m^*}{B+1} -  \tilde{r}_N (z t_N(z))' v_N \mathbf{1}_{\alpha > 2/3}.
\end{equation}
We express $\eta_N(z)$ as $\eta_N(z) = \eta_{1,N}(z) + \eta_{2,N}(z) + \eta_{3,N}(z)$
where $(\eta_{i,N})_{i=1,2,3}$ are defined by 
\begin{align*}
    \eta_{1,N}(z) = \frac{1}{M} \sum_{m=1}^{M} (z\Q)'_{mm} \left( \frac{\x_m \Phibs_m \x_m^*}{B+1} - \frac{1}{B+1} \Tr \Phibs_m \right) & \\
    \eta_{2,N}(z) =  \frac{1}{M} \sum_{m=1}^{M} \Ebb[(z\Q)'_{mm}] \, \frac{1}{B+1} \Tr \Phibs_m
    - \tilde{r}_N (z t_N(z))' v_N \mathbf{1}_{\alpha > 2/3} & \\
    \eta_{3,N}(z) =  \frac{1}{M} \sum_{m=1}^{M} ((z\Q)'_{mm})^{\circ} \, \frac{1}{B+1} \Tr \Phibs_m
\end{align*}
and denote by $(\delta_{i,N})_{i=1,2,3}$ the contributions of $(\eta_{i,N})_{i=1,2,3}$ 
to $\delta_N$. We recall the definition (\ref{eq:def-xrond}) of $ ((z\Q)'_{mm})^{\circ}$. In order to evaluate $\delta_{1,N}$, we note that $|(z\Q)'_{mm}| = |Q_{mm} + z \Q^{2}_{mm}| \leq C(z)$ and that 
$$
|\eta_{1,N}(z)| \leq C(z) \, \sup_{m=1, \ldots, M} \left|  \frac{\x_m \Phibs_m \x_m^*}{B+1} - \frac{1}{B+1} \Tr \Phibs_m \right| .
$$
Therefore, for $k$ large enough, $\delta_{1,N}$ satisfies $|\delta_{1,N}| \leq C \sup_{m=1, \ldots, M}  \left|  \frac{\x_m \Phibs_m \x_m^*}{B+1} - \frac{1}{B+1} \Tr \Phibs_m \right|$. The Hanson-Wright inequality as well as the bound (\ref{eq:frobenius-norm-phim}) of the Frobenius norm of $\Phibs_m$ imply that $|\delta_{1,N}|\ \prec u_N$. We now evaluate $\delta_{2,N}$. For this, 
we notice that the results reviewed in Paragraph \ref{subsubsec:resolvent-MP} imply that 
$\mathbb{E} (z\Q)'_{mm} = (z \beta_N(z))' = (z t_N(z))' + (z\epsilon_N(z))'$ where $|(z\epsilon_N(z))'| \leq \frac{C(z)}{M^{2}}$. 
Therefore, using (\ref{eq:expre-Trace-Phim}), we obtain that 
\begin{align*}
\eta_{2,N} (z)) = & \left( \frac{1}{M} \sum_{m=1}^{2M} \frac{s''_m}{s_m} \right)  \, (z t_N(z))' \, v_N + \epsilon_{1,N}(z) - \tilde{r}_N \, (z t_N(z))' \, v_N \, \mathbf{1}_{\alpha > 2/3}     \\
 = & \tilde{r}_N  \, (z t_N(z))' \, v_N \, \mathbf{1}_{\alpha \leq 2/3} + \epsilon_{1,N}(z)
\end{align*}
where $\epsilon_{1,N}(z)$ satisfies $|\epsilon_{1,N}(z)| \leq C(z) \, u_N$. We then deduce that $|\eta_{2,N}(z)| \leq C(z) u_N$ because if $\alpha \leq 2/3$, $v_N \leq u_N$. This implies that $|\delta_{2,N}| = \Ocal(u_N)$. In order to address $\delta_{3,N}$, we interpret $\delta_{3,N}$ 
as a function $g$ of $(\X,\X^*)$, and use the Gaussian concentration inequality presented in Paragraph \ref{subsection:lipschitz_concentration}. In particular, we verify that 
$$
\| \nabla g \| \leq  C \, \frac{1}{\sqrt{B}} \left( \frac{B}{N} \right)^{2} = o(u_N).
$$
As $\mathbb{E}(\delta_{3,N}) = 0$, this leads immediately to $|\delta_{3,N}| \prec u_N$. We just check that 
\begin{equation}
\label{eq:partial-delta3}
\sum_{i,j} \left| \frac{\partial g}{\partial X_{ij}} \right|^{2} \leq C \frac{1}{B} \left( \frac{B}{N} \right)^{4}.
\end{equation}
For this, we express $(z\Q)'_{mm}$ as $(z\Q)'_{mm} = Q_{mm} + z \Q^{2}_{mm}$ and notice that
\begin{align*}
 \frac{\partial Q_{mm}}{\partial X_{ij}} = & - Q_{mi} \left( \frac{\X^*}{B+1} \Q \right)_{jm} \\
  \frac{\partial \Q^{2}_{mm}}{\partial X_{ij}} = & -(\Q^{2})_{mi} \left( \frac{\X^*}{B+1} \Q \right)_{jm} - 
  Q_{mi} \left( \frac{\X^*}{B+1} \Q^{2} \right)_{jm} .
\end{align*}
Using the Jensen inequality, we obtain that
$$
\left| \frac{ \partial g}{\partial X_{ij}} \right|^{2} \leq \int_\Dcal\diff x \diff y|\overline{\partial}\Phi_k(f)(z)|^{2} \frac{1}{M} \sum_{m=1}^{M} \left| \frac{\partial (z\Q)'_{mm}}{\partial X_{ij}} \right|^{2} \left( \frac{1}{B+1} \Tr \Phibs_m \right)^{2}.
$$
Summing over $i,j$ leads to the expected evaluation of (\ref{eq:partial-delta3}) and to $|\delta_{3,N}| \prec u_N$.
 This, in turn, completes the proof of (\ref{eq:controle_Q_XX}) and of 
(\ref{eq:reduction}).

Up to the Lemma \ref{lemma:variance-zeta} and Lemma \ref{lemma:esperance-zeta}, Theorem 
\ref{theo:domination-psi} is proved.
\end{proof}
\subsubsection{Proof of  Lemma \ref{lemma:variance-zeta} and Lemma \ref{lemma:esperance-zeta}}
\label{paragraph:emmes-zeta}
We now establish Lemma \ref{lemma:variance-zeta} and Lemma \ref{lemma:esperance-zeta}.
\begin{lemma}
\label{lemma:variance-zeta}
The family of random variables $\zeta(\nu)-\Ebb\zeta(\nu)$, $\nu \in [0,1]$ satisfies the following property: 
\begin{equation}
    \label{eq:variance-zeta}
    |\zeta(\nu)-\Ebb\zeta(\nu)|\prec \frac{1}{B}.
\end{equation}
\end{lemma}

\begin{proof} 
$\zeta$ defined by \eqref{eq:def_zeta} can be written as
\begin{align*}
    \zeta &= \int_\Dcal\diff x \diff y\,\overline{\partial}\Phi_k(f)(z) \frac{1}{M} \sum_{m=1}^{M} Q_{mm} \left( \frac{\|\x_m\|_2^{2}}{B+1} -1 \right) + \\
    &\hskip3cm \int_\Dcal\diff x \diff y\,\overline{\partial}\Phi_k(f)(z) \frac{1}{M} \sum_{m=1}^{M} z(\Q^2)_{mm} \left( \frac{\|\x_m\|_2^{2}}{B+1} -1 \right) \\
    &:=  \zeta_1 + \zeta_2.
\end{align*} 
In the following, we omit to 
evaluate $|\zeta_1(\nu) - \mathbb{E}(\zeta_1(\nu))|$, and just establish that $|\zeta_2(\nu) - \mathbb{E}(\zeta_2(\nu))| \prec \frac{1}{B}$ using the Gaussian concentration inequality from Paragraph \ref{subsection:lipschitz_concentration}. 

Recall that $\|\x_m\|_2^2$ is a $\chi^2_{2(B+1)}$ random variable. Therefore it is clear that: 
$$ \left|\frac{\|\x_m\|_2^2}{B+1}-1\right| \prec \frac{1}{\sqrt{B}}. $$

Knowing this, the idea is to show that, conditioned on the event where 
the random variables $\left(\frac{\|\x_m\|_2^2}{B+1}-1 \right)_{m=1, \ldots, M}$ are localized, which holds with exponentially high probability, $\zeta_2$ is a $\Ocal(\frac{1}{B^{1-\epsilon}})$--Lipschitz function of the entries of the matrix $\X$ for any $\epsilon>0$. Let $0<\epsilon < \frac{1}{2}$, and define the family of events $A_{m,\epsilon}(\nu)$, $m=1, \ldots, M, \,  \nu \in [0,1]$ given by
\begin{equation}
    \label{equation:A_N}
    A_{m,\epsilon}(\nu)=\left\{\frac{\|\x_m(\nu)\|_2^2}{B+1}\in\left[1-\frac{B^\epsilon}{\sqrt{B}},1+\frac{B^\epsilon}{\sqrt{B}}\right]\right\}
\end{equation} 
as well as $A_\epsilon(\nu) = \cap_{m=1}^{M} A_{m,\epsilon}(\nu)$. It is clear that the family of events $A_{m,\epsilon}(\nu)$, $m=1, \ldots, M$, $\nu \in [0,1]$ holds with exponentially high probability, and that the 
same property holds for the family $A_\epsilon(\nu), \nu \in [0,1]$. We claim that there exists a family of $\mathcal{C}^\infty$ functions
$(g_{B,\epsilon})_{B \geq 1}$ satisfying
\[   
g_{B,\epsilon}(t) = 
     \begin{cases}
       t-1 &\quad\text{if }t\in[1-\frac{B^\epsilon}{\sqrt{B}},1+\frac{B^\epsilon}{\sqrt{B}}]\\
       0 &\quad\text{if } t\notin [1-2\frac{B^\epsilon}{\sqrt{B}},1+2\frac{B^\epsilon}{\sqrt{B}}]
     \end{cases}
\]
and 
\begin{equation}
    \label{equation:g_B_bound}
    \sup_t |g_{B,\epsilon}(t)|\le C \,  \frac{B^\epsilon}{\sqrt{B}}, \quad \sup_t |g_{B,\epsilon}'(t)|\le C
\end{equation} 
for each $B$, where $C$ is a nice constant. 
Indeed consider $h\in C^\infty$ such that it satisfies $|h(t)| \leq 2 |t|$ for each $t$ and 
\[
h(t) = 
     \begin{cases}
       t &\quad\text{if }t\in[-1,1]\\
       0 &\quad\text{if } t\notin [-2,2]. \\
     \end{cases}
\]
Then, it is easy to check that the family $(g_{B,\epsilon})_{B \geq 1}$ defined by 
$$
g_{B,\epsilon}(t) = \frac{B^\epsilon}{\sqrt{B}} \, h\left( \frac{\sqrt{B}}{B^\epsilon} \, (t-1) \right)
$$
satisfies the requirements \eqref{equation:g_B_bound}. \\

We define $\tilde{\zeta}_{2,\epsilon}$ by 
\[
    \tilde{\zeta}_{2,\epsilon} = \int_\Dcal\diff x \diff y\overline{\partial} \Phi_k(f)(z) \frac{1}{M} \sum_{m=1}^{M} (z \Q^2)_{mm} \, g_{B,\epsilon}\left(\frac{\|\x_m\|_2^2}{B+1}\right)
\]
and notice that $\zeta_2$ and $\tilde{\zeta}_{2,\epsilon}$ coincide on the exponentially high probability
event $A_\epsilon(\nu)$. We claim that if $| \tilde{\zeta}_{2,\epsilon} - \mathbb{E}(\tilde{\zeta}_{2,\epsilon})| \prec \frac{1}{B^{1-\epsilon}}$, then  $|\zeta_{2} - \mathbb{E}(\zeta_{2})| \prec \frac{1}{B^{1-\epsilon}}$. Since $\epsilon$ is arbitrary and $B^\epsilon=\Ocal(N^{\alpha\epsilon})$, Remark \ref{remark:domination_stochastique} will imply that $| \zeta_2 - \mathbb{E}(\zeta_2)| \prec \frac{1}{B}$. To justify this, we evaluate 
$\Pbb\left(|\tilde{\zeta}_{2,\epsilon} - \mathbb{E}(\tilde{\zeta}_{2,\epsilon})| > \frac{1}{B^{1-\epsilon}} N^{\delta} \right)$ for each $\delta > 0$. 
It holds that 
$$
\Pbb\left(|\zeta_2 - \mathbb{E}(\zeta_2)| > \frac{N^{\alpha \epsilon +\delta}}{B}\right) \leq \Pbb\left(|\zeta_2 - \mathbb{E}(\zeta_2)| > \frac{N^{\alpha \epsilon+\delta}}{B}, A_{\epsilon} \right) + \Pbb(A_{\epsilon}^{c}).
$$
As $\Pbb(A_{\epsilon}^{c})$ converges towards zero exponentially, we have just to consider 
$$
\Pbb\left(|\zeta_2 - \mathbb{E}(\zeta_2)| > \frac{N^{\alpha \epsilon+\delta}}{B}, A_{\epsilon} \right)
$$
and write, since $\zeta_2$ and $\tilde{\zeta}_{2,\epsilon}$ coincide on $A_{\epsilon}$, 
\begin{multline*}
\Pbb\left(|\zeta_2 - \mathbb{E}(\zeta_2)| > \frac{N^{\alpha \epsilon+\delta}}{B}, A_{\epsilon} \right) = \Pbb\left(|\tilde{\zeta}_{2,\epsilon} - \mathbb{E}(\zeta_2)| > \frac{N^{\alpha \epsilon+\delta}}{B}, A_{\epsilon} \right) \\
\leq 
\Pbb\left(|\tilde{\zeta}_{2,\epsilon} - \mathbb{E}(\tilde{\zeta}_{2,\epsilon})| > \frac{N^{\alpha \epsilon+\delta}}{B} - |E(\zeta_2 - \tilde{\zeta}_{2,\epsilon})|, A_{\epsilon} \right) .
\end{multline*}

We now prove that $|\mathbb{E}(\zeta_2 - \tilde{\zeta}_{2,\epsilon})|$ converges towards 0 exponentially. For this, we notice that as $\zeta_2$ and $\tilde{\zeta}_{2,\epsilon}$ coincide on $A_{\epsilon}$, then 
$$
|\mathbb{E}(\zeta_2 - \tilde{\zeta}_{2,\epsilon})| = 
\left|\mathbb{E}((\zeta_2 - \tilde{\zeta}_{2,\epsilon}) \mathbb{I}_{A_{\epsilon}^{c}})\right| \leq \left( \mathbb{E}\left|\zeta_2 - \tilde{\zeta}_{2,\epsilon}\right|^{2} \right)^{1/2} \, \left(\Pbb(A_{\epsilon}^{c})\right)^{1/2} .
$$
A rough evaluation of $ \left(\mathbb{E}\left|\zeta_2 - \tilde{\zeta}_{2,\epsilon}\right|^{2}\right)^{1/2}$ leads to 
$\left(\mathbb{E}\left|\zeta_2 - \tilde{\zeta}_{2,\epsilon}\right|^{2}\right)^{1/2} \leq C$
 for some nice constant $C$. Therefore, $\left( \mathbb{E}\left|\zeta_2 - \tilde{\zeta}_{2,\epsilon}\right|^{2} \right)^{1/2} \, \left(\Pbb(A_{\epsilon}^{c})\right)^{1/2}$, and thus $|\mathbb{E}(\zeta_2 - \tilde{\zeta}_{2,\epsilon})|$, converge towards $0$ exponentially. For each $N$ large enough, 
 we thus have 
\begin{align*}
    \Pbb\left(|\tilde{\zeta}_{2,\epsilon} - \mathbb{E}(\tilde{\zeta}_{2,\epsilon})| > \frac{N^{\alpha \epsilon+\delta}}{B} \right. &\left.- |\mathbb{E}(\zeta_2 - \tilde{\zeta}_{2,\epsilon})|, A_{\epsilon} \right) \\
    &\leq  \Pbb\left(|\tilde{\zeta}_{2,\epsilon} - \mathbb{E}(\tilde{\zeta}_{2,\epsilon})| > \frac{N^{\alpha \epsilon+\delta/2}}{B}, A_{\epsilon} \right) \\
    &\leq  \Pbb\left(|\tilde{\zeta}_{2,\epsilon} - \mathbb{E}(\tilde{\zeta}_{2,\epsilon})| > \frac{N^{\alpha \epsilon+\delta/2}}{B} \right) .
\end{align*}
We have therefore established that 
$$
\Pbb\left(|\zeta_2 - \mathbb{E}(\zeta_2)| > \frac{N^{\alpha \epsilon+\delta}}{B}, A_{\epsilon} \right) \leq  \Pbb\left(|\tilde{\zeta}_{2,\epsilon} - \mathbb{E}(\tilde{\zeta}_{2,\epsilon})| > \frac{N^{\alpha \epsilon+\delta/2}}{B} \right)
$$
which finally justifies that if $| \tilde{\zeta}_{2,\epsilon} - \mathbb{E}(\tilde{\zeta}_{2,\epsilon})| \prec \frac{B^\epsilon}{B}$, then  $|\zeta_{2} - \mathbb{E}(\zeta_{2})| \prec \frac{B^\epsilon}{B}$. \\

Therefore, it remains to prove that $|\tilde{\zeta}_{2,\epsilon} - \mathbb{E}(\tilde{\zeta}_{2,\epsilon})| \prec \frac{B^\epsilon}{B}$. This is true by Lemma \ref{lemma:controle_zeta_tilde} below. The stochastic domination relation  $ |\zeta_1-\Ebb\zeta_1| \prec \frac{B^\epsilon}{B}$ is proved similarly. This completes the proof of Lemma \ref{lemma:variance-zeta}. 
\end{proof}

\begin{lemma}
\label{lemma:controle_zeta_tilde}
$$|\tilde{\zeta}_{2,\epsilon} - \mathbb{E}(\tilde{\zeta}_{2,\epsilon})| \prec \frac{B^\epsilon}{B}.$$
\end{lemma}

\begin{proof}
In the following, we evaluate the norm square of the gradient of $\tilde{\zeta}_{2,\epsilon}$ w.r.t. the variables 
$X_{i,j}, X_{i,j}^{*}$ and just compute  $\sum_{i,j} \left|\frac{\partial\tilde{\zeta}_{2,\epsilon}}{\partial X_{ij}}\right|^{2}$
because  $\sum_{i,j} \left|\frac{\partial\tilde{\zeta}_{2,\epsilon}}{\partial X_{ij}^{*}}\right|^{2}$ is of the same order 
of magnitude. 

We recall that 
\begin{align}
\label{equation:derivative_Q_square}
    \frac{\partial (\Q^2)_{mm}}{\partial X_{ij}} = - \left(\frac{(\Q^2)_{mi}(\X^*\Q)_{jm}}{B+1} + \frac{Q_{mi}(\X^*\Q^2)_{jm}}{B+1}\right).
\end{align}


Moreover it is clear that
\begin{equation}
    \label{equation:derivative_g_B}
    \frac{\partial }{\partial X_{ij}}\left(g_{B,\epsilon}\left(\frac{\|\x_m\|_2^2}{B+1}\right)\right) = \delta_{im}\frac{\overline{X_{m,j}}}{B+1}g_{B,\epsilon}'\left(\frac{\|\x_m\|_2^2}{B+1}\right).
\end{equation}

Collecting the derivatives \eqref{equation:derivative_Q_square} and \eqref{equation:derivative_g_B}  we get after some algebra that 
\begin{multline}
\label{equation:derivative_Q^2_g_B}
    \frac{\partial}{\partial X_{ij}}\left(\sum_{m=1}^{M} (\Q^2)_{mm}g_{B,\epsilon}\left(\frac{\|\x_m\|_2^2}{B+1}\right)\right) = \frac{\overline{X_{i,j}}}{B+1}g_{B,\epsilon}'\left(\frac{\|\x_i\|_2^2}{B+1}\right)(\Q^2)_{ii} \\ -\sum_{m=1}^M g_{B,\epsilon}\left(\frac{\|\x_m\|_2^2}{B+1}\right) \left(\frac{(\Q^2)_{mi}(\X^*\Q)_{jm}}{B+1} + \frac{Q_{mi}(\X^*\Q^2)_{jm}}{B+1}\right).
\end{multline}

It remains to control $\sum_{i,j} \left|\frac{\partial\tilde{\zeta}_{2,\epsilon}}{\partial X_{ij}}\right|^{2}$. From the integral representation of $\tilde{\zeta}_{2,\epsilon}$, the derivative with respect to $X_{ij}$ is applied only on the integrand as follows: 
$$ \frac{\partial\tilde{\zeta}_{2,\epsilon}}{\partial X_{ij}} = \frac{1}{M} \, \int_\Dcal\diff x \diff y\, \overline{\partial}\Phi_k(f)(z)\frac{\partial}{\partial X_{ij}}\left(\sum_{m=1}^{M} z (\Q^2)_{mm} g_{B,\epsilon}\left(\frac{\|\x_m\|_2^2}{B+1}\right)\right).  $$
Plugging in the derivative computed in \eqref{equation:derivative_Q^2_g_B} we get:
\begin{align*}
    \frac{\partial\tilde{\zeta}_{2,\epsilon}}{\partial X_{ij}} &= \frac{1}{M} \, \int_\Dcal\diff x \diff y\, \overline{\partial}\Phi_k(f)(z) \, z \left\{\frac{\overline{X_{i,j}}}{B+1}g_{B,\epsilon}'\left(\frac{\|\x_i\|_2^2}{B+1}\right)(\Q^2)_{ii}\right. \\
    &\hskip1cm  \left.- \sum_{m=1}^M g_{B,\epsilon}\left(\frac{\|\x_m\|_2^2}{B+1}\right) \left(\frac{(\Q^2)_{mi}(\X^*\Q)_{jm}}{B+1} + \frac{Q_{mi}(\X^*\Q^2)_{jm}}{B+1}\right)\right\}.
\end{align*}

Using the bounds of $g_{B,\epsilon}$ and $g_{B,\epsilon}'$ from inequalities \eqref{equation:g_B_bound}, 
the observation that $g_{B,\epsilon}'(t) = 0$ if $|t-1| \geq \frac{2 B^{\epsilon}}{\sqrt{B}}$, 
and that $|z|$ is bounded on $\Dcal$, one can write:
\begin{align*}
    \left|\frac{\partial\tilde{\zeta}_{2,\epsilon}}{\partial X_{ij}}\right|^2 &\le \frac{C}{M^{2}} \, \int_\Dcal \diff x \diff y \left|\overline{\partial}\Phi_k(f)(z)\right|^2 \left|\frac{\overline{X_{i,j}}}{B+1}g_{B,\epsilon}'\left(\frac{\|\x_i\|_2^2}{B+1}\right)(\Q^2)_{ii} \right.\\
    &\hskip0.5cm- \left.\sum_{m=1}^Mg_{B,\epsilon}\left(\frac{\|\x_m\|_2^2}{B+1}\right)\left(\frac{(\Q^2)_{mi}(\X^*\Q)_{jm}}{B+1} + \frac{Q_{mi}(\X^*\Q^2)_{jm}}{B+1}\right)\right|^2 \\
    &\le \frac{C}{M^{2}} \, \int_\Dcal \diff x \diff y \left|\overline{\partial}\Phi_k(f)(z)\right|^2  \left|\frac{\overline{X_{i,j}}}{B+1}(\Q^2)_{ii}\right|^2 \, \mathds{1} \left(\left|\frac{\|\x_i\|_2^2}{B+1} - 1\right| \leq  \frac{2 B^{\epsilon}}{\sqrt{B}}\right) \\
    &\hskip0.5cm+  \frac{C}{M^{2}} \,\int_\Dcal \diff x \diff y \left|\overline{\partial}\Phi_k(f)(z)\right|^2 \left(\frac{B^\epsilon}{\sqrt{B}}\right)^2\left|\sum_{m=1}^M\frac{\left|(\Q^2)_{mi}\right|\left|(\X^*\Q)_{jm}\right|}{B+1}\right|^2 \\
    &\hskip0.5cm+  \frac{C}{M^{2}} \, \int_\Dcal \diff x \diff y \left|\overline{\partial}\Phi_k(f)(z)\right|^2\left(\frac{B^\epsilon}{\sqrt{B}}\right)^2\left|\sum_{m=1}^M\frac{\left|Q_{mi}\right|\left|(\X^*\Q^2)_{jm}\right|}{B+1}\right|^2 \\
    &:=  \frac{C}{M^{2}} \, (T_{ij}^{(1)}+T_{ij}^{(2)}+T_{ij}^{(3)}).
\end{align*}

It remains to sum over $i,j$. 
\begin{align*}
    &\sum_{i,j=1}^M T_{ij}^{(1)} \\
    &= \int_\Dcal \diff x \diff y \left|\overline{\partial}\Phi_k(f)(z)\right|^2  \sum_{i,j=1}^M\left|\frac{\overline{X_{i,j}}}{B+1}(\Q^2)_{ii}\right|^2 \,  \mathds{1} \left(\left|\frac{\|\x_i\|_2^2}{B+1} - 1\right| \leq  \frac{2 B^{\epsilon}}{\sqrt{B}}\right)\\
    &\le \int_\Dcal \diff x \diff y \left|\overline{\partial}\Phi_k(f)(z)\right|^2  \sum_{i=1}^M|(\Q^2)_{ii}|^2 \,  \mathds{1} \left(\left|\frac{\|\x_i\|_2^2}{B+1} - 1\right| \leq  \frac{2 B^{\epsilon}}{\sqrt{B}}\right) \,  \sum_{j=1}^M\left|\frac{\overline{X_{i,j}}}{B+1}\right|^2 \\
    &= \frac{C}{B+1}\int_\Dcal \diff x \diff y \left|\overline{\partial}\Phi_k(f)(z)\right|^2  \sum_{i=1}^M|(\Q^2)_{ii}|^2 \frac{\|\x_i\|_2^2}{B+1}  \, \mathds{1} \left(\left|\frac{\|\x_i\|_2^2}{B+1} - 1\right| \leq  \frac{2 B^{\epsilon}}{\sqrt{B}}\right)\\
    &\le \frac{C}{B+1}(1+\frac{2B^\epsilon}{\sqrt{B}})\int_\Dcal \diff x \diff y \left|\overline{\partial}\Phi_k(f)(z)\right|^2  \sum_{i=1}^M|(\Q^2)_{ii}|^2.
\end{align*}
Since 
$$ \sum_{i=1}^M|(\Q^2)_{ii}|^2 \le M \|\Q\|^4$$
it can be written that:
$$ \sum_{i,j=1}^M T_{ij}^{(1)} \le C\frac{M}{B+1}\int_\Dcal \diff x \diff y \left|\overline{\partial}\Phi_k(f)(z)\right|^2 \|\Q\|^4 . $$

Inspecting $T_{ij}^{(2)}$, one can see that by Jensen's inequality
\[
    \left|\sum_{m=1}^M \left|(\Q^2)_{mi}\right|\left|(\X^*\Q)_{jm}\right|\right|^2 \le M\sum_{m=1}^M \left|(\Q^2)_{mi}\right|^2\left|(\X^*\Q)_{jm}\right|^2
\]
so summing over $i$ and $j$ provides:
\begin{multline*}
    \sum_{i,j=1}^MT_{ij}^{(2)} \le \frac{B^{2\epsilon}M}{(B+1)^3} \\ \times  \int_\Dcal \diff x \diff y \left|\overline{\partial}\Phi_k(f)(z)\right|^2 \sum_{m=1}^M \left(\sum_{i=1}^M\left|(\Q^2)_{mi}\right|^2\right)\left(\sum_{j=1}^M\left|(\X^*\Q)_{jm}\right|^2\right).
\end{multline*}
Notice that since $\sum_{i=1}^M\left|(\Q^2)_{mi}\right|^2$ is the square euclidean norm of line $m$ of $\Q^2$:
\[
    \sum_{i=1}^M\left|(\Q^2)_{mi}\right|^2 \le \|\Q^2\|^2 \le \|\Q\|^4 .
\]
Moreover, 
\begin{multline*}
    \sum_{m=1}^M\left(\sum_{j=1}^M\left|(\X^*\Q)_{jm}\right|^2\right) = \Tr \X^*\Q\Q^*\X \\ = (B+1)\Tr\left(\left(\I+z\Q\right)\Q^*\right) \le M(B+1) (\|\Q\| + |z|\|\Q\|^2)
\end{multline*}
therefore
\[
    \sum_{i,j=1}^M T_{ij}^{(2)} \le B^{2\epsilon}\left(\frac{M}{B+1}\right)^2 \int_\Dcal \diff x \diff y \left|\overline{\partial}\Phi_k(f)(z)\right|^2 \|\Q^4\| (\|\Q\| + |z|\|\Q\|^2)
\]

and similarly for $T_{ij}^{(3)}$ one gets:
\[
    \sum_{i,j=1}^M T_{ij}^{(3)} \le B^{2\epsilon}\left(\frac{M}{B+1}\right)^2 \int_\Dcal \diff x \diff y \left|\overline{\partial}\Phi_k(f)(z)\right|^2 \|\Q^2\| (\|\Q\|^3 + |z|\|\Q\|^4).
\]

Collecting the terms in $T_{ij}^{(1)}, T_{ij}^{(2)}$ and $T_{ij}^{(3)}$, and since $M/(B+1)=\Ocal(1)$ by Assumption \ref{assumption:rate_NBM}, we can write:
$$ \sum_{i,j} \left|\frac{\partial\tilde{\zeta}_{2,\epsilon}}{\partial X_{ij}}\right|^{2}\le \frac{C}{M^{2}} B^{2\epsilon}\int_\Dcal \diff x \diff y \left|\overline{\partial}\Phi_k(f)(z)\right|^2 (\|\Q\|^4+\|\Q^5\|+|z|\|\Q^6\|)$$
As $\|\Q\|^4+\|\Q^5\|+|z|\|\Q^6\| \leq C(z)$, we obtain that for $k$ large enough, 
$$ \sum_{i,j} \left|\frac{\partial\tilde{\zeta}_{2,\epsilon}}{\partial X_{ij}}\right|^{2} = \Ocal\left(\frac{B^{2\epsilon}}{B^{2}}\right)$$
as expected. 
\end{proof}

It remains to study $\Ebb[\zeta]$, and establish the following Lemma. 
\begin{lemma}
\label{lemma:esperance-zeta}
$$ \left|\Ebb\zeta\right| = \Ocal\left(\frac{1}{B}\right).$$
\end{lemma}
\begin{proof}

As in the proof of Lemma \ref{lemma:variance-zeta}, we only consider 
\[
    \Ebb[\zeta_2] =  \int_\Dcal \diff x \diff y\, \overline{\partial}\Phi_k(f)(z) \frac{1}{M} \sum_{m=1}^M z\Ebb\left[(\Q^2)_{mm} \left(\frac{\|\x_m\|_2^2}{B+1}-1\right)\right]
\]
as $\Ebb[\zeta_1]$ is shown to be also $\Ocal(\frac{1}{B})$ with the same argument. As $\Ebb\left[\frac{\|\x_m\|_2^2}{B+1}-1\right] = 0$, we have
$$
    \Ebb\left[(\Q^2)_{mm} \left(\frac{\|\x_m\|_2^2}{B+1}-1\right)\right] =  \Ebb\left[\left((\Q^2)_{mm}-\Ebb[(\Q^2)_{mm}]\right) \left(\frac{\|\x_m\|_2^2}{B+1}-1\right)\right] .
$$
Apply now the Cauchy-Schwartz inequality:
\begin{equation}
\label{equation:controle_esperance_zeta_2}
    \left|\Ebb[\zeta_2]\right|  \le \int_\Dcal \diff x \diff y\, |\overline{\partial}\Phi_k(f)(z)| \frac{1}{M} \sum_{m=1}^M |z|\sqrt{\Var(\Q^2)_{mm}}\sqrt{\Ebb\left|\frac{\|\x_m\|_2^2}{B+1}-1\right|^2}.
\end{equation}
As it is clear that $\Ebb\left|\frac{\|\x_m\|_2^2}{B+1}-1\right|^2 = \Ocal\left(\frac{1}{B}\right)$, it remains to control $\Var (\Q^2)_{mm} = \Var (\Tr \Q^2 \e_m \e_m^T)$
where $(\e_m)_{m=1,\ldots,M}$ is the canonical basis of $\Cbb^M$. A direct application 
of (\ref{eq:var-trace-resolvent}) for $i=2$ leads immediately to 
\begin{equation}
\label{equation:controle_var_Q_square}
    \Var (\Q^2)_{mm} \le \frac{C(z)}{B}
\end{equation}
for some nice constant C. Using \eqref{equation:controle_var_Q_square} in \eqref{equation:controle_esperance_zeta_2},  we get that for $k$ large enough: 
\begin{align*}
    \left|\Ebb\zeta_2\right| &\le  \frac{1}{B} \int_\Dcal \diff x \diff y\, |\overline{\partial}\Phi_k(f)(z)| \sqrt{C(z)} \\
    & \le  \frac{1}{B} \int_\Dcal \diff x \diff y\, |\overline{\partial}\Phi_k(f)(z)|  (1 + C(z)) \leq C \frac{1}{B}.
\end{align*}
This completes the proof of Lemma \ref{lemma:esperance-zeta}. 
\end{proof}
\begin{remark}
\label{re:decomposition-fondamentale}
We notice that, instead of using (\ref{eq:fundamental-decomposition}), an alternative approach to study $\frac{1}{M} \Tr f(\hat{\C}(\nu)) - \int f \, d\mu_{MP}^{(c_N)}$ could have been based on the decomposition
\begin{align}
\nonumber \frac{1}{M} \mathrm{Tr}\left( f(\hat{\C}(\nu))\right) -  \int_{\Rbb^{+}}f\diff\mu_{MP}^{(c_N)} = \frac{1}{M} \mathrm{Tr}\left( f(\hat{\C}(\nu))\right) - \mathbb{E} \left[ \frac{1}{M} \mathrm{Tr}\left( f(\hat{\C}(\nu))\right) \right]+ \\ \nonumber \mathbb{E} \left[\frac{1}{M} \mathrm{Tr}\left( f(\hat{\C}(\nu))\right) \right] - \mathbb{E} \left[ \frac{1}{M} \mathrm{Tr}\left( f(\tilde{\C}(\nu))\right) \right] + \\
\nonumber \mathbb{E} \left[ \frac{1}{M} \mathrm{Tr}\left( f(\tilde{\C}(\nu))\right) 
- \frac{1}{M} \mathrm{Tr}\left( f(\frac{\X(\nu)\X^*(\nu)}{B+1}) \right) \right] + \\ 
\label{eq:-alternative-fundamental-decomposition}
\mathbb{E} \left[  \frac{1}{M} \mathrm{Tr}\left( f(\frac{\X(\nu)\X^*(\nu)}{B+1}) \right) \right] - \int_{\Rbb^{+}}f\diff\mu_{MP}^{(c_N)}.    
\end{align}
The first term of the r.h.s. of (\ref{eq:-alternative-fundamental-decomposition}) can be addressed using the Gaussian concentration inequality. However, the calculations are more complicated than the evaluation of $\frac{1}{M} \mathrm{Tr}\left( f(\tilde{\C}(\nu))\right) - \mathbb{E} \left[ \frac{1}{M} \mathrm{Tr}\left( f(\tilde{\C}(\nu))\right) \right]$ because, considered as a function of $(\X, \X^*)$, $\frac{1}{M} \mathrm{Tr}\left( f(\hat{\C}(\nu))\right)$ is not a Lipschitz function. Using techniques similar to those developed to evaluate $\zeta - \mathbb{E}(\zeta)$ (see Lemma \ref{lemma:variance-zeta}), it could however be shown that 
\begin{equation}
    \label{eq:new-stochastic-domination}
\left|\frac{1}{M} \mathrm{Tr}\left( f(\hat{\C}(\nu))\right) - \mathbb{E} \left[ \frac{1}{M} \mathrm{Tr}\left( f(\hat{\C}(\nu))\right) \right]\right| \prec \frac{1}{B}.
\end{equation}
In order to evaluate the second term of the r.h.s. of (\ref{eq:-alternative-fundamental-decomposition}), 
one should prove that 
\begin{multline}
\label{eq:expectation-reduction}
    \mathbb{E} \left[ \int_{\Dcal} \diff x \diff y\, \left(\bar{\partial}\Phi_k(f)(z) \frac{1}{M} \Tr\{\hat{\Q}-\tilde{\Q}\} - \tilde{r}_N(\nu) \; \tilde{p}_N(z) \; v_N \; \mathbf{1}_{\alpha > 2/3}\right) - \zeta\right] \\= \Ocal(u_N)
\end{multline}
and $\mathbb{E}(\zeta) = \Ocal(\frac{1}{B})$. The proof of (\ref{eq:expectation-reduction}) does not appear simpler 
than the proof of (\ref{eq:reduction}): the 3 steps that allowed to establish (\ref{eq:reduction}) should still be used, except that the stochastic domination properties should be replaced by properties of the mathematical expectation of the various terms. However, proving stochastic domination appears simpler than showing the desired properties of the above mathematical expectations. In sum, while the use of decomposition 
(\ref{eq:-alternative-fundamental-decomposition}) allows to avoid 
Lemma \ref{lemma:variance-zeta}, the justification of (\ref{eq:new-stochastic-domination}) needs to develop tools that are similar to those of Lemma \ref{lemma:variance-zeta}, and the proof of (\ref{eq:expectation-reduction}) tends to be more complicated than the
proof of (\ref{eq:reduction}). This explains why we have chosen to use decomposition 
(\ref{eq:fundamental-decomposition}) rather than (\ref{eq:-alternative-fundamental-decomposition}). 
\end{remark}
\subsection{Step 4: evaluation of \texorpdfstring{$\mathbb{E} \left[ \frac{1}{M} \mathrm{Tr}\left( f(\tilde{\C}(\nu))\right) 
- \frac{1}{M} \mathrm{Tr}\left( f(\frac{\X(\nu)\X^*(\nu)}{B+1}) \right) \right] $}{}}
The Helffer-Sjöstrand formula implies that 
\begin{multline*}
\mathbb{E} \left[ \frac{1}{M} \mathrm{Tr}\left( f(\tilde{\C}(\nu))\right) 
- \frac{1}{M} \mathrm{Tr}\left( f(\frac{\X(\nu)\X^*(\nu)}{B+1}) \right) \right] \\ =\frac{1}{\pi} \Re \int_\Dcal \diff x \diff y\, \overline{\partial}\Phi_k(f)(z) 
\mathbb{E} \left[ \frac{1}{M} \Tr ( \tilde{\Q}_N(z) - \Q_N(z) ) \right].
\end{multline*}
Therefore, we are back to evaluate $\mathbb{E} \left[ \frac{1}{M} \Tr ( \tilde{\Q}_N(z) - \Q_N(z) )\right]$. \\

In order to simplify the exposition of the results of this paragraph, we introduce the following notation. If $(h_N(z))_{N \geq 1}$ is a sequence of complex-valued functions defined on $\mathbb{C}^{+}$ and if $(w_N)_{N \geq 1}$ is a sequence of positive real numbers, the notation $h_N(z) = \Ocal_z(w_N)$ means that there exists two nice polynomials $P_1$ and $P_2$ such that 
$|h_N(z)| \leq w_N P_1(|z|) P_2(\frac{1}{\Im z})$ for each $z \in \mathbb{C}^{+}$. \\

In this paragraph, we establish the following Proposition. 
\begin{proposition}
\label{prop:EtildeQ-EQ}
 $\mathbb{E} \left[ \frac{1}{M} \Tr ( \tilde{\Q}_N(z) - \Q_N(z) ) \right]$ can be written as 
 \begin{multline}
\mathbb{E} \left[ \frac{1}{M} \Tr ( \tilde{\Q}_N(z) - \Q_N(z) ) \right] = \left( \frac{1}{M} \sum_{m=1}^{M} \frac{s_m'}{s_m} \right)^{2} \, p_N(z) \, v_N -  \\
\left( \frac{1}{2M} \sum_{m=1}^{M} \frac{s_m''}{s_m} \right) \, \tilde{p}_N(z) \, v_N + \Ocal_z\left( \left(\frac{B}{N}\right)^{3} + \frac{1}{N} \right)
\label{eq:EtildeQ-EQ}.
 \end{multline}
\end{proposition}
The Helffer-Sjöstrand formula thus leads to the following Corollary:
\begin{corollary}
\label{coro:Eftilde-EfXX}
$\mathbb{E} \left[ \frac{1}{M} \mathrm{Tr}\left( f(\tilde{\C}(\nu))\right) 
- \frac{1}{M} \mathrm{Tr}\left( f(\frac{\X(\nu)\X^*(\nu)}{B+1}) \right) \right]$ is given by
\begin{multline}
    \label{eq:expre-Eftilde-EfXX}
\mathbb{E} \left[ \frac{1}{M} \mathrm{Tr}\left( f(\tilde{\C}(\nu))\right) 
- \frac{1}{M} \mathrm{Tr}\left( f(\frac{\X(\nu)\X^*(\nu)}{B+1}) \right) \right] = \\
\left( \frac{1}{M} \sum_{m=1}^{M} \frac{s_m'}{s_m} \right)^{2} \, \phi_N(f) \, v_N - \left( \frac{1}{2M} \sum_{m=1}^{M} \frac{s_m''}{s_m} \right) \, \tilde{\phi}_N(f) \, v_N + \Ocal\left(\left(\frac{B}{N}\right)^{3} + \frac{1}{N}\right).
\end{multline}
\end{corollary}
Corollary \ref{coro:Eftilde-EfXX} first implies that $\mathbb{E} \left[ \frac{1}{M} \mathrm{Tr}\left( f(\tilde{\C}(\nu))\right) 
- \frac{1}{M} \mathrm{Tr}\left( f(\frac{\X(\nu)\X^*(\nu)}{B+1}) \right) \right]$
is $\Ocal\left( \frac{B}{N} \right)^{2}$, a result which is not a priori obvious. In particular, the stochastic representation (\ref{equation:tildeC_approximation_Wishart}) of the matrix $\tilde{\C}$ can be shown to provide the more pessimistic $\Ocal(\frac{B}{N})$ rate of convergence. The comparison of (\ref{eq:expre-Eftilde-EfXX}) with (\ref{eq:evaluation-trace-hatC-tildeC}) also leads to the conclusion that if $\alpha > 2/3$, the dominant $\Ocal\left( \frac{B}{N} \right)^{2}$  deterministic term of $\frac{1}{M} \Tr( f(\hat{\C}(\nu)) - f(\tilde{\C}(\nu))$ is cancelled by the second term  of the righthandside of (\ref{eq:expre-Eftilde-EfXX}), thus explaining the structure of the $\Ocal\left( \frac{B}{N} \right)^{2}$  deterministic correction of $\frac{1}{M} \Tr( f(\hat{\C}(\nu)) - \int f \, d\mu_{MP}^{(c_N)}$. In particular, establishing 
(\ref{eq:EtildeQ-EQ}) (and thus (\ref{eq:expre-Eftilde-EfXX})) will complete the proof of 
Theorem \ref{theo:domination-psi}. \\

\begin{proof} 
The proof of (\ref{eq:EtildeQ-EQ}) is based on the Gaussian tools reviewed in  
Paragraph  \ref{subsubsec:resolvent-MP}, and needs long and very tedious calculations. 
Therefore, we just provide a sketch of proof. In particular, we justify that 
$\mathbb{E} \left[ \frac{1}{M} \Tr ( \tilde{\Q}_N(z) - \Q_N(z) ) \right]$ is a 
$\Ocal_z\left(\frac{B}{N}\right)^{2}$ term, but do not establish its expression
(\ref{eq:EtildeQ-EQ}). \\

The starting point of the proof is to express $\tilde{\Q} - \Q$ as 
$$
\tilde{\Q} - \Q = - \tilde{\Q} \tilde{\Deltabs} \Q = -\Q \tilde{\Deltabs} \Q + 
\Q \tilde{\Deltabs} \Q \tilde{\Deltabs} \Q - \tilde{\Q} \tilde{\Deltabs} \Q \tilde{\Deltabs} \Q \tilde{\Deltabs} \Q.
$$
Therefore, $\mathbb{E} \left[ \frac{1}{M} \Tr ( \tilde{\Q}_N(z) - \Q_N(z) ) \right]$ can be written as  
\begin{multline}
\mathbb{E} \left[ \frac{1}{M} \Tr ( \tilde{\Q} - \Q ) \right] = 
- \mathbb{E} \left[ \frac{1}{M} \Tr (\Q^{2} \tilde{\Deltabs}) \right] + \mathbb{E} \left[ \frac{1}{M} \Tr (\Q^{2} \tilde{\Deltabs} \Q  \tilde{\Deltabs}) \right] \\ - \mathbb{E} \left[ \frac{1}{M} \Tr (\tilde{\Q} \tilde{\Deltabs} \Q  \tilde{\Deltabs} \Q \tilde{\Deltabs}) \right].
 \label{eq:expre-EtildeQ-Q-1}
\end{multline}
It is clear that the moduli of the second and third terms of the right hand side 
of (\ref{eq:expre-EtildeQ-Q-1}) are controlled by $C(z) \mathbb{E}(\| \tilde{\Deltabs}\|^{2})$ 
and  $C(z) \mathbb{E}(\| \tilde{\Deltabs}\|^{3})$ respectively. We now state the following useful Lemma, proved in the Appendix, which implies that these terms are $\Ocal_z\left(\frac{B}{N}\right)^{2}$ and $\Ocal_z\left(\frac{B}{N}\right)^{3}$
respectively. 
\begin{lemma}
\label{le:EnormtildeDelta}
For each $k \geq 1$, there exists a nice constant $C$ depending on $k$ such that $\Ebb \left(\| \tilde{\Deltabs} \|^{k}\right) \leq C \, \left(\frac{B}{N}\right)^{k}$
\end{lemma}
In order to prove that $\mathbb{E} \left[ \frac{1}{M} \Tr ( \tilde{\Q}_N(z) - \Q_N(z) ) \right] = \Ocal_z\left(\frac{B}{N}\right)^{2}$, we thus have to check that 
\begin{equation}
    \label{eq:magnitude-trace-Q2tildeDelta}
    \mathbb{E} \left[ \frac{1}{M} \Tr \Q^{2} \tilde{\Deltabs} \right] = \Ocal_z\left(\frac{B}{N}\right)^{2}.
\end{equation}
For this, we first express $ \mathbb{E} \left[ \frac{1}{M} \Tr \Q^{2} \tilde{\Deltabs} \right]$ 
as 
\begin{multline*}
 \mathbb{E} \left[ \frac{1}{M} \Tr \Q^{2} \tilde{\Deltabs} \right] = 
 \mathbb{E}  \left( \frac{1}{M} \Tr \Q^{2} \frac{\Gammabs \X^*}{B+1} \right) +
  \mathbb{E}  \left( \frac{1}{M} \Tr \Q^{2} \frac{\X \Gammabs^*}{B+1} \right) \\ + 
 \mathbb{E}  \left( \frac{1}{M} \Tr \Q^{2} \frac{\Gammabs \Gammabs^*}{B+1} \right).
\end{multline*}
The third term of the right hand side is clearly $\Ocal_z((\frac{B}{N})^{2})$. We thus need to check that the first two terms are also $\Ocal_z((\frac{B}{N})^{2})$. We just verify this property for the first term. For this, we evaluate $\mathbb{E}  \left( \frac{1}{M} \Tr \Q \frac{\Gammabs \X^*}{B+1} \right)$ using the Gaussian tools, and take the derivative w.r.t. $z$ to obtain the expression 
of $\mathbb{E}  \left( \frac{1}{M} \Tr \Q^{2} \frac{\Gammabs \X^*}{B+1} \right)$. \\

In order to simplify the notations, we denote by $\W$ the matrix $\W = \frac{\X}{\sqrt{B+1}}$, and 
denote by $\w_1 = \frac{\x_1}{\sqrt{B+1}}, \ldots, \w_M = \frac{\x_M}{\sqrt{B+1}}$ its $M$ rows. 
In particular, the row $m$ of the matrix $\frac{\Gammabs}{\sqrt{B+1}}$ coincides with $\w_m \Psibs_m$ where we recall that 
matrix $\Psibs_m$ is defined by (\ref{eq:expre-IplusPhi1/2}). If $(\e_1, \ldots, \e_{m})$ represents the canonical basis of $\mathbb{C}^{M}$,  $\mathbb{E}  \left( \frac{1}{M} \Tr \Q \frac{\Gammabs \X^*}{B+1} \right)$ can be written as 
$$
\mathbb{E}  \left( \frac{1}{M} \Tr \Q \frac{\Gammabs \X^*}{B+1} \right) = \frac{1}{M} \sum_{m=1}^{M} \mathbb{E}\left( \w_m \Psibs_m \W^* \Q \e_m \right).
$$
We now state the following Lemma whose proof is given in Appendix. We recall 
that $\beta_N(z) = \mathbb{E}((\Q_N(z))_{mm}$ for each $m$. 
\begin{lemma}
\label{le:EwAWQ}
If $\A$ represents a $(B+1) \times (B+1)$ matrix, the following equality holds
\begin{multline}
    \label{eq:expre-EwAW*Q}
\mathbb{E}\left( \w_m \A \W^* \Q \e_m \right) = \frac{\beta}{1 + \beta c} \frac{1}{B+1} \Tr \A 
- \mathbb{E} \left[ \left(\frac{1}{B+1} \Tr \W \A \W^* \Q\right)^{\circ} \, \Q^{\circ}_{m,m} \right]+ \\
\frac{\beta c}{1+ \beta c} \mathbb{E} \left[ \left(\frac{1}{B+1} \Tr \W \A \W^* \Q\right)^{\circ} \, \frac{1}{B+1} \Tr \Q^{\circ} \right].
\end{multline}
\end{lemma}
Using (\ref{eq:control-inverse-1pluss}) in the case $\s_{\mu}(z) = \beta(z)$ as well as 
(\ref{eq:beta-t}), we easily obtain that $\frac{\beta}{1 + \beta c} = \frac{t}{1 + c t} + \epsilon_1(z) = -z t(z) \tilde{t}(z) + \epsilon_1(z)$ where $\epsilon_1(z) = \Ocal_z(\frac{1}{B^{2}})$. Moreover, it follows from  (\ref{eq:control-derivee-epsilon}) that $\epsilon_1'(z)$ is also a  $\Ocal_z(\frac{1}{B^{2}})$. We now use (\ref{eq:expre-EwAW*Q}) for $\A = \Psibs_m$, and 
differentiate (\ref{eq:expre-EwAW*Q}) for $\A = \Psibs_m$ w.r.t. $z$. Using the Schwartz inequality 
the inequalities (\ref{eq:var-trace-resolvent}) and (\ref{eq:var-trace-resolvent-XX*}), and (\ref{eq:controle-norm-Psi}), we obtain immediately that 
$$
\mathbb{E}\left( \w_m \Psi_m \W^* \Q \e_m \right) = -\left(z t(z) \tilde{t}(z)\right)' \, 
\frac{1}{B+1} \Tr \Psibs_m + \Ocal_z\left(\frac{B}{N}\right)^{3} +  \Ocal_z\left(\frac{1}{\sqrt{B}N}\right)
$$
and that 
\begin{multline*}
    \mathbb{E}  \left( \frac{1}{M} \Tr \Q^{2} \frac{\Gammabs \X^*}{B+1} \right) = -\left(z t(z) \tilde{t}(z)\right)' \, \frac{1}{B+1} \Tr \left( \frac{1}{M} \sum_{m=1}^{M} \Psibs_m \right) \\ + \Ocal_z\left(\frac{B}{N}\right)^{3} +  \Ocal_z\left(\frac{1}{\sqrt{B}N}\right).
\end{multline*}
It is easily checked that 
$$
\Psibs_m = (\I + \Phibs_m)^{1/2} - \I = \frac{1}{2} \Phibs_m - \frac{1}{8} \Phibs_m^{2} + \Xibs_m
$$
where $\|  \Xibs_m \| \leq C \left( \frac{B}{N} \right)^{3}$. It is easily seen 
that $\frac{1}{B+1} \Tr \Phibs_m^{2} = \left(\frac{s_m'}{s_m}\right)^{2} v_N + \Ocal((\frac{B}{N})^3 + \frac{1}{N})$. Using (\ref{eq:expre-Trace-Phim}), we thus obtain that  
\begin{multline*}
\mathbb{E}  \left( \frac{1}{M} \Tr \Q^{2} \frac{\Gammabs \X^*}{B+1} \right) =  -(z t_N(z) \tilde{t}_N(z))' \, \left( \frac{1}{M} \sum_{m=1}^{M} (\frac{s_m''}{2 s_m} - \frac{(s_m')^{2}}{8 (s_m)^{2}}) \right) v_N \\ + \Ocal_z(u_N) 
\end{multline*}
because 
$$
\Ocal_z\left(\frac{B}{N}\right)^{3} +  \Ocal_z\left(\frac{1}{\sqrt{B}N}\right) + \Ocal_z\left(\frac{1}{N}\right) = \Ocal_z(u_N).
$$
We have thus established that $\mathbb{E}  \left( \frac{1}{M} \Tr \Q^{2} \frac{\Gammabs \X^*}{B+1} \right)$ is a $\Ocal_z\left(\frac{B}{N}\right)^{2}$ term, and have evaluated the 
corresponding principal term. Using similar calculations, we can obtain easily the expression 
of the $\Ocal_z\left(\frac{B}{N}\right)^{2}$ term of  $\mathbb{E}  \left( \frac{1}{M} \Tr \Q^{2} \tilde{\Deltabs} \right)$. In order to establish (\ref{eq:expre-Eftilde-EfXX}), it is necessary to evaluate  the  $\Ocal_z\left(\frac{B}{N}\right)^{2}$ term of  $\mathbb{E}  \left( \frac{1}{M} \Tr \Q \tilde{\Deltabs} \Q \tilde{\Deltabs} \Q \right)$. This step needs very long  calculations that are omitted. 
\end{proof}
\subsection{Estimation of \texorpdfstring{$r_N(\nu)$}{}}
\label{subsec:estimation-rN}
The term $\sup_{\nu}|\psi_{N}(f,\nu)|$ depends on the unknown true spectral densities 
$(s_m)_{m=1, \ldots, M}$ through the term $r_N(\nu)$ defined by (\ref{def-eq-rN}). 
In order to be able to use Theorem \ref{theo:domination-psi} in practice, it appears 
necessary to estimate $r_N(\nu)$ by an accurate enough estimate $\hat{r}_N(\nu)$, and to replace $\psi_N(f,\nu)$ by $\hat{\psi}_N(f,\nu)$ defined by 
\begin{equation}
    \label{eq:def-hatpsi}
 \hat{\psi}_N(f,\nu) = \frac{1}{M} \mathrm{Tr}\left( f(\hat{\C}(\nu))\right) -  \int_{\Rbb^{+}}f\diff\mu_{MP}^{(c_N)}  -  \hat{r}_N(\nu)  \; \phi_N(f) \; v_N \; \mathbf{1}_{\alpha > 2/3}  
\end{equation}
$\hat{r}_N(\nu)$ has to be chosen in such a way that $|\hat{\psi}_N(f,\nu)| \prec u_N$, 
a condition that will be satisfied if $|\hat{r}_N(\nu) - r_N(\nu)| \prec \frac{u_N}{v_N}$ if $\alpha > \frac{2}{3}$. A natural choice for $\hat{r}_N(\nu)$ would be to replace the true spectral densities $(s_m)_{m=1, \ldots, M}$ by their frequency smoothed estimates $(\hat{s}_m)_{m=1, \ldots, M}$ defined by (\ref{eq:def-hatsm}), and the derivatives $(s_m')_{m=1, \ldots, M}$
by $(\hat{s}_m')_{m=1, \ldots, M}$. However, $\hat{s}_m'$ is not an accurate estimate of $s_m'$ so that the corresponding estimate of $r_N(\nu)$ does not satisfy $|\hat{r}_N(\nu) - r_N(\nu)| \prec \frac{u_N}{v_N}$ if $\alpha > \frac{2}{3}$. If $L<N$ is an integer, we introduce the lag window estimator $\hat{s}_{m,L}$ of $s_m$ defined by 
\begin{equation}
    \label{eq:def-hatsLm}
  \hat{s}_{m,L}(\nu) = \int_0^{1} |\xi_{\y_m}(\mu) |^{2} w_{L}(\nu - \mu) d\mu =  \sum_{l=-L}^{L} \hat{r}_{m,l} \, e^{-2i\pi l \nu}  
\end{equation}
where $w_L(\nu)=\sum_{l=-L}^{L} e^{-2 i \pi  l \nu}$ is the Fourier transform of the rectangular window and $\hat{r}_{m,l}$ represents the biased estimate of the autocovariance coefficient $r_{m,l}$ of $y_{m}$ at lag $l$ defined by 
\begin{equation}
    \label{eq:def-hatrml}
 \hat{r}_{m,l} = \frac{1}{N} \sum_{n=1}^{N-l} y_{m,n+l}y_{m,n}^*  
\end{equation}
and $\hat{r}_{m,-l} = \hat{r}_{m,l}^*$ for $l \geq 0$. Then, the following result holds. 
\begin{proposition}
\label{prop:hatrNL}
Assume that $L = L(N) = \Ocal(N^{\frac{1}{2\gamma_0+1}})$, where $\gamma_0 \geq 3$ is defined
by (\ref{eq:condition-R}). Then, the estimate $\hat{r}_N(\nu)$ defined by 
\begin{equation}
    \label{eq:def-hatrNL}
 \hat{r}_N(\nu) = \left( \frac{1}{M} \sum_{m=1}^{M} \frac{\hat{s}'_{m,L}(\nu)}{\hat{s}_{m,L}(\nu)} \right)^{2}   
\end{equation}
satisfies 
\begin{eqnarray}
\label{eq:domination-hatrN}
|\hat{r}_N(\nu) - r_N(\nu)| & \prec  & \frac{1}{N^{(\gamma_0 -1)/(2\gamma_0+1)}} \\
\label{eq:domination-bis-hatrN}
|\hat{r}_N(\nu) - r_N(\nu)| & \prec &\frac{u_N}{v_N} \; \mbox{if $\alpha > \frac{2}{3}$}
\end{eqnarray}
as well as 
\begin{equation}
    \label{eq:domination-hatpsi}
    |\hat{\psi}_N(f,\nu)| \prec u_N.
\end{equation}
\end{proposition}
\begin{proof}
We denote by  $\d_N(\nu)$ the $N$--dimensional vector defined by $\d_N(\nu) = (1, e^{-2 i \pi \nu}, \ldots, e^{-2 i \pi (N-1) \nu})^{T}$. We recall that $\y_m$ is the 
$N$--dimensional vector $\y_m = (y_{m,1}, \ldots, y_{m,N})^{T}$ which can be written as $\y_m = \R_m^{1/2} \z_m$ where $\R_m = \mathbb{E}(\y_m \y_m^*)$ and $\z_m$ is $\mathcal{N}_{\Cbb}(0,\I_N)$ distributed. It is clear that $\hat{s}_{m,L}(\nu)$ can be written as 
$$
\hat{s}_{m,L}(\nu) = \z_m^* \R_m^{1/2} \Omegabs(\nu) \R_m^{1/2} \z_m
$$
with 
$$
\Omegabs(\nu) = \frac{1}{N} \int \d_N(\mu) \d_N(\mu)^* w_L(\nu - \mu) d\mu
$$
while $\hat{s}'_{m,L}(\nu)$ is equal to 
$$
\hat{s}'_{m,L}(\nu)  =  \z_m^* \R_m^{1/2} \Omegabs'(\nu)  \R_m^{1/2} \z_m 
$$
with 
$$
            \Omegabs'(\nu) = \frac{-2 i \pi}{N} \int \d_N(\mu) \d_N(\mu)^* \left( \sum_{l=-L}^{L} l \, e^{-2 i \pi l (\nu -\mu)} \right) \, d\mu
$$
It is easy to check that $\|  \Omegabs'(\nu) \|_F = \Ocal(\frac{L^{3/2}}{N^{1/2}})$ 
and therefore that $ \| \R_m^{1/2} \Omegabs'(\nu)  \R_m^{1/2} \|_F = \Ocal(\frac{L^{3/2}}{N^{1/2}})$. The Hanson-Wright inequality leads immediately 
to $|\hat{s}'_{m,L}(\nu) - \mathbb{E}(\hat{s}'_{m,L}(\nu))| \prec \frac{L^{3/2}}{N^{1/2}}$. Moreover, it is easy to check that (\ref{eq:condition-R}) implies that 
$$
|\mathbb{E}(\hat{s}'_{m,L}(\nu)) - s_m'(\nu)| \leq \frac{C}{L^{\gamma_0-1}} 
$$
where $C$ is a nice constant. For $L=L(N)$ in such a way that 
$\Ocal(\frac{L^{3/2}}{N^{1/2}}) = \frac{1}{L^{\gamma_0-1}}$, i.e. 
$L = \Ocal(N^{\frac{1}{2\gamma_0+1}})$, we obtain that 
$$
|\hat{s}'_{m,L}(\nu) - s_m'(\nu)| \prec \frac{1}{N^{(\gamma_0 -1)/(2\gamma_0+1)}}.
$$
Moreover, a similar analysis leads to 
$$
|\hat{s}_{m,L}(\nu) - s_m(\nu)| \prec \frac{1}{N^{\gamma_0/(2\gamma_0+1)}}
$$
from which we deduce that the estimate $\hat{r}_N(\nu)$ defined by (\ref{eq:def-hatrNL}) satisfies 
(\ref{eq:domination-hatrN}). 
It is then easily checked that if $\gamma_0 \geq 3$, then (\ref{eq:domination-bis-hatrN}) holds, which implies 
that $|\hat{\psi}_N(f,\nu)| \prec u_N$ holds. 
\end{proof}
\section{Use of Lipschitz properties of the functions \texorpdfstring{$\nu \rightarrow \psi_N(f,\nu)$ and $\nu \rightarrow \hat{\psi}_N(f,\nu)$}{}}
\label{sec:lipschitz}
In this section, we establish Lipschitz properties of $\nu \rightarrow \psi_N(f,\nu)$ and $\nu \rightarrow \hat{\psi}_N(f,\nu)$, and deduce that the stochastic domination properties (\ref{eq:stochastic-domination-psi}) and (\ref{eq:domination-hatpsi}) continues to remain valid 
for $\sup_{\nu \in [0,1]} |\psi_N(f,\nu)|$ and $\sup_{\nu \in [0,1]} |\hat{\psi}_N(f,\nu)|$ where $\hat{\psi}_N(f,\nu)$ is defined by 
(\ref{eq:def-hatpsi}, \ref{eq:def-hatrNL}). 
\subsection{Lipschitz properties}
The goal of this paragraph is to prove the following Proposition. 
\begin{proposition}
\label{prop:psi-hatpsi-lipschitz}
Functions $\nu \rightarrow \psi_N(f,\nu)$ and $\nu \rightarrow \hat{\psi}_N(f,\nu)$ satisfy 
\begin{eqnarray}
\label{eq:lipschitz-constant-psi}
 \sup_{\delta \neq 0} \sup_{\nu\in[0, 1]} \frac{\|\psi_N(f,\nu)- \psi_N(f,\nu+\delta)\|}{|\delta|} &\prec MN^{3/2} \\
 \label{eq:lipschitz-constant-hatpsi}
 \sup_{\delta \neq 0} \sup_{\nu\in[0, 1]} \frac{\|\hat{\psi}_N(f,\nu)- \hat{\psi}_N(f,\nu+\delta)\|}{|\delta|} & \prec MN^{3/2}
\end{eqnarray}
\end{proposition}
In the following, we just establish (\ref{eq:lipschitz-constant-hatpsi}). For this, we evaluate separately 
the Lipschitz constants of $\nu \rightarrow \frac{1}{M} \Tr f(\hat{\C}(\nu))$ and of 
$\nu \rightarrow \hat{r}_N(\nu)$. 
\subsubsection{Lipschitz constant of \texorpdfstring{$\nu \rightarrow \frac{1}{M} \Tr f(\hat{\C}(\nu))$}{}}
To show that  $\nu \rightarrow \frac{1}{M} \Tr f(\hat{\C}(\nu))$ is $MN^{3/2}$-Lipschitz with overwhelming probability, we need to establish a number of intermediate properties. 
\begin{proposition}
\label{proposition:Shat_lipschitz}
It holds that 
\begin{equation}
    \label{equation:Shat_lipschitz}
    \sup_{\delta \neq 0} \sup_{\nu\in[0, 1]} \frac{\|\hat{\S}(\nu)-\hat{\S}(\nu+\delta)\|}{|\delta|}\prec MN^{3/2}.
\end{equation}
\end{proposition}

\begin{proof}
Let $\delta\in\mathbb{R}$ and $\nu\in[0, 1]$. As the random variables $(y_{m,n})_{m=1, \ldots, M, n=1, \ldots, N}$ are complex Gaussian and that $\sup_{m \geq 1} \mathbb{E}|y_{m,n}|^{2} < +\infty$, 
the family $(y_{m,n})_{m=1, \ldots, M, n=1, \ldots, N}$ satisfies $|y_{m,n}| \prec 1$. Therefore, it holds that 
\begin{equation}
 \label{eq:domination_stochastique_y_m_n_normalise}
    \frac{1}{\sqrt{N}}\sum_{n=1}^N|y_{m,n}| \prec \sqrt{N} .  
\end{equation}
For the same reasons, 
the family $\xi_{y_m}(\nu), m=1, \ldots, M, \, \nu \in [0,1]$ satisfies. 
\begin{equation}
\label{eq:domination_stochastique_simple_xi}
    |\xi_{y_m}(\nu)| \prec 1 .
\end{equation}
We also claim that 
\begin{equation}
    \label{eq:domination_stochastique_xi}
    \sup_{\nu \in [0,1]} |\xi_{y_m}(\nu)| \prec 1.
\end{equation}
In order to verify (\ref{eq:domination_stochastique_xi}), we first observe that 
for any $n\ge1$, we have the following control:
$$|e^{-2i\pi n\nu} - e^{-2i\pi n(\nu+\delta)} | \le 2|\sin\pi n\delta|\le  2\pi n|\delta|.$$
\eqref{eq:domination_stochastique_y_m_n_normalise} implies that 
\begin{align}
\notag
    &\sup_{\delta \neq 0}\sup_{\nu\in[0, 1]}\left|\frac{\xi_{y_m}(\nu)-\xi_{y_m}(\nu+\delta)}{\delta}\right| \\
\notag 
    &\hskip3cm= \sup_{\delta \neq 0}\sup_{\nu\in[0, 1]}\frac{1}{\sqrt{N}}\left|\sum_{n=1}^N y_{m,n}\frac{e^{-2i\pi n\nu} - e^{-2i\pi n(\nu+\delta)}}{\delta}\right| \\
\notag
    &\hskip3cm\le 2\pi N\frac{1}{\sqrt{N}}\sum_{n=1}^N|y_{m,n}| \\
\label{eq:domination_stochastique_derivee_xi}
    &\hskip3cm\prec N^{3/2}.
\end{align}
We consider a frequency $\nu_* \in [0,1]$ (depending on $m$) where $|\xi_{y_m}(\nu)|$ is maximum, and have thus to establish 
that for each $\epsilon > 0$, then there exists $\gamma > 0$ depending only on $\epsilon$ such that 
$$
\Pbb(|\xi_{y_m}(\nu_*)| > N^{\epsilon}) \leq \exp{-N^{\gamma}}
$$
for each $N$ larger than a certain integer $N_0(\epsilon)$. 
We introduce the discrete set 
\begin{equation}
    \label{eq:def-V_N_p}
    \Vcal_N^p=\left\{\frac{k}{N^p}: k\in\{0,\ldots,N^p-1\}\right\}
\end{equation}
whose cardinality is $|\Vcal_N^p|=N^p$. We notice that   (\ref{eq:domination_stochastique_simple_xi}) in conjunction with the union bound implies that 
$\sup_{\nu_p \in \Vcal_N^p} |\xi_{y_m}(\nu_p)| \prec 1$. We denote by $\nu_{*,p}$ the element of $\Vcal_N^p$ for which
$|\nu_* - \nu_p|$ is minimum, and notice that $|\nu_* - \nu_{*,p}| \leq \frac{1}{N^{p}}$. Then, we have the following inequality
\begin{align}
\nonumber
    &\Pbb(|\xi_{y_m}(\nu_*)| > N^{\epsilon})  \\
\nonumber
    &\hskip1cm\leq \Pbb\left(|\xi_{y_m}(\nu_*) - \xi_{y_m}(\nu_{*,p})| > \frac{N^{\epsilon}}{2}\right)+
\Pbb\left(|\xi_{y_m}(\nu_{*,p})| > \frac{N^{\epsilon}}{2}\right) \\
\label{eq:second-inequality}
   &\hskip1cm\leq \Pbb\left(|\xi_{y_m}(\nu_*) - \xi_{y_m}(\nu_{*,p})| > \frac{N^{\epsilon}}{2}\right)+\Pbb\left(\sup_{\nu_p \in \Vcal_N^p} |\xi_{y_m}(\nu_{p})| > \frac{N^{\epsilon}}{2}\right).
\end{align}
As $\sup_{\nu_p \in \Vcal_N^p} |\xi_{y_m}(\nu_p)| \prec 1$, the second term of the right hand side of (\ref{eq:second-inequality}) converges exponentially towards $0$. In order to evaluate the first term of the r.h.s. of (\ref{eq:second-inequality}), we use (\ref{eq:domination_stochastique_derivee_xi}), and obtain that 
\begin{align*}
\Pbb\left(|\xi_{y_m}(\nu_*) - \xi_{y_m}(\nu_{*,p})| > \frac{N^{\epsilon}}{2} \right) & \leq
\Pbb \left( N \, \frac{1}{\sqrt{N}}\sum_{n=1}^N|y_{m,n}| \geq \frac{\pi}{2|\nu_* - \nu_{*,p}|} N^{\epsilon} \right)\\ 
    & \leq \Pbb \left( \frac{1}{\sqrt{N}}\sum_{n=1}^N|y_{m,n}| \geq \frac{\pi}{2} N^{p+\epsilon - 1} \right).
\end{align*}
We choose $p$ so that $p-1 > 3/2$, and use  \eqref{eq:domination_stochastique_y_m_n_normalise}  to conclude that $\Pbb\left(|\xi_{y_m}(\nu_*) - \xi_{y_m}(\nu_{*,p})| > \frac{N^{\epsilon}}{2} \right)$ converges towards $0$ exponentially. This establishes 
(\ref{eq:domination_stochastique_xi}). \\

In order to complete the proof of Proposition \ref{proposition:Shat_lipschitz}, we consider an individual entry $\hat{s}_{ij}(\nu)$ of $\hat{\S}(\nu)$ for $i,j\le M$, and write
\begin{align*}
    &\left|\hat{s}_{ij}(\nu)-\hat{s}_{ij}(\nu+\delta)\right| \\
    &\hskip1cm= \frac{1}{B+1}\left|\sum_{b=-B/2}^{B/2}\xi_i\left(\nu+\frac{b}{N}\right)\xi_j\left(\nu+\frac{b}{N}\right)^*\right.\\
    &\hskip2cm-\left.\xi_i\left(\nu+\delta+\frac{b}{N}\right)\xi_j\left(\nu+\delta+\frac{b}{N}\right)^*\right| \\
    &\hskip1cm\le \frac{1}{B+1}\sum_{b=-B/2}^{B/2} \left|\xi_i\left(\nu+\frac{b}{N}\right)\left(\xi_j\left(\nu+\frac{b}{N}\right)^*-\xi_j\left(\nu+\delta+\frac{b}{N}\right)^*\right) \right| \\
    &\hskip2cm+\left|\left(\xi_i\left(\nu+\frac{b}{N}\right)-\xi_i\left(\nu+\delta+\frac{b}{N}\right)\right)\xi_j\left(\nu+\delta+\frac{b}{N}\right)^*\right|.
\end{align*}
Using the estimations \eqref{eq:domination_stochastique_xi} and \eqref{eq:domination_stochastique_derivee_xi}, we get:
\begin{equation}
    \label{equation:domination_sij}
    \sup_{i,j}\sup_{\delta \neq 0}\sup_{\nu\in[0, 1]}\left|\frac{\hat{s}_{ij}(\nu)-\hat{s}_{ij}(\nu+\delta)}{\delta}\right| \prec N^{3/2}
\end{equation}
and deduce (\ref{equation:Shat_lipschitz}) from the rough bound
\begin{align*}
    \sup_{\nu\in[0, 1]}\|\hat{\S}(\nu)-\hat{\S}(\nu+\delta)\| &\le \sup_{\nu\in[0, 1]}\sup_i\sum_j |\hat{s}_{ij}(\nu)-\hat{s}_{ij}(\nu+\delta)| \\
    &\le M\sup_{\nu\in[0, 1]}\sup_{i,j}|\hat{s}_{ij}(\nu)-\hat{s}_{ij}(\nu+\delta)|.
\end{align*} 
\end{proof}

Combining the eigenvalue localisation result from Corollary \ref{corollary:spectre-tildeC-spectre-hatS} and the Lipschitz behaviour of $\hat{\S}$ from Proposition \ref{proposition:Shat_lipschitz}, the following statement holds.
\begin{corollary}{($\nu$ uniform version of Corollary \ref{corollary:spectre-tildeC-spectre-hatS}.)}
\label{corollary:spectre_S_D_uniforme}
Denote for $\epsilon>0$:
\begin{align*}
    &\Lambda_\epsilon^{\hat{\S}}= \left\{\forall\nu\in[0, 1]:\sigma(\hat{\S}(\nu))\subset\Supp\mu_{MP}^{(c)} \times [\barbelow{s}, \bar{s}]+\epsilon\right\} \\
    &\Lambda_\epsilon^{\hat{\D}}=\left\{\forall\nu\in[0,1]:\sigma(\hat{\D}(\nu))\subset[\barbelow{s},\bar{s}]+\epsilon\right\}.
\end{align*}
 
Then, $\Lambda_\epsilon^{\hat{\S}}$ and  $\Lambda_\epsilon^{\hat{\D}}$ hold with exponentially high probability. 
\end{corollary}

\begin{proof}
As the proof for $\Lambda_\epsilon^{\hat{\D}}$ is strictly similar to the one of $\Lambda_\epsilon^{\hat{\S}}$, we will only write the arguments for $\Lambda_\epsilon^{\hat{\S}}$. For any fixed $\nu\in[0,1]$, Corollary \ref{corollary:spectre-tildeC-spectre-hatS} ensures that $\Lambda_\epsilon^{\hat{\S}}(\nu)$ holds with exponentially high probability. For $p \geq 1$, we still consider the set $\Vcal_N^p$ defined by (\ref{eq:def-V_N_p}) 
and denote by $ \Lambda_{\epsilon,p}^{\hat{\S}}$ the event defined by 
$$ \Lambda_{\epsilon,p}^{\hat{\S}}=\left\{\forall\nu_p\in\Vcal_N^p:\sigma(\hat{\S}(\nu_p))\subset\Supp\mu_{MP}^{(c)}\times[\barbelow{s},\bar{s}]+\epsilon\right\} $$
which is $\Lambda_\epsilon^{\hat{\S}}$ but where $\nu$ runs only on the finite grid $\Vcal_N^p$. It is immediate (by the union bound) that $\Lambda_{\epsilon,p}^{\hat{\S}}$ holds with exponentially high probability for any fixed $p\in\Nbb$. Moreover, it is clear from the definitions of $\Lambda_\epsilon^{\hat{\S}}$ and $\Lambda_{\epsilon,p}^{\hat{\S}}$ that $ \Lambda_\epsilon^{\hat{\S}} \subset \Lambda_{\epsilon,p}^{\hat{\S}}$. We now show the following inclusion:
\begin{align}
    \notag
    \left(\Lambda_\epsilon^{\hat{\S}}\right)^c &\subset \left(\Lambda_{\epsilon/2,p}^{\hat{\S}}\right)^c \\
    &\quad\cup \left\{\exists \nu\in[0,1]: \|\hat{\S}(\nu) - \hat{\S}(\nu_p^*)\|>\epsilon/2 \text{ where }\nu_p^*\in\argmin_{\nu_p\in\Vcal_N^p}|\nu-\nu_p|\right\} 
    \label{eq:decomposition_Lambda_S_epsilon}.
\end{align}

Suppose that $(\Lambda_\epsilon^{\hat{\S}})^c$ is realized, and denote by $\nu^*\in[0,1]$ 
a frequency such that $\sigma(\hat{\S})(\nu^*)\not\subset\Supp_{MP}^{(c)}\times[\barbelow{s},\bar{s}]+\epsilon$. Denote also $\nu_p^*\in\argmin_{\nu_p\in\Vcal_N^p}|\nu_p-\nu^*|$. We just consider the case where $\lambda_1(\hat{\S}(\nu^*))>\bar{s}(1+\sqrt{c})^2+\epsilon$, since in the case where 
$\lambda_M(\hat{\S}(\nu^*))<\barbelow{s} (1-\sqrt{c})^2-\epsilon$, the proof is similar. 
Then, either:
\begin{enumerate}
    \item $\|\hat{\S}(\nu_p^*)-\hat{\S}(\nu^*)\|\leq\epsilon/2$, which implies the following estimation for the location of $\lambda_1(\hat{\S}(\nu_p^*))$: 
    $$\lambda_1(\hat{\S}(\nu^*)) - \frac{\epsilon}{2} \le \lambda_1(\hat{\S}(\nu_p^*)) \le  \lambda_1(\hat{\S}(\nu^*)) + \frac{\epsilon}{2}  $$
    and in particular, $\lambda_1(\hat{\S}(\nu_p^*)) \geq \bar{s}(1+\sqrt{c})^2+\epsilon/2$. This means that $\left(\Lambda_{\epsilon/2,p}^{\hat{\S}}\right)^c$ holds.
    \item $\|\hat{\S}(\nu_p^*)-\hat{\S}(\nu^*)\|>\epsilon/2$, which exactly means that $\left\{\exists \nu\in[0,1]: \|\hat{\S}(\nu) - \hat{\S}(\nu_p^*)\|>\epsilon/2 \text{ where }\nu_p^*\in\argmin_{\nu_p\in\Vcal_N^p}|\nu-\nu_p|\right\}$ is realized
\end{enumerate}
\eqref{eq:decomposition_Lambda_S_epsilon} is now proved. \\

We already showed that $\left(\Lambda_{\epsilon/2,p}^{\hat{\S}}\right)^c$ holds with exponentially small probability, and establish now that the set 
$$
\left\{\exists \nu\in[0,1]: \|\hat{\S}(\nu_p^*) - \hat{\S}(\nu)\|>\epsilon/2 \text{ where }\nu_p^*\in\argmin_{\nu_p\in\Vcal_N^p}|\nu-\nu_p|\right\}
$$
has the same property. To justify this claim, we note that Proposition \ref{proposition:Shat_lipschitz} implies that for each $\kappa > 0$, the probability
$$  \Prob\left[\left\{\exists \nu, \nu' \in[0,1], \;  \frac{\|\hat{\S}(\nu) - \hat{\S}(\nu')\|}{|\nu - \nu'|} > N^{\kappa} M N^{3/2} \right\} \right] $$
converges to $0$ exponentially fast. As the following inclusion
\begin{align*}
    &\left\{\exists \nu\in[0,1], \; \frac{\|\hat{\S}(\nu) - \hat{\S}(\nu_{p}^{*})\|}{|\nu - \nu_p^{*}|}, \, 
 > N^{\kappa} M N^{3/2}, \, \mbox{where} \, \nu_p^*\in\argmin_{\nu_p\in\Vcal_N^p}|\nu-\nu_p| \right \} \\
 &\hskip2cm\subset 
\left\{\exists \nu, \nu' \in[0,1], \; \frac{\|\hat{\S}(\nu) - \hat{\S}(\nu')\|}{|\nu - \nu'|} > N^{\kappa} M N^{3/2} \right \}
\end{align*}
holds, we get that 
$$
\Prob\left[\left\{\exists \nu\in[0,1], \; \|\hat{\S}(\nu) - \hat{\S}(\nu_{p}^{*})\|  > |\nu - \nu_p^*| N^{\kappa} M N^{3/2} \right\}  \right] \rightarrow 0
$$
exponentially fast. Moreover, as for each $\nu$, $|\nu - \nu_p^*| \leq \frac{1}{N^p}$, we obtain that 
$$
\Prob\left[\left\{\exists \nu\in[0,1], \; \|\hat{\S}(\nu) - \hat{\S}(\nu_p^*)\|  > \frac{1}{N^{p}} N^{\kappa} M N^{3/2} \right\} \right] \rightarrow 0
$$
exponentially fast as well. For $p$ large enough, $N^{\kappa} \frac{1}{N^{p}} M N^{3/2}$ will finally become smaller than $\epsilon/2$. This proves that 
$$ \left\{\exists \nu\in[0,1], \; \|\hat{\S}(\nu_p^*) - \hat{\S}(\nu)\|>\epsilon/2 \text{ where }\nu_p^*\in\argmin_{\nu_p\in\Vcal_N^p}|\nu-\nu_p|\right\} $$
holds with exponentially small probability. \\

The same argument can be used to control $\Lambda_\epsilon^{\hat{\D}}$. This completes the proof
of Corollary \ref{corollary:spectre_S_D_uniforme}.
\end{proof}

We deduce immediately from Corollary \ref{corollary:spectre_S_D_uniforme} the following result that can be seen 
as a refinement of \eqref{eq:domination-norm-hatS} and of Lemma \ref{lemma:localization_s_m}.
\begin{corollary}
\label{coro:uniform-spectrum}
It holds that
$$ \sup_{\nu\in[0, 1]}\|\hat{\D}(\nu)^{-1/2}\|\prec1, \quad \sup_{\nu\in[0, 1]}\|\hat{\S}(\nu)\|\prec1. $$
\end{corollary}
 A useful consequence of this is the following Corollary, which states that the Lipschitz result holds for $\hat{\C}(\nu)$.
\begin{corollary}
\label{corollary:Chat_lipschitz}
It holds that 
\begin{equation}
    \label{equation:Chat_lipschitz}
    \sup_{\delta \neq 0}\sup_{\nu\in[0, 1]}\left\|\frac{\hat{\C}(\nu)-\hat{\C}(\nu+\delta)}{\delta}\right\| \prec MN^{3/2}
\end{equation}
\end{corollary}

\begin{proof}
For more clarity in the following argument, denote $\nu_1=\nu$ and $\nu_2=\nu+\delta$. Recall that $\hat{\D}=\diag\hat{\S}$. Using the definition of $\hat{\C}$ from equation \eqref{equation:definition_coherency}, we write:
\begin{align*}
    \hat{\C}(\nu_2)-\hat{\C}(\nu_1) &= \hat{\D}^{-1/2}(\nu_2)\hat{\S}(\nu_2)\hat{\D}^{-1/2}(\nu_2)-\hat{\D}^{-1/2}(\nu_1)\hat{\S}(\nu_1)\hat{\D}^{-1/2}(\nu_1) \\
    &= (\hat{\D}^{-1/2}(\nu_2)-\hat{\D}^{-1/2}(\nu_1))\hat{\S}(\nu_2)\hat{\D}^{-1/2}(\nu_2) \\
    &\hskip2cm+ \hat{\D}^{-1/2}(\nu_1)(\hat{\S}(\nu_2)\hat{\D}^{-1/2}(\nu_2)-\hat{\S}(\nu_1)\hat{\D}^{-1/2}(\nu_1)) .
\end{align*}

Moreover, we write that 
\begin{align*}
    &\hat{\S}(\nu_2)\hat{\D}^{-1/2}(\nu_2)-\hat{\S}(\nu_1)\hat{\D}^{-1/2}(\nu_1) \\
    &\hskip2cm= (\hat{\S}(\nu_2)-\hat{\S}(\nu_1))\hat{\D}^{-1/2}(\nu_2) + \hat{\S}(\nu_1)(\hat{\D}^{-1/2}(\nu_2)-\hat{\D}^{-1/2}(\nu_1)).
\end{align*}

Therefore, applying the operator norm, we get by the triangle inequality:
\begin{align*}
    \|\hat{\C}(\nu_2)-\hat{\C}(\nu_1)\| &\le \|\hat{\D}^{-1/2}(\nu_2)-\hat{\D}^{-1/2}(\nu_1)\|\|\hat{\S}(\nu_2)\|\|\hat{\D}^{-1/2}(\nu_2)\|\\ &+\|\hat{\D}^{-1/2}(\nu_1)\|\|\hat{\S}(\nu_2)-\hat{\S}(\nu_1)\|\|\hat{\D}^{-1/2}(\nu_2)\| \\
    &+\|\hat{\D}^{-1/2}(\nu_1)\|\|\hat{\S}(\nu_1)\|\|\hat{\D}^{-1/2}(\nu_2)-\hat{\D}^{-1/2}(\nu_1)\|
\end{align*}

It is easy to check that 
$$ \sup_{\delta \neq 0}\sup_{|\nu_2-\nu_1|=\delta}\left\|\frac{\hat{\D}^{-1/2}(\nu_2)-\hat{\D}^{-1/2}(\nu_1)}{\delta}\right\| \prec N^{3/2}
$$
holds. Therefore, Proposition \ref{proposition:Shat_lipschitz} and Corollary \ref{coro:uniform-spectrum} immediately imply (\ref{equation:Chat_lipschitz}).
\end{proof}

Finally, we can write for the spectrum of $\hat{\C}$ the same kind of result as in Corollary \ref{corollary:spectre_S_D_uniforme}.
\begin{corollary}
\label{corollary:spectre_C_uniforme}
For each $\epsilon > 0$, we define $\Lambda^{\hat{\C}}_\epsilon$ as the event
\begin{equation*}
    \Lambda^{\hat{\C}}_\epsilon=\left\{\forall \nu\in[0,1]:\sigma(\hat{\C}(\nu))\subset\Supp\mu_{MP}^{(c)}+\epsilon\right\}.
\end{equation*}
Then, $\Lambda^{\hat{\C}}_\epsilon$ holds with exponentially 
high probability. 
\end{corollary}

\begin{proof}
The proof is similar to the proof of Corollary \ref{corollary:spectre_S_D_uniforme}
and is thus omitted.
\end{proof}

We finally use the above results to prove that $\nu \rightarrow \frac{1}{M} \Tr f(\hat{\C}(\nu)) - \int f \, d\mu_{MP}^{(c_N)}$ is $MN^{3/2}$-Lipschitz with overwhelming probability. For this, we establish the 
following Proposition. 
\begin{proposition}
\label{prop:lipschitz-constant-Trf(hatC)}
It holds that
\begin{equation}
 \label{eqlipschitz-constant-Trf(hatC)}
  \sup_{\delta \neq 0}\sup_{\nu\in[0, 1]}\frac{1}{|\delta|} \, \left| \frac{1}{M} \Tr f(\hat{\C}(\nu)) - \frac{1}{M} \Tr f(\hat{\C}(\nu+\delta)) \right| \prec MN^{3/2}.
\end{equation}
\end{proposition}
\begin{proof}
By Corollary \ref{corollary:spectre_C_uniforme}, the event $\Lambda^{\hat{\C}}_\epsilon$ holds with exponentially high probability. Therefore, it is sufficient to establish that 
$$
\mathbf{1}_{\Lambda^{\hat{\C}}_\epsilon} \; \sup_{\delta \neq 0}\sup_{\nu\in[0, 1]}\frac{1}{|\delta|} \, \left| \frac{1}{M} \Tr f(\hat{\C}(\nu+\delta)) - \frac{1}{M} \Tr f(\hat{\C}(\nu)) \right| \prec MN^{3/2}
$$
We express $ \frac{1}{M} \Tr f(\hat{\C}(\nu+\delta)) - \frac{1}{M} \Tr f(\hat{\C}(\nu))$ as 
$$
\frac{1}{M} \Tr f(\hat{\C}(\nu+\delta)) - \frac{1}{M} \Tr f(\hat{\C}(\nu)) = 
\frac{1}{M}\sum_{m=1}^M f(\lambda_m(\hat{\C}(\nu+\delta)))-f(\lambda_m(\hat{\C}(\nu)).
$$
As $f$ is $\mathcal{C}^{\infty}$ on a neighborhood of $\Supp_{MP}^{(c)}$, on the set $\Lambda^{\hat{\C}}_\epsilon$, there exists some random quantities $(\tilde{\lambda}_m)_{1\le m\le M}$ between $\lambda_m(\hat{\C}(\nu))$ and $\lambda_m(\hat{\C}(\nu+\delta))$ such that
\begin{multline*}
    \frac{1}{M}\sum_{m=1}^M f(\lambda_m(\hat{\C}(\nu+\delta)))-f(\lambda_m(\hat{\C}(\nu)) \\= \frac{1}{M}\sum_{m=1}^M \left(\lambda_m(\hat{\C}(\nu+\delta))- \lambda_m(\hat{\C}(\nu)) \right) \, f'(\tilde{\lambda}_m).
\end{multline*}
Using the following eigenvalue inequality for Hermitian matrices:
$$ \left|\lambda_m(\hat{\C}(\nu+\delta))-\lambda_m(\hat{\C}(\nu))\right| \le \|\hat{\C}(\nu+\delta)-\hat{\C}(\nu)\|$$
in conjunction with the fact that $\sup_{1\le m\le M}|f'(\tilde{\lambda}_m|$ is bounded by some nice constant $C$ on the event $\Lambda^{\hat{\C}}_\epsilon$, we obtain that 
\begin{multline*}
    \Prob\left[ \sup_{\delta \neq 0} \sup_{\nu \in [0,1]} \left| \frac{1}{M}\sum_{m=1}^M f'(\tilde{\lambda}_m) (\lambda_m(\hat{\C}(\nu+\delta))-\lambda_m(\hat{\C}(\nu))\right| \right. \\ \left. > |\delta| N^\kappa MN^{3/2},\ \Lambda^{\hat{\C}}_\epsilon\right] \\
    \le \Prob\left[ \sup_{\delta \neq 0} \sup_{\nu \in [0,1]}  C\|\hat{\C}(\nu+\delta)-\hat{\C}(\nu)\|>|\delta|N^\kappa MN^{3/2},\ \Lambda^{\hat{\C}}_\epsilon\right]
\end{multline*} 
(\ref{equation:Chat_lipschitz}) finally leads to (\ref{eqlipschitz-constant-Trf(hatC)}). 

\end{proof}

\subsubsection{Lipschitz constant of \texorpdfstring{$\nu \rightarrow \hat{r}_N(\nu)$}{}.}
The function $\nu \rightarrow \hat{r}_N(\nu)$ satisfies the following property:
\begin{proposition}
\label{prop:lipschitz-constant-hatrN}
\begin{equation}
\label{eq:lipschitz-constant-hatrN}
 \sup_{\delta \neq 0}\sup_{\nu\in[0, 1]}\frac{1}{|\delta|} \, \left| \hat{r}_N(\nu+\delta) - \hat{r}_N(\nu) \right| \prec N^{3/(2\gamma_0+1)}.
 \end{equation}
\end{proposition}
We just provide the main steps the proof, and leave the details to the reader. We first prove that 
$\sup_{\nu \in [0,1]} \sum_{m=1}^{M} \frac{1}{\hat{s}_{m,L}(\nu)} \prec 1$ by verifying that 
the event $\{ \forall \nu \in [0,1], \forall m=1, \ldots, M, \hat{s}_{m,L}(\nu) \in [\underbar{s}, \bar{s}] + \epsilon \} $ holds with exponentially high probability. Then, we establish that $\nu \rightarrow \hat{s}_{m,L}(\nu)$ and 
$\nu \rightarrow \hat{s}'_{m,L}(\nu)$ are $N^{2/(2\gamma_0+1)}$ Lipschitz and $N^{3/(2\gamma_0+1)}$ Lipschitz with overwhelming probability. This leads immediately to (\ref{eq:lipschitz-constant-hatrN}). \\

As $v_N \, N^{3/(2\gamma_0+1)} \ll MN^{3/2}$, Propositions \ref{prop:lipschitz-constant-Trf(hatC)}  and \ref{prop:lipschitz-constant-hatrN} lead to (\ref{eq:lipschitz-constant-hatpsi}). This completes 
the proof of Proposition \ref{prop:psi-hatpsi-lipschitz}.
\subsection{Stochastic domination of \texorpdfstring{$\sup_{\nu \in [0,1]} |\psi_N(f,\nu)|$}{} and 
\texorpdfstring{$\sup_{\nu \in [0,1]} |\hat{\psi}_N(f,\nu)|$}{}}
We are now in a position to establish the main result of this paper. 
\begin{theorem}
\label{th:main-result}
$\sup_{\nu \in [0,1]} |\psi_N(f,\nu)|$ and 
$\sup_{\nu \in [0,1]} |\hat{\psi}_N(f,\nu)|$ satisfy the following stochastic domination 
property:
\begin{eqnarray}
    \label{eq:domination-sup-psi}
 \sup_{\nu \in [0,1]} |\psi_N(f,\nu)| & \prec & u_N  \\
  \label{eq:domination-sup-hatpsi}
  \sup_{\nu \in [0,1]} |\hat{\psi}_N(f,\nu)| & \prec & u_N .
\end{eqnarray}
\end{theorem}
\begin{proof} 
We just establish (\ref{eq:domination-sup-hatpsi}). We consider $\epsilon>0$
and evaluate
$$ \Prob\left[\sup_{\nu\in[0, 1]}\left| \hat{\psi}_N(f,\nu)\right|>N^\epsilon u_N\right].$$
We denote by $\nu^*\in[0, 1]$ an element where the supremum is achieved, and consider $\nu_p^*$ the closest element of $\Vcal_N^p$ to $\nu^*$, where we recall that $\Vcal_N^p$ is defined by (\ref{eq:def-V_N_p}). Therefore, one can write:
\begin{align*}
\Prob\left[\sup_{\nu\in[0, 1]}\left|\hat{\psi}_N(f,\nu)\right|>N^\epsilon u_N \right] \leq &
\Prob\left[\left|\hat{\psi}_N(f,\nu^*)- \hat{\psi}_N(f,\nu_p^*))\right| > \frac{1}{2} N^\epsilon u_N \right] + \\
 &  \Prob\left[ \left|\hat{\psi}_N(f,\nu_p^*)\right|>\frac{1}{2} N^\epsilon u_N\right] .
\end{align*}
(\ref{eq:domination-hatpsi}) implies that $\Prob\left[ \left|\hat{\psi}_N(f,\nu_p^*)\right|>\frac{1}{2} N^\epsilon u_N\right] $ converges exponentially towards 0. It thus remains to study 
$\Prob\left[\left|\hat{\psi}_N(f,\nu^*)- \hat{\psi}_N(f,\nu_p^*))\right| > \frac{1}{2} N^\epsilon u_N \right]$. For this, we of course use (\ref{eq:lipschitz-constant-hatpsi}), Corollary \ref{corollary:Chat_lipschitz}, and write 
\begin{align*}
    &\Prob\left[\left| \hat{\psi}_N(f,\nu^*)- \hat{\psi}_N(f,\nu_p^*)) \right|>\frac{1}{2} N^\epsilon u_N \right] \\
    &\hskip3cm= \Prob\left[\left| \frac{ \hat{\psi}_N(f,\nu^*)- \hat{\psi}_N(f,\nu_p^*))}{\nu^* - \nu_p^*} \right|>\frac{1}{2|\nu^* - \nu_p^*|}N^\epsilon u_N \right] \\
     &\hskip3cm \leq \Prob\left[\left| \frac{ \hat{\psi}_N(f,\nu^*)- \hat{\psi}_N(f,\nu_p^*))}{\nu^* - \nu_p^*} \right|> 
     \frac{1}{2} N^{p} N^{\epsilon} u_N \right] .
\end{align*}
If we choose $p$ large enough, $MN^{3/2}$ satisfies $MN^{3/2} \ll N^{p} u_N $, and 
$\Prob\left[\left| \frac{ \hat{\psi}_N(f,\nu^*)- \hat{\psi}_N(f,\nu_p^*))}{\nu^* - \nu_p^*} \right|> 
     \frac{1}{2} N^{p} N^{\epsilon} u_N \right] $ 
converges towards $0$ exponentially as expected. 
This completes the proof of (\ref{eq:domination-sup-hatpsi}). 
\end{proof}

\section{Numerical simulations}
\label{sec:simulations}
In this section we examine the impact of the correction quantity $r_N(\nu)\phi_N(f)v_N$ when $\alpha > \frac{2}{3}$ and see how it improves the estimation of the LSS $\frac{1}{M}\Tr f(\hat{\C}(\nu))$. More precisely, we start by examining the behaviour of the LSS
\[
    \left|\frac{1}{M}\Tr f(\hat{\C}(\nu))-\int_\Rbb f\diff\mu_{MP}^{(c_N)}\right|
\]
and the impact of the correction term 
\begin{align*}
    &\left(\frac{1}{M}\sum_{m=1}^M\frac{s_m'(\nu)}{s_m(\nu)}\right)^2\phi_N(f)\left(\frac{1}{B+1}\sum_{b=-B/2}^{B/2}\left(\frac{b}{N}\right)\right)^2 = r_N(\nu)\phi_N(f)v_N \\
    &\left(\frac{1}{M}\sum_{m=1}^M\frac{\hat{s}_m'(\nu)}{\hat{s}_m(\nu)}\right)^2\phi_N(f)\left(\frac{1}{B+1}\sum_{b=-B/2}^{B/2}\left(\frac{b}{N}\right)\right)^2 = \hat{r}_N(\nu)\phi_N(f)v_N.
\end{align*}
under $\Hcal_0$. We recall that $\phi_N(f)$ is the deterministic term defined as the action of $f$ on the compactly supported distribution $D_N$, whose Stieltjes transform is:
\[
    p_N(z) = -\frac{c_N(zt_N(z)\tilde{t}_N(z))^3}{1-c_N(zt_N(z)\tilde{t}_N(z))^2}.
\]
Motivated by \cite{mestre2017correlation}, we consider $f(\lambda) = (\lambda-1)^{2}$ where it can be satisfied with a bit of algebra and residue calculus that
\begin{align*}
    \int_\Rbb f(\lambda)\diff\mu_{MP}^{(c_N)}(\lambda) = c_N.
\end{align*} and $\phi_N(f)=c_N$. Take $\y_n$ generated by the following simple model:
\begin{equation}
\label{equation:state_space_model_H0}
    \begin{array}{ll}
        \y_{n+1} = \A\y_n + \epsilonbs_n \\
    \end{array}
\end{equation}
where $(\epsilonbs_n)_{n \in \mathbb{Z}}$ is an independent sequence of $\Ncal_\Cbb(0,\I_M)$ distributed random vectors, and where $\A$ is the diagonal matrix defined by 
$\A = \theta \, \I_M$ for $\theta\in\Cbb$ such that $|\theta|<1$. Under \eqref{equation:state_space_model_H0}, each time series is independent AR(1) processes. In Figure \ref{figure:correction_all_nu} is represented on the left the values of the LSS associated to $f(\lambda)=(\lambda-1)^2$ for each $\nu\in(0,1)$ when $(N,B,M,L)=(10119,1600,800,21)$ (so $\alpha=0.8$ and $c=1/2$) and $\theta=0.4$, where we recall that $L$ represent the lag window size in the estimation of $\hat{r}_N(\nu)$.. We see that the correction term captures the majority of the deviation of the LSS from zero. Moreover, the correction where the spectral densities $s_m$ and $s_m'$ are estimated still provide a good approximation of the $\Ocal(\frac{B}{N})^2$ term. On the right side is represented the LSS against $\psi_N(f,\nu)$ and $\hat{\psi}_N(f,\nu)$. We again observe that the majority of the deviation from zero of the LSS is corrected by the $\Ocal(\frac{B}{N})^2$ terms. Around $\nu=\pm0.1$, the corrections precision seems to have degraded. This can be understood since $\nu=\pm0.1$ corresponds to peaks in $s_m'$, which leads to greater estimation errors for $\hat{s}_m'$ at this frequency than for the other ones.

\begin{figure}[!ht]
\centering
     \includegraphics[width=\textwidth]{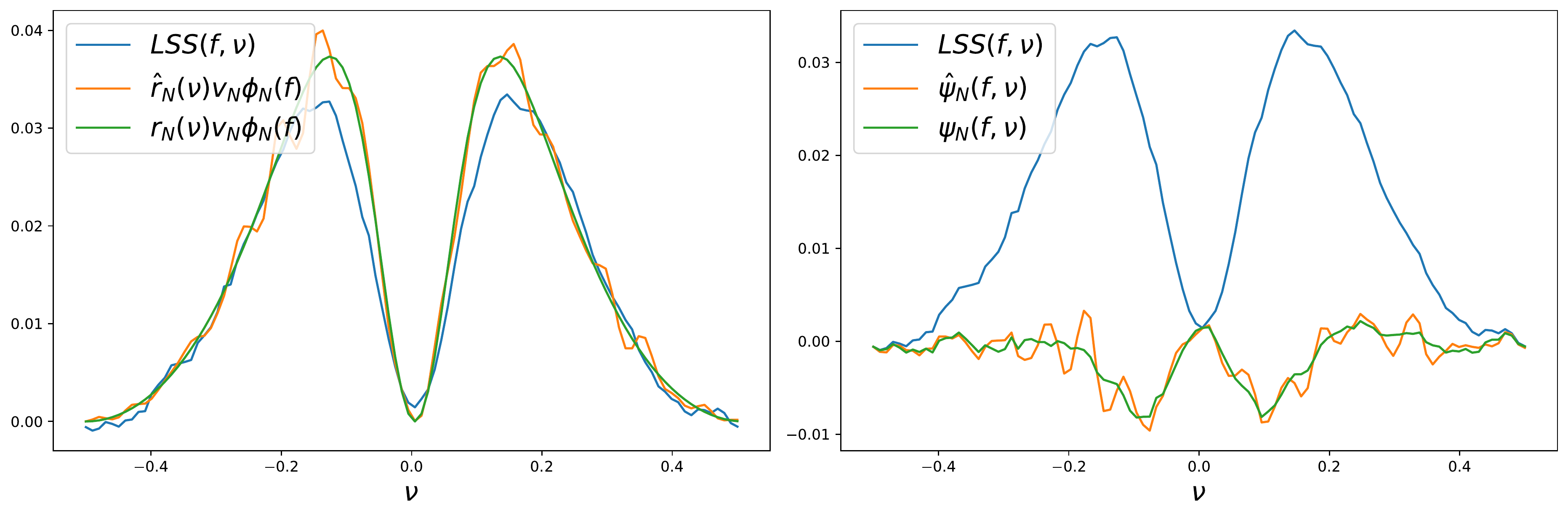}
      \caption{Linear Spectral Statistics vs the correction term. $f(\lambda)=(\lambda-1)^2$, $(N,B,M,L)=(10119,1600,800,21)$, and $\theta=0.4$.}
      \label{figure:correction_all_nu}
\end{figure}

We now check the derived speed of convergence towards zero in Theorem \ref{th:main-result}, and more precisely that the following estimations hold true:
\begin{align*}
    &\sup_{\nu\in[0,1]}\left|\frac{1}{M}\Tr f(\hat{\C}(\nu))-\int_\Rbb f\diff\mu_{MP}^{(c_N)}\right| = \Ocal\left(\frac{1}{B}\right)\mathds{1}_{1/2\le\alpha\le2/3}+\Ocal\left(\frac{B}{N}\right)^2\mathds{1}_{\alpha\ge2/3} \\
    &\sup_{\nu\in[0,1]}\left|\psi(f,\nu)\right| = \Ocal(u_N)
\end{align*}
and we abbreviate $\frac{1}{M}\Tr f(\hat{\C}(\nu))-\int_\Rbb f\diff\mu_{MP}^{(c_N)}$ by $LSS(f,\nu)$. In the following we take $c=\frac{M}{B+1}=\frac{1}{2}$ and $\alpha=4/5$. In this case we recall that $u_N=\Ocal(\frac{B}{N})^3$. On the left of Figure \ref{figure:correction_nu_max} is represented for $M\in\{20,30,\ldots,1500\}$ the value of $\sup_{\nu\in[0,1]} |LSS(f,\nu)|$ against $\sup_{\nu\in[0,1]}|\psi(f,\nu)|$ and $\sup_{\nu\in[0,1]}|\hat{\psi}(f,\nu)|$. On the right of Figure \ref{figure:correction_nu_max} we rescale all quantities by $(\frac{N}{B})^2$ and observe, in accordance with Theorem \ref{theo:domination-psi} that $LSS(f,\nu)$ remains $\Ocal(1)$ while the corrected quantities are $o(1)$. Finally, in Figure \ref{figure:correction_nu_max_cube} are represented $\sup_{\nu\in[0,1]}|\psi(f,\nu)|$ and $\sup_{\nu\in[0,1]}|\hat{\psi}(f,\nu)|$ rescaled by $(\frac{N}{B})^3$, and observe that these quantities are now $\Ocal(1)$, again in accordance with Theorem \ref{theo:domination-psi}.
\begin{figure}[!ht]
\centering
     \includegraphics[width=\textwidth]{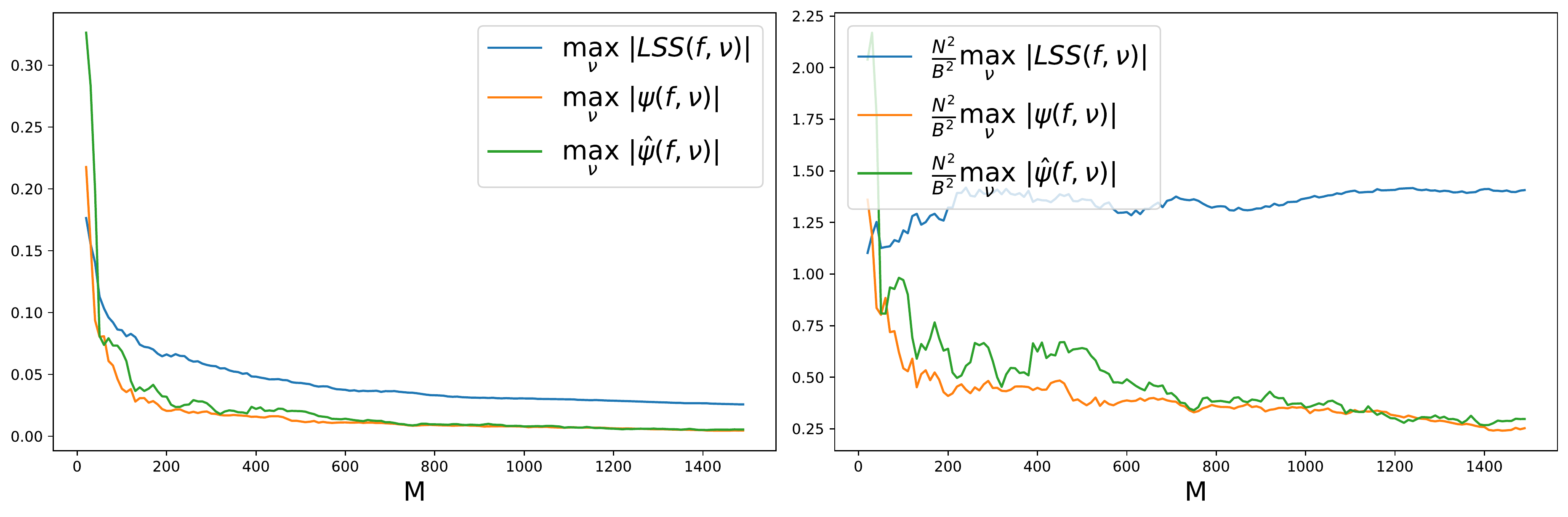}
      \caption{$\sup_{\nu\in\Fcal_N}\left|\frac{1}{M}\Tr f(\hat{\C}(\nu))-\int f\diff\mu_{MP}^{(c_N)}\right|$ against $\sup_{\nu\in\Fcal_N}\psi_N(f,\nu)$ and $\sup_{\nu\in\Fcal_N}\hat{\psi}_N(f,\nu)$ as functions of $M$. On the right the quantities are rescaled by $(\frac{N}{B})^2$. $\alpha=0.8$, $c=1/2$, $\theta=0.4$}
      \label{figure:correction_nu_max}
\end{figure}

\begin{figure}[!ht]
\centering
     \includegraphics[width=\textwidth]{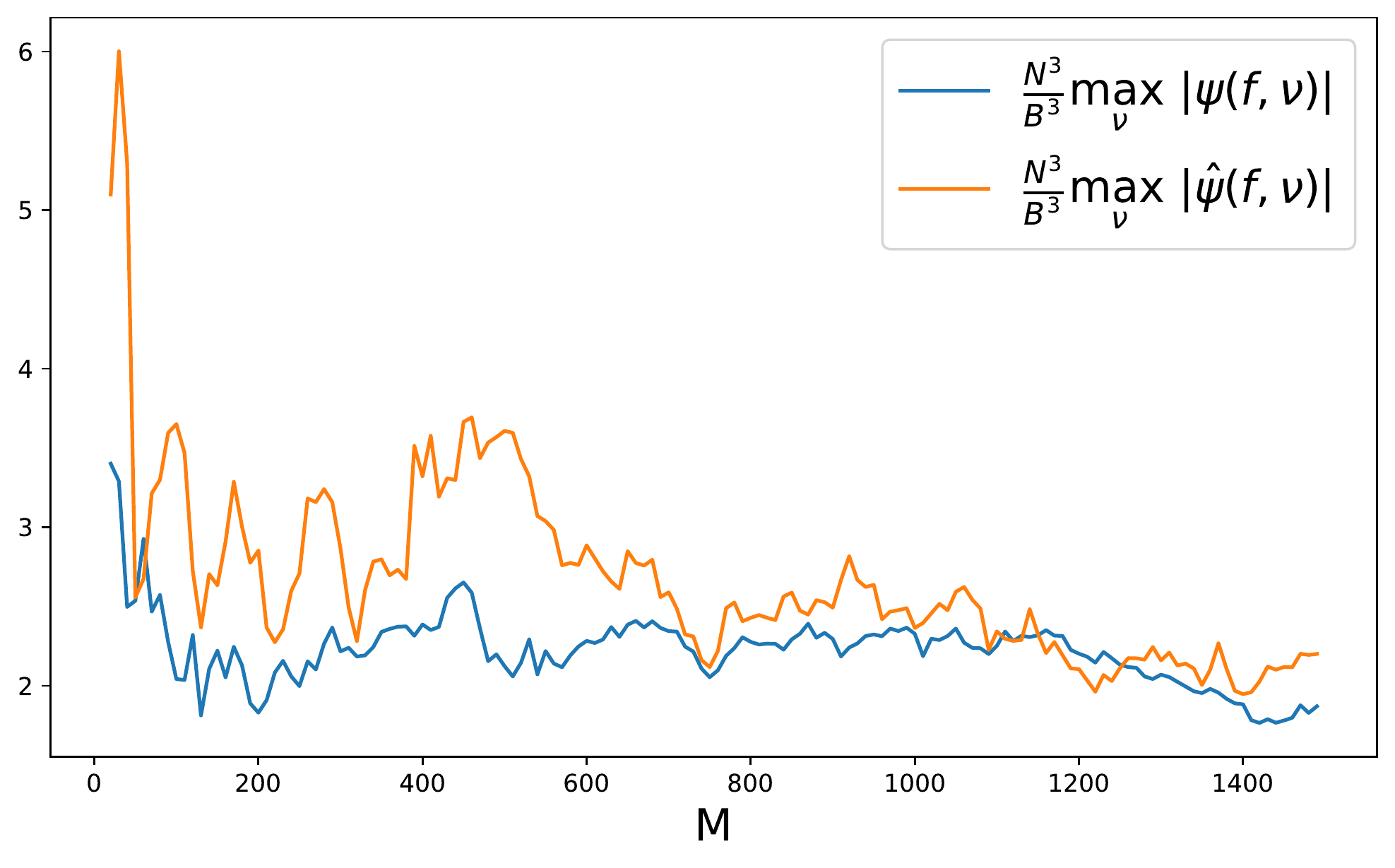}
      \caption{$\sup_{\nu\in\Fcal_N}\psi_N(f,\nu)$ and $\sup_{\nu\in\Fcal_N}\hat{\psi}_N(f,\nu)$ rescaled by $(\frac{N}{B})^3$ as functions of $M$. $\alpha=0.8$, $c=1/2$, $\theta=0.4$.}
      \label{figure:correction_nu_max_cube}
\end{figure}

Finally in Figure \ref{figure:correction_nu_max_hist} is represented 20000 realisations of the LSS $\sup_{\nu\in\Fcal_N}\left|\frac{1}{M}\Tr f(\hat{\C}(\nu))-\int f\diff\mu_{MP}^{(c_N)}\right|$ against its improved estimations $\sup_{\nu\in\Fcal_N}|\psi_N(f,\nu)|$ and $\sup_{\nu\in\Fcal_N}|\hat{\psi}_N(f,\nu)|$. We see that the oracle corrected statistics $\psi(f,\nu)$ is more concentrated around $0$, and that its estimated counterpart $\hat{\psi}(f,\nu)$ is close to $\psi(f,\nu)$ but exhibits more spread due to the additional estimation step of $\hat{s}_m(\nu)$.  

\begin{figure}[!ht]
\centering
     \includegraphics[width=0.8\textwidth]{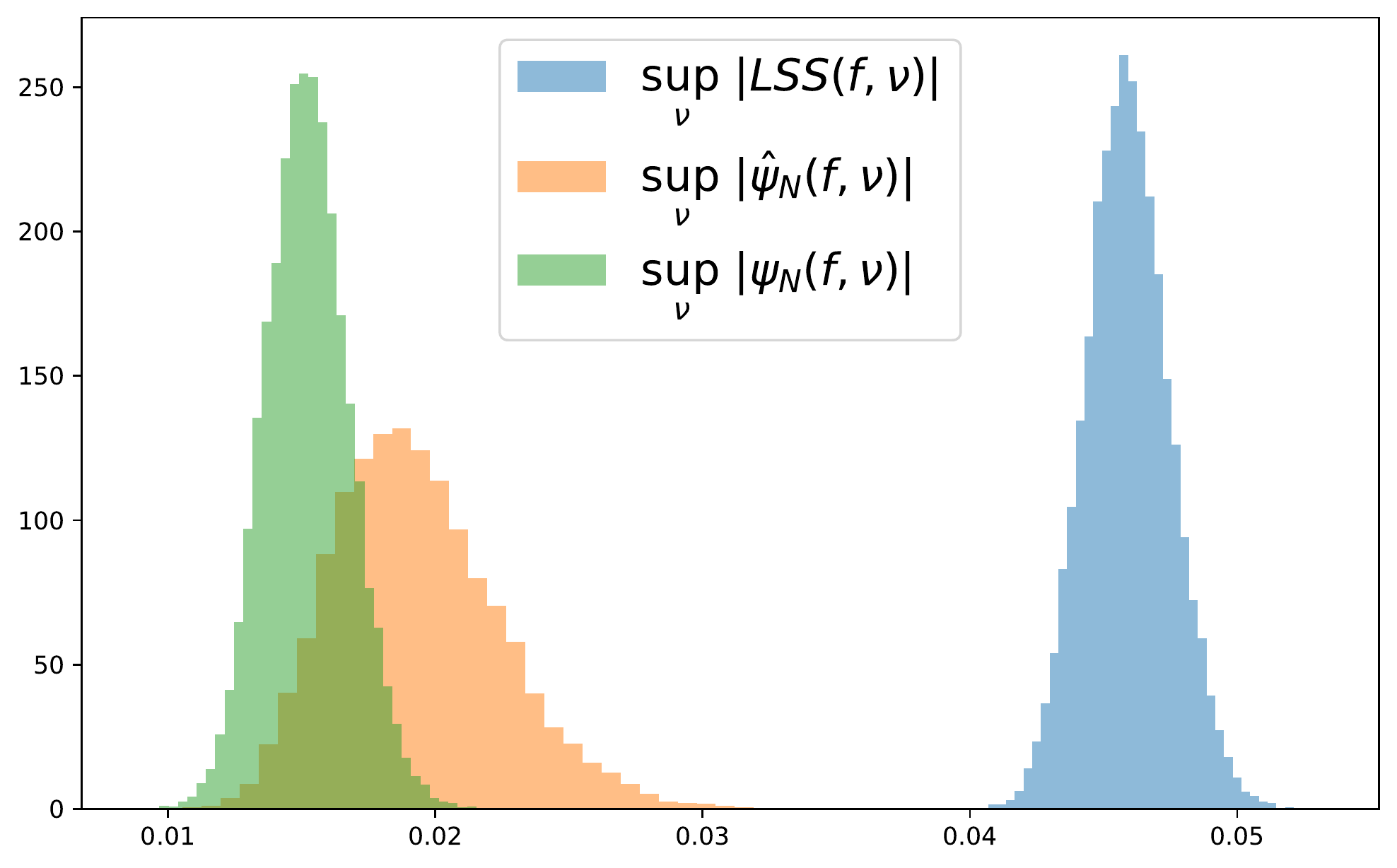}
      \caption{$\sup_{\nu\in\Fcal_N}|\frac{1}{M}\Tr f(\hat{\C}(\nu))-\int f\diff\mu_{MP}^{(c_N)}|$ against $\sup_{\nu\in\Fcal_N}|\psi_N(f,\nu)|$ and $\sup_{\nu\in\Fcal_N}|\hat{\psi}_N(f,\nu)|$. $(N,B,M,L)=(4254,800,400,16)$, $\theta=0.4$.}
      \label{figure:correction_nu_max_hist}
\end{figure}

\appendix

\section{Appendix}

\subsection{Proof of Lemma \ref{lemma:brillinger_uniformity_covariance}}
Lemma \ref{lemma:brillinger_uniformity_covariance} is a slight variation of Theorem 4.3.2 \cite{brillinger1981time}.

\begin{lemma}
\label{lemma:brillinger_uniformity_covariance}
For any $\nu_1$ and $\nu_2$ in $[0,1]$, such that there exists $k\in \{0, 1, \ldots, N-1 \}$ satisfying $\nu_2-\nu_1=k/N$, the following bound holds:
\begin{equation}
\label{eq:uniformite-brillinger-covariance}
    \sup_{m\ge1}\left|\Ebb\left[\xi_{y_m}(\nu_1)\xi_{y_m}(\nu_2)^*\right]-s_m(\nu_1)\delta_{\nu_1=\nu_2}\right|=\Ocal\left(\frac{1}{N}\right).
\end{equation}
\end{lemma}

\begin{proof}
\begin{align*}
    &\Ebb\left[\xi_{y_m}(\nu_1)\xi_{y_m}(\nu_2)^*\right] \\
    &\hskip2cm=\frac{1}{N}\sum_{n_1,n_2=1}^N\Ebb[y_{m,n_1}y_{m,n_2}^*]e^{-2i\pi(n_1-1)\nu_1}e^{2i\pi(n_2-1)\nu_2}\\
    &\hskip2cm=\frac{1}{N}\sum_{n_1,n_2=1}^N r_{m,n_1-n_2}e^{-2i\pi(n_1-1)\nu_1+2i\pi(n_2-1)\nu_2}\\
    &\hskip2cm=\frac{1}{N}\sum_{u=-(N-1), n_1,n_2\in{0,\ldots,N-1}}^{(N-1)} r_{m,u}\sum_{n_1-n_2=u}e^{-2i\pi n_1\nu_1+2i\pi n_2\nu_2}
\end{align*}

Splitting this expression for $u=0, u>0$ and $u<0$ provides

\begin{multline}
\label{eq:split-cov-u}
    \Ebb[\xi_{y_m}(\nu_1)\xi_{y_m}(\nu_2)^*]=\frac{1}{N} r_{m,0}\sum_{n_1=0}^{N-1}e^{-2i\pi n_1(\nu_2-\nu_1)} \\ +\frac{1}{N}\sum_{u=1}^{(N-1)} r_{m,u}\sum_{n_2=0}^{N-1-u}e^{-2i\pi(u+n_2)\nu_1}e^{2i\pi n_2\nu_2}\\
    +\frac{1}{N}\sum_{u=-(N-1)}^{-1} r_{m,u}\sum_{n_2=-u}^{N-1}e^{-2i\pi(u+n_2)\nu_1}e^{2i\pi n_2\nu_2}.
\end{multline}

The first term of the right hand side of \eqref{eq:split-cov-u} can be computed in the case $\nu_1=\nu_2$:
$$ \frac{1}{N} r_{m,0}\sum_{n_1=0}^{N-1}e^{-2i\pi n_1(\nu_2-\nu_1)} = r_{m,0} $$
and in the case $\nu_1\neq\nu_2$, 
$$ \frac{1}{N}r_{m,0}\sum_{n_1=0}^{N-1}e^{-2i\pi n_1\frac{k}{N}}=0. $$
Therefore, the first term of the right hand side of \eqref{eq:split-cov-u} is equal 
to $r_{m,0}\delta_{\nu_1=\nu_2}$.

Consider now the second term of \eqref{eq:split-cov-u} (where $u>0$):
\begin{align}
\label{equation:esperance_covariance_xi_u_negatif}
\notag
    &\frac{1}{N}\sum_{u=1}^{N-1}r_{m,u}\sum_{n_2=0}^{N-1-u}e^{-2i\pi(u+n_2)\nu_1}e^{2i\pi n_2\nu_2} \\
    &\hskip3cm= \frac{1}{N}\sum_{u=1}^{N-1}r_{m,u}e^{-2i\pi u\nu_1}\sum_{n_2=0}^{N-1-u}e^{-2i\pi n_2(\nu_2-\nu_1)}.
\end{align}

The right hand side of \eqref{equation:esperance_covariance_xi_u_negatif} can also be explicitly written in the case $\nu_1=\nu_2$ :
\begin{align*}
    &\frac{1}{N}\sum_{u=1}^{N-1}r_{m,u}e^{-2i\pi u\nu_1}\sum_{n_2=0}^{N-1-u}e^{-2i\pi n_2(\nu_2-\nu_1)} \\
    &\hskip5cm= \frac{1}{N}\sum_{u=1}^{N-1}r_{m,u}e^{-2i\pi u\nu_1}(N-u) \\
    &\hskip5cm= \sum_{u=1}^{N-1}r_{m,u}e^{-2i\pi u\nu_1}\frac{N-u}{N} \\ 
    &\hskip5cm= \sum_{u=1}^{N-1}r_{m,u}e^{-2i\pi u\nu_1} -\frac{1}{N} \sum_{u=1}^{N-1}u\, r_{m,u}e^{2i\pi u\nu_1}.
\end{align*}
By Assumption \ref{assumption:regularity}, $\sup_{m\ge1}\sum_{u\in\Zbb}|u||r_{m,u}|<+\infty$, so we have:
$$\sup_{m\ge1}\frac{1}{N} \left| \sum_{u=1}^{N-1}u\, r_{m,u}e^{2i\pi u\nu_1} \right|= \Ocal\left(\frac{1}{N}\right).$$
Therefore:
\begin{align}
    \label{eq:u_positif_1}
    \sup_{m\ge1}\left|\frac{1}{N}\sum_{u=1}^{N-1}r_{m,u}e^{-2i\pi u\nu_1}\sum_{n_2=0}^{N-1-u}e^{-2i\pi n_2 (\nu_2-\nu_1)}-\sum_{u=1}^{N-1}r_{m,u}e^{-2i\pi u\nu_1}\right|=\Ocal\left(\frac{1}{N}\right).
\end{align}

In the case where $\nu_1\neq\nu_2$, note that $\nu_1-\nu_2=k/N$ with $k\neq0$, therefore:
\begin{equation}
\label{eq:fourier_transform_0}
    \sum_{n_2=0}^{N-1}e^{-2i\pi n_2(\nu_2-\nu_1)} = \sum_{n_2=0}^{N-1}e^{-2i\pi n_2\frac{k}{N}} = 0.
\end{equation}
Using \eqref{eq:fourier_transform_0}, one can rewrite the right hand side of \eqref{equation:esperance_covariance_xi_u_negatif} as
\begin{align*}
    &\left|\frac{1}{N}\sum_{u=1}^{N-1}r_{m,u}e^{-2i\pi u\nu_1}\sum_{n_2=0}^{N-1-u}e^{-2i\pi n_2(\nu_2-\nu_1)}\right| \\
    &\hskip4cm= \left|-\frac{1}{N}\sum_{u=1}^{N-1}r_{m,u}e^{-2i\pi u\nu_1}\sum_{n_2=N-u}^{N}e^{-2i\pi n_2(\nu_2-\nu_1)}\right| \\
    &\hskip4cm\le \frac{1}{N}\sum_{u=1}^{N-1}|u||r_{m,u}|
\end{align*}
which, again by Assumption \ref{assumption:regularity}, provides the bound:
\begin{align}
\label{eq:u_positif_2}
    \sup_{m\ge1}\left|\frac{1}{N}\sum_{u=1}^{N-1}r_{m,u}e^{-2i\pi u\nu_1}\sum_{n_2=0}^{N-1-u}e^{-2i\pi n_2(\nu_2-\nu_1)}\right| = \Ocal\left(\frac{1}{N}\right).
\end{align}

Combining \eqref{eq:u_positif_1} and \eqref{eq:u_positif_2}, the second term of the right hand side of \eqref{eq:split-cov-u} can be estimated as follow: 
\begin{multline*}
    \sup_{m\ge1}\left|\frac{1}{N}\sum_{u=1}^{(N-1)}r_{m,u}e^{-2i\pi u\nu_1}\sum_{n_2=0}^{N-1-u}e^{-2i\pi n_2(\nu_2-\nu_1)}-\delta_{\nu_1=\nu_2}\sum_{u=1}^{N-1}r_{m,u}e^{-2i\pi u\nu_1}\right|\\
    =\Ocal\left(\frac{1}{N}\right).
\end{multline*}

The term for $u<0$ in equation \eqref{eq:split-cov-u} is similar. Gathering the three terms of equation \eqref{eq:split-cov-u} leads to 
\begin{equation}
    \label{eq:reste-brillinger}
    \sup_{m\ge1}\left|\Ebb[\xi_{y_m}(\nu_1)\xi_{y_m}(\nu_2)^*] - \delta_{\nu_1=\nu_2}\left(\sum_{u=-(N-1)}^{N-1}r_{m,u}e^{-2i\pi u\nu_1}\right)\right| = \Ocal\left(\frac{1}{N}\right).
\end{equation}

Finally, using again Assumption \ref{assumption:regularity} we have:
$$\left| \sum_{|u|>N}r_m(u)e^{-2i\pi u \nu_1} \right| \le \frac{1}{N}\sum_{|u|>N}|u||r_m(u)|=\Ocal\left(\frac{1}{N}\right).$$

Inserting this into equation \eqref{eq:reste-brillinger}, we obtain equation \eqref{eq:uniformite-brillinger-covariance}
\end{proof}

\subsection{Proof of Lemma \ref{lemma:localization_s_m}}
\label{appendix:concentration_sm}

\begin{proof}
Consider the complement of the event $\Lambda_\epsilon^{\hat{\D}}(\nu)$ and notice that:
\begin{equation}
\label{equation:decomposition_Lambda_D_hat}
    \Lambda_\epsilon^{\hat{\D}}(\nu)^c \subset \{\exists m\in\{1,\ldots,M\}: \hat{s}_m>\bar{s}+\epsilon\} \cup \{\exists m\in\{1,\ldots,M\}: \hat{s}_m<\barbelow{s}-\epsilon\}.
\end{equation}

We start by proving that the first set of the right hand side of  \eqref{equation:decomposition_Lambda_D_hat} holds with exponentially small probability, ie. for any $\epsilon>0$, there exists $\gamma>0$ such that:
$$\Prob\left[\exists m\in\{1,\ldots,M\}: \hat{s}_m>\bar{s}+\epsilon\right]\le\exp-N^{\gamma}.$$

By Lemma \ref{lemma:biais_uniformity} (see below), $|\Ebb\hat{s}_m-s_m|=\Ocal(B^2/N^2)$ so for $N$ large enough, this bias term will be smaller than $\epsilon/2$. Moreover, for any $m\in\{1,\ldots,M\}$, $s_m-\bar{s}\leq0$. Therefore, one can write for large enough $N$:
\begin{align*}
    &\Prob\left[\exists m\in\{1,\ldots,M\}: \hat{s}_m>\bar{s}+\epsilon\right] \\
    &\hskip3cm= \Prob\left[\sup_{m\in\{1,\ldots,M\}}(\hat{s}_m-\Ebb\hat{s}_m+\Ebb\hat{s}_m-s_m+s_m-\bar{s})>\epsilon\right] \\
    &\hskip3cm\le\Prob\left[\sup_{m\in\{1,\ldots,M\}}|\hat{s}_m-\Ebb\hat{s}_m|>\epsilon/2\right]
\end{align*}
which holds with exponentially high probability by Lemma \ref{lemma:concentration_shat_esperance} (see below). The proof for the lower bound is similar.
\end{proof}

It remains to prove Lemma \ref{lemma:biais_uniformity} and Lemma \ref{lemma:concentration_shat_esperance}. Concerning the proof of Lemma \ref{lemma:biais_uniformity}, we follow the same approach as the one used in Theorem 5.4.2 in \cite{brillinger1981time}.

\begin{lemma}
\label{lemma:biais_uniformity}
For any $\nu\in[0, 1]$, the following results hold:
\begin{equation}
    \label{eq:expre-biais-hatsm}
    \mathbb{E}(\hat{s}_m(\nu)) - s_m(\nu) = \frac{s_m''(\nu)}{2} \, v_N + \Ocal\left( \left( \frac{B}{N} \right)^3 + \frac{1}{N} \right)
\end{equation}
and 
\begin{equation}
\label{equation:biais_uniformity}
    \sup_{m=1, \ldots, M}|\Ebb\hat{s}_m(\nu) - s_m(\nu)| = \Ocal\left(\frac{B}{N}\right)^2.
\end{equation}
\end{lemma}

\begin{proof}
It is clear that $\hat{s}_m(\nu) = \hat{\S}_{m,m}(\nu) = s_m(\nu) \tilde{\C}_{m,m}(\nu)$ can be written as 
\begin{equation}
    \label{eq:expre-hatsm}
  \hat{s}_m(\nu) = s_m(\nu) \frac{x_m (\I + \Phibs_m) x_m^*}{B+1}  .
\end{equation}
Therefore, $\mathbb{E}(\hat{s}_m(\nu)) = s_m(\nu) (1 + \frac{1}{B+1} \Tr \Phibs_m)$. (\ref{eq:expre-biais-hatsm}) thus follows immediately from (\ref{eq:expre-Trace-Phim}). (\ref{equation:biais_uniformity}) is an immediate consequence of (\ref{eq:expre-biais-hatsm}). 
\end{proof}

\begin{lemma}
\label{lemma:concentration_shat_esperance}
The family of random variables $\sup_{m=1, \ldots, M}|\hat{s}_m(\nu)-\Ebb[\hat{s}_m(\nu)]|, \nu \in [0,1]$ satisfies
\begin{equation}
\label{equation:concentration_shat_esperance}
    \sup_{m=1, \ldots, M}|\hat{s}_m-\Ebb[\hat{s}_m]| \prec \frac{1}{\sqrt{B}}.
\end{equation}
\end{lemma}

\begin{proof}
(\ref{eq:expre-hatsm}) implies that $\hat{s}_m-\Ebb[\hat{s}_m]$ can be written as
$\hat{s}_m-\Ebb[\hat{s}_m] = s_m \left( \frac{x_m(\I + \Phibs_m)x_m^*}{B+1} - \frac{1}{B+1} \mathrm{Tr}(\I + \Phibs_m)\right)$. It is clear that $\sup_{m} \| \frac{(\I + \Phibs_m)}{B+1} \|_F \leq \frac{C}{B}$ 
for some nice constant $C$. Therefore, (\ref{equation:concentration_shat_esperance}) leads 
immediately to (\ref{equation:hanson_wright_stochastic_domination}) . 
\end{proof}

\subsection{Proof of Lemma \ref{lemma:concentration_ratio_s_shat}}
\label{appendix:various}

\begin{proof}
These estimates can be proved in a compact way by using the calculus rules available in the stochastic domination framework introduced in Definition \ref{definition:stochastic_domination} and proved in Lemma \ref{lemma:algebra_domination}. Using Lemma \ref{lemma:concentration_sm_inf_sup} and Lemma \ref{lemma:hat_s_s_domination_stochastique} (see below): 
\begin{align*}
    \left| \frac{1}{\sqrt{\hat{s}_m}} - \frac{1}{\sqrt{s_m}} \right| &= \left| \frac{\sqrt{s_m}-\sqrt{\hat{s}_m}}{\sqrt{s_m}\sqrt{\hat{s}_m}}\right| \\
    &\le \underbrace{\left|\sqrt{s_m}-\sqrt{\hat{s}_m}\right|}_{\Ocal_\prec(\frac{1}{\sqrt{B}}+\frac{B^2}{N^2})} \times \underbrace{\left| \sqrt{\frac{1}{s_m}}\right|}_{\Ocal_\prec(1)} \times  \underbrace{\left| \sqrt{\frac{1}{\hat{s}_m}}\right|}_{\Ocal_\prec(1)} \\
    &\prec \frac{1}{\sqrt{B}} + \frac{B^2}{N^2}.
\end{align*}

The second inequality is similar to prove:
\begin{align*}
    \left|\sqrt{\frac{s_m}{\hat{s}_m}}-1\right| &= \left|\frac{\sqrt{s_m}-\sqrt{\hat{s}_m}}{\sqrt{\hat{s}_m}}\right| \\
    &\le  \underbrace{\left|s_m-\hat{s}_m\right|}_{\Ocal_\prec(\frac{1}{\sqrt{B}}+\frac{B^2}{N^2})} \times \underbrace{\left|\frac{1}{\hat{s}_m(\sqrt{s_m+\hat{s}_m})}\right|}_{\Ocal_\prec(1)} \\
    &\prec \frac{1}{\sqrt{B}} + \frac{B^2}{N^2}.
\end{align*}

\end{proof}

\begin{lemma}
\label{lemma:hat_s_s_domination_stochastique}
The family of random variables $(\sup_{m=1, \ldots, M} |\hat{s}_m(\nu)-s_m(\nu)|)$, $\nu\in[0, 1]$ satisfies
$$ \sup_{m=1, \ldots, M} |\hat{s}_m-s_m|\prec \frac{1}{\sqrt{B}}+\frac{B^2}{N^2}. $$
\end{lemma}

\begin{proof}
It is sufficient to check that the family of random variables $(|\hat{s}_m-s_m|)_{m=1, \ldots, M}, \nu \in [0,1]$ satisfies $|\hat{s}_m-s_m| \prec \frac{1}{\sqrt{B}}+\frac{B^2}{N^2}$. Using Lemma \ref{lemma:biais_uniformity} and Lemma \ref{lemma:concentration_shat_esperance}, 
we obtain as expected that
$$|\hat{s}_m-s_m|=|s_m-\Ebb\hat{s}_m+\Ebb\hat{s}_m-\hat{s}_m|  \le \underbrace{|s_m-\Ebb\hat{s}_m|}_{\Ocal(\frac{B^2}{N^2})} + \underbrace{|\Ebb\hat{s}_m-\hat{s}_m|}_{\Ocal_\prec(\frac{1}{\sqrt{B}})} \prec \frac{1}{\sqrt{B}}+\frac{B^2}{N^2}.$$
\end{proof}

\subsection{Proof of Lemma \ref{lemma:concentration_somme_carre_shat}}

\begin{lemma}
\label{lemma:concentration_somme_carre_shat}
The set of random variable $(\sum_{m=1}^M|\hat{s}_m(\nu)-s_m(\nu)|^2)$, $\nu\in[0, 1]$ satisfies
$$ \sum_{m=1}^M|\hat{s}_m-s_m|^2 \prec 1+\frac{B^5}{N^4} .$$
\end{lemma}

\begin{proof}
Using Lemma \ref{lemma:hat_s_s_domination_stochastique}, we have
$$ |\hat{s}_m-s_m|^2 \prec \frac{1}{B} + \frac{B^4}{N^4} $$
and summing over $m=1\ldots M$, one immediately get:
$$ \sum_{m=1}^M|\hat{s}_m-s_m|^2 \prec 1+\frac{B^5}{N^4}. $$
\

\end{proof}

\subsection{Proof of Lemma \ref{le:EnormtildeDelta}}
We express $\tilde{\Deltabs}$ as $\tilde{\Deltabs} = \frac{\X \Gammabs^*}{B+1}
 + \frac{\Gammabs \X^*}{B+1} + \frac{\Gammabs \Gammabs^*}{B+1}$. Therefore, we have
 $$
\| \tilde{\Deltabs} \|^{k} \leq C \left( \left\| \frac{\X}{\sqrt{B+1}} \right\|^{k} \left\|\frac{\Gammabs}{\sqrt{B+1}} \right\|^{k} + 
 \left\|\frac{\Gammabs \Gammabs^*}{B+1} \right\|^{k} \right).
$$
Using the Schwartz inequality, we obtain that 
$$
\mathbb{E}\| \tilde{\Deltabs} \|^{k} \leq  C \left( \left(\mathbb{E} \left\| \frac{\X \X^*}{B+1} \right\|^{k} \right)^{1/2} \left( \mathbb{E} \left\|\frac{\Gammabs \Gammabs^*}{B+1} \right\|^{k} \right)^{1/2} + \mathbb{E}  \left\|\frac{\Gammabs \Gammabs^*}{B+1} \right\|^{k} \right) .
$$
It is well-known that $\mathbb{E}\left( \left\| \frac{\X \X^*}{B+1} \right\|^{k} \right) \leq C$  for some nice constant depending on $k$.  Therefore, we establish that 
$$
\mathbb{E}  \left\|\frac{\Gammabs \Gammabs^*}{B+1} \right\|^{k}  \leq  C \, \left(\frac{B}{N}\right)^{2k}
$$
a property which will imply that 
$\mathbb{E}\| \tilde{\Deltabs} \|^{k} \leq  C \left( \frac{B}{N}\right)^{k}$. 
For this, we put $\Z =  \frac{\Gammabs \Gammabs^*}{B+1}$. As (\ref{eq:esperance-Z}) holds, it remains to verify that $\mathbb{E}( \| \Z - \mathbb{E}(\Z)\|^{k}) = \Ocal\left(\frac{B}{N}\right)^{2k}$. For this, we use the concentration inequality  (\ref{eq:concentration-normZ-EZ}). We choose $t_N = w^{1/k} (\frac{B}{N})^{2}$, and obtain that 
\begin{equation}
    \label{eq:newconcentration-Z-EZ}
\Prob \left[ \frac{\|\Z - \mathbb{E}(\Z) \|}{(\frac{B}{N})^{2}} >  w^{1/k} \right] \leq 
 2 \, C_0 \exp-CB(w^{1/k} - w_0^{1/k})
\end{equation}
for some $w_0 > 0$. If we denote by $z_N$ the random variable $z_N =  \left( \frac{\|\Z - \mathbb{E}(\Z) \|}{(\frac{B}{N})^{2}}\right)^{k}$, we have to establish that $\mathbb{E}(z_N) = \Ocal(1)$. For this, we express 
$\mathbb{E}(z_N)$ as
$$
\mathbb{E}(z_N) = \int_{0}^{+\infty} \Prob(z_N > w) \, dw = \int_{0}^{w_0} \Prob(z_N > w) dw 
+ \int_{w_0}^{+\infty}  \Prob(z_N > w) \, dw.
$$
As $\Prob(z_N > w) = \Prob(z_N^{1/k}> w^{1/k})$, (\ref{eq:newconcentration-Z-EZ}) immediately implies that $\mathbb{E}(z_N) = \Ocal(1)$. 
\subsection{Proof of Lemma \ref{le:EwAWQ}}
We denote by $\eta_m$ the term of interest, i.e. $\eta_m = \mathbb{E}(\w_m \A \W^* \Q \e_m)$. It can be written as 
$$
\eta_m = \sum_{n_1} \left( \sum_{n_2,m'} \mathbb{E}(\W_{m,n_2} \overline{\W}_{m',n_1} \Q_{m',m}) \A_{n_2,n_1} \right).
$$
The integration by parts formula (\ref{eq:integration-by-part}) leads to 
\begin{multline*}
    \mathbb{E}(\W_{m,n_2} \overline{\W}_{m',n_1} \Q_{m',m}) = \\ \delta_{m-m'} \delta_{n_1-n_2} \frac{1}{B+1} \mathbb{E}(\Q_{m,m}) - \frac{1}{B+1} \mathbb{E} \left[ \overline{\W}_{m',n_1} (\Q \W)_{m',n_2} \Q_{m,m} \right].
\end{multline*}
Therefore, we obtain that 
\begin{multline*}
    \sum_{n_2,m'} \mathbb{E}(\W_{m,n_2} \overline{\W}_{m',n_1} \Q_{m',m}) \A_{n_2,n_1} \\ = \frac{1}{B+1} \mathbb{E}(\Q_{m,m}) \A_{n_1,n_1} - \frac{1}{B+1} \mathbb{E} \left[ (\W^* \Q \W \A)_{n_1,n_1} \Q_{m,m} \right]
\end{multline*}
and that 
\begin{align}
\nonumber
\eta_m & =  \Ebb[\Q_{m,m}] \frac{1}{B+1} \Tr \A - \Ebb\left[ \left( \frac{1}{B+1} \Tr \W^* \Q \W \A \right) \Q_{m,m} \right] \\
\nonumber
    &  =  \beta \frac{1}{B+1} \Tr \A -  \beta \, \Ebb \left( \frac{1}{B+1} \Tr \W^* \Q \W \A \right) \\ 
    &\hskip3cm- \Ebb\left[ \left( \frac{1}{B+1} \Tr \W^* \Q \W \A \right)^{\circ} \Q^{\circ}_{m,m} \right].
\label{eq:expre-1-etam}
\end{align}
In order to evaluate  $\Ebb \left( \frac{1}{B+1} \Tr \W^* \Q \W \A \right)$, we note that 
\begin{multline*}
\Ebb \left( \frac{1}{M} \Tr \W^* \Q \W \A \right) = \frac{1}{M} \sum_{m=1}^{M} \eta_m \\ = \beta \frac{1}{B+1} \Tr \A -  \beta \, c \, \Ebb \left( \frac{1}{M} \Tr \W^* \Q \W \A \right) -  \\ \Ebb\left[ \left( \frac{1}{B+1} \Tr \W^* \Q \W \A \right)^{\circ} \frac{1}{M} \Tr\Q^{\circ} \right]
\end{multline*}
from which we deduce that 
\begin{multline*}
\Ebb \left( \frac{1}{M} \Tr \W^* \Q \W \A \right) = \frac{\beta}{1+\beta c} \, \frac{1}{B+1} \Tr \A \\ - \frac{1}{1+\beta c} \, \Ebb\left[ \left( \frac{1}{B+1} \Tr \W^* \Q \W \A \right)^{\circ} \frac{1}{M} \Tr\Q^{\circ} \right].
\end{multline*}
Plugging this relation into (\ref{eq:expre-1-etam}) leads immediately to (\ref{eq:expre-EwAW*Q}).

\bibliographystyle{imsart-nameyear.bst} 
\bibliography{references.bib}       




\end{document}